\documentclass[11pt]{article}
\usepackage[utf8]{inputenc}
\usepackage{amsmath,amssymb,amsthm,amsfonts}
\usepackage{mathtools}
\usepackage{bm}
\usepackage{enumitem}
\usepackage{booktabs}
\usepackage{fullpage}
\usepackage[hidelinks]{hyperref}
\theoremstyle{plain}
\newtheorem{theorem}{Theorem}
\newtheorem{proposition}[theorem]{Proposition}
\newtheorem{lemma}[theorem]{Lemma}

\theoremstyle{definition}
\newtheorem{definition}[theorem]{Definition}
\theoremstyle{remark}
\newtheorem{remark}[theorem]{Remark}
\numberwithin{equation}{section}
\newcommand{\Z}{\mathbb{Z}}
\newcommand{\R}{\mathbb{R}}
\newcommand{\C}{\mathbb{C}}
\newcommand{\E}{\mathbb{E}}
\newcommand{\PP}{\mathbb{P}}
\newcommand{\one}{\mathbf{1}}

\newcommand{\diag}{\operatorname{diag}}
\newcommand{\tr}{\operatorname{tr}}

\newcommand{\dist}{\operatorname{dist}}
\newcommand{\Lat}{\mathcal{L}}
\newcommand{\Sym}{\mathrm{Sym}_n}
\newcommand{\NP}{\mathsf{NP}}
\newcommand{\coNP}{\mathsf{coNP}}
\newcommand{\cfar}{c_{\mathrm{far}}}
\newcommand{\cAR}{c_{\mathrm{AR}}}
\newcommand{\norm}[1]{\lVert #1\rVert}

\title{Improving the Constant in the Aharonov--Regev Theorem}
\author{Shuhong Gao\thanks{AI Use Disclosure: The main idea and framework are
due to the author, and Claude Fable 5 assisted in the research and the writing of the paper.}\\[4pt]
\normalsize School of Mathematical and Statistical Sciences\\
Clemson University, Clemson, SC 29634-0975\\
Email: \texttt{sgao@clemson.edu}}
\date{\today}

\begin{document}
\maketitle
\begin{abstract}
The Aharonov--Regev proof that $\mathrm{GapCVP}_{c\sqrt n}$ lies in $\NP\cap\coNP$ uses a verifier that tests dual lattice vectors; its proof gives $c=100$. We study this verifier through a linear program in which a certificate is modeled as a sample from a probability distribution on the dual lattice. This viewpoint gives an exact analysis of the close-target tests: adding gradient and Hessian checks increases the certified close radius from $\sigma^{-1}/4$ to $\sigma^{-1}\sqrt{3/16}$, and higher derivatives cannot improve it within this framework. With Gaussian certificates for far targets, we obtain $\mathrm{GapCVP}_{c\sqrt n},\mathrm{GapSVP}_{c\sqrt n}\in\NP\cap\coNP$ for every $c>2/(\pi\sqrt3)\approx0.3676$. An appendix gives a finite-bit implementation of the verifier. We also show that the Gaussian certificate is not always optimal and solve the unbounded-support program exactly for targets of order two, where its value is determined by the spectral capacity of an odd dual coset.
\end{abstract}

\section{Introduction}\label{sec:intro}

Let $\Lat=B\Z^n$ be a full-rank lattice in $\R^n$, $t\in\R^n$ a target, and
\[\dist(t,\Lat)=\min_{v\in\Lat}\norm{t-v}.\]
For $\gamma\ge1$ the promise problem $\mathrm{GapCVP}_\gamma$ asks to decide whether $\dist(t,\Lat)\le r$ (YES) or $\dist(t,\Lat)>\gamma r$ (NO), given $(B,t,r)$ with rational entries; $\mathrm{GapSVP}_\gamma$ is the analogous problem for the length $\lambda_1(\Lat)$ of a shortest nonzero lattice vector. Aharonov and Regev~\cite{AR05} proved that for an absolute constant $c$,
\begin{equation}\label{eq:AR}
\mathrm{GapCVP}_{c\sqrt n}\in\NP\cap\coNP\qquad\text{and}\qquad\mathrm{GapSVP}_{c\sqrt n}\in\NP\cap\coNP .
\end{equation}
Their $\coNP$ certificate for a NO instance is a list $W=(w_1,\dots,w_N)$ of dual lattice vectors, and the verifier checks that the matrix $\sum_iw_iw_i^{\top}$ has bounded largest eigenvalue and that the trigonometric sum
\[f_W(t')=\frac1N\sum_{i=1}^N\cos(2\pi\langle w_i,t'\rangle)\]
at the (rescaled) target is below a threshold. The theorem is the basic obstacle to proving $\NP$-hardness of lattice problems at polynomial approximation factors: $\NP$-hardness of $\mathrm{GapCVP}_{c\sqrt n}$ under deterministic many-one reductions, or under Cook reductions that make only queries satisfying the promise, would imply $\NP=\coNP$.

Until recently the constant $c$ was of no practical importance, because the best $\NP$-hardness results for Euclidean $\mathrm{GapCVP}$ were at factor $n^{a/\log\log n}$~\cite{DKRS03}, far below any polynomial. That has changed: Song~\cite{Song26} gives a deterministic polynomial-time reduction from 3SAT establishing $\NP$-hardness of $\mathrm{GapCVP}_{n^{1/2-\epsilon}}$ for every fixed $\epsilon>0$. The gap between what is provably hard and what is in $\NP\cap\coNP$ is therefore now the interval $[n^{1/2-\epsilon},c\sqrt n]$: the exponents agree up to $\epsilon$, and what remains is a factor $n^\epsilon$ on the hardness side and the constant $c$ on the $\NP\cap\coNP$ side. This paper is about the constant. The paper~\cite{AR05} states~\eqref{eq:AR} with an unspecified $c$; reading its proof, one finds that it certifies $c=100$ (close instances are rescaled to distance $1/100$, far instances to distance $\sqrt n$), and to our knowledge no smaller value has been established since.

Our approach is to regard a certificate as a sample from a probability distribution on the dual lattice and the verifier, which sees only the empirical distribution of the sample, as a linear program over such distributions. Recall that a \emph{verifier} for the $\coNP$ side of $\mathrm{GapCVP}_\gamma$ is a deterministic polynomial-time algorithm that receives an instance $(B,t,r)$ together with a string of polynomial length, the \emph{certificate}, and accepts or rejects by checking a finite list of inequalities, its \emph{tests}. It must be \emph{sound}, rejecting every YES instance whatever the certificate, and \emph{complete}, accepting every NO instance for some certificate; the existence of such a verifier is what $\mathrm{GapCVP}_\gamma\in\coNP$ means. All verifiers in this paper first rescale the instance by a factor that makes the YES instances those with $\dist(t',\Lat')\le r_k$, where the \emph{close radius} $r_k$ depends only on $n$ and the parameters of the verifier (the index $k$ is the order of its derivative tests, Section~\ref{sec:verifier}), and then apply tests that do not involve $r$.

\paragraph{Certificates as samples from a probability distribution.}
In the verifier of Aharonov and Regev, and in all verifiers of this paper, a certificate is a list $W=(w_1,\dots,w_N)$ of vectors of the dual lattice $\Lat'^*$ of the rescaled lattice, given by their integer coefficients. Such a list is meant to be a sample from a probability distribution on $\Lat'^*$. Since $\Lat'^*$ is a countable set, a probability distribution on it is given by its masses: a function $P:\Lat'^*\to[0,1]$ with $\sum_{w\in\Lat'^*}P(w)=1$, which assigns to a set $A\subseteq\Lat'^*$ the probability $P(A)=\sum_{w\in A}P(w)$ and to a function $\varphi$ on $\Lat'^*$ the expectation $\E_P\varphi=\sum_{w}P(w)\varphi(w)$. The example that matters is the \emph{discrete Gaussian distribution}
\[D_{\Lat'^*}(w)=\frac{\rho(w)}{\rho(\Lat'^*)},\qquad\rho(w)=e^{-\pi\norm w^2},\quad\rho(\Lat'^*)=\sum_{u\in\Lat'^*}\rho(u).\]
To \emph{sample} $w_1,\dots,w_N$ independently from $P$ means that
\[\PP(w_1=u_1,\dots,w_N=u_N)=\prod_{i=1}^NP(u_i)\qquad\text{for all }u_1,\dots,u_N\in\Lat'^*.\]
The verifier of Aharonov and Regev is designed for a certificate that is such a sample: completeness is proved with $w_1,\dots,w_N$ drawn independently from $D_{\Lat'^*}$, and the tests are that every $w_i$ lies in $\Lat'^*$, the \emph{spectral test} $\frac1N\sum_iw_iw_i^{\top}\preceq\sigma^2I$, and the \emph{threshold test} $f_W(t')<\tau$. Throughout, $\sigma^2$ is the spectral bound; for Gaussian samples the expected value of $\frac1N\sum_iw_iw_i^{\top}$ is at most $\frac1{2\pi}I$ (Lemma~\ref{lem:transference}), so the natural choice is $\sigma^2=(1+o(1))/2\pi$, and every radius below scales as $\sigma^{-1}$.

Soundness, on the other hand, must hold for every list whatever, sampled from some distribution or not, and the verifier applies the same tests to every list. These tests never use the order of the list, and all of them are averages over it: they depend on $W$ only through its \emph{empirical distribution}
\[P_W(w):=\frac{\#\{i:\ w_i=w\}}N\qquad(w\in\Lat'^*),\]
the fraction of the entries of the list equal to $w$, repeated entries being counted with their multiplicity. This is a probability distribution on $\Lat'^*$ in the sense above, with finite support and rational masses, and for every function $\varphi$
\[\E_{P_W}\varphi=\sum_{w\in\Lat'^*}P_W(w)\varphi(w)=\frac1N\sum_{i=1}^N\varphi(w_i).\]
When the list is a sample from $P$, its empirical distribution is close to $P$: for every bounded $\varphi$ the average $\frac1N\sum_i\varphi(w_i)$ concentrates around $\E_P\varphi$, by Hoeffding's inequality, and likewise for matrix-valued $\varphi$ by the matrix Bernstein inequality (Lemmas~\ref{lem:hoeffding} and~\ref{lem:mb}). There are thus two distributions in play: a distribution $P$ from which a certificate may be sampled, and the empirical distribution $P_W$ of the certificate actually received. The verifier sees only $P_W$.

Each test of the verifier is therefore a linear constraint on $P_W$: dual membership is a constraint on the support, the spectral test says $\E_{P_W}ww^{\top}\preceq\sigma^2I$, a linear matrix inequality, and the threshold test bounds the linear functional $\E_{P_W}\cos(2\pi\langle w,t'\rangle)$. We call a verifier all of whose tests are of this kind a \emph{linear verifier} (Definition~\ref{def:linearverifier}); the Aharonov--Regev verifier and all the verifiers of this paper are linear, while a verifier that, say, runs a lattice algorithm on $W$ is not. A linear verifier is described by its \emph{feasible set} $\mathcal F(\Lat',t')$, the convex set of probability distributions on $\Lat'^*$ satisfying its constraints other than the threshold test, and by the \emph{value}
\[\Phi(\Lat',t'):=\inf_{P\in\mathcal F(\Lat',t')}\E_P\cos(2\pi\langle w,t'\rangle)\]
of the resulting linear program. The empirical distribution of an accepted certificate is feasible with objective below $\tau$; conversely, a feasible distribution supported on vectors of polynomial norm with objective below $\tau$ yields an accepted certificate by sampling, at the cost of the slack of concentration (Theorem~\ref{thm:constant}). So, up to that slack, \emph{the verifier accepts the pair $(\Lat',t')$ for some certificate if and only if $\Phi(\Lat',t')<\tau$.} A \emph{prover} is any rule that produces a feasible distribution for the pairs it must certify, efficient or not, since $\coNP$ asks only for existence. The \emph{Gaussian prover}, which is the one of Aharonov and Regev, takes $P=D_{\Lat'^*}$; an \emph{optimal prover} takes distributions whose objective equals or approaches $\Phi(\Lat',t')$. (For the spectral and gradient tests the infimum is attained, since the bounded second moments make the feasible set compact for weak convergence; with the Hessian test it need not be, Theorem~\ref{thm:order2}.)

\paragraph{The close side and the far side.}
The targets $t'$ with $\dist(t',\Lat')\le r_k$ form the \emph{close side}: there the verifier must reject every certificate. Soundness is the statement
\[\Phi(\Lat',t')\ \ge\ \tau\qquad\text{whenever }\dist(t',\Lat')\le r_k,\]
and it is proved by exhibiting, for each such pair, a dual solution of the program with value at least $\tau$: on the close side the dual solution is an \emph{envelope inequality} in the phases of the certificate (below). The targets with $\dist(t',\Lat')>\gamma r_k$ form the \emph{far side}: there some certificate must be accepted. Completeness at factor $\gamma$ is the statement
\[\Phi(\Lat',t')\ <\ \tau\qquad\text{whenever }\dist(t',\Lat')>\gamma r_k,\]
and it is proved by exhibiting a feasible distribution with small objective; what a prover must overcome is the dual solution of the far side, a \emph{quadratic form} $M\succeq0$ with $\cos(2\pi\langle w,t'\rangle)+w^{\top}Mw\ge\lambda$ on $\Lat'^*$. The \emph{far threshold} of a prover (or of the verifier, for an optimal prover) is the least $c'$ such that, at the spectral level $\sigma^2=1/2\pi$, the prover certifies every target at distance greater than $c'\sqrt n$, for all large $n$. The constant $c$ of the verifier, the smallest with $\gamma=c\sqrt n$, is then in the limit of large $n$ the ratio of the far threshold to the close radius at the same spectral level (Theorem~\ref{thm:constant}): the two sides are the same program read at different targets. The duality theory of Section~\ref{sec:duality} applies to every linear verifier with a strictly feasible point (Remark~\ref{rem:linearverifiers}); the exact results of Section~\ref{sec:close} and the identification of the constant in Section~\ref{sec:far} are proved for the verifiers that test derivatives of $f_W$ at the target.

\paragraph{Duality and transference.}
Two dualities meet in this problem. The first is the duality between a lattice and its dual, the subject of \emph{transference theorems}, inequalities between primal and dual invariants such as Banaszczyk's bound $n/2$ on the product of the covering radius of $\Lat$ and $\lambda_1(\Lat^*)$~\cite{Banaszczyk93}. Banaszczyk proves them with the coset mass function
\[f_\Lat(x)=\frac{\rho(\Lat+x)}{\rho(\Lat)}=\E_{w\sim D_{\Lat^*}}\cos(2\pi\langle w,x\rangle),\]
which is at least $e^{-\pi\dist(x,\Lat)^2}$ near the lattice and at most $\beta(c)^n$ at distance $c\sqrt n$ when $c>1/\sqrt{2\pi}$ (Lemma~\ref{lem:facts}), and his transference theorems place $\mathrm{GapCVP}_n$ and $\mathrm{GapSVP}_n$ in $\coNP$; the theorem of Aharonov and Regev improves this factor $n$ to $c\sqrt n$. Their verifier is Banaszczyk's proof made checkable. The expectation over $D_{\Lat^*}$ becomes an average over $N$ samples; the far bound survives sampling, by Hoeffding's inequality, but the near bound rests on the positivity of the Gaussian weights and cannot be checked on a list, and the spectral test together with the Taylor bound $\cos\varphi\ge1-\varphi^2/2$ takes its place; the spectral level $1/2\pi$ is itself a transference fact (Lemma~\ref{lem:transference}).

The second is linear-programming duality, and it makes the connection exact. If $t'=y+e$ with $y\in\Lat'$, then $\cos(2\pi\langle w,t'\rangle)=\cos(2\pi\langle w,e\rangle)$ for $w\in\Lat'^*$, and the inequality $\cos2\pi\theta+4\pi\theta^2\ge\pi/4$, valid for all real $\theta$ (Lemma~\ref{lem:env01}), gives a solution of the dual program (Theorem~\ref{thm:prover-duality}),
\[M=4\pi\,ee^{\top},\qquad\lambda=\frac\pi4,\qquad\text{of value}\qquad\lambda-\sigma^2\tr M=\frac\pi4-4\pi\sigma^2\norm e^2,\]
which is at least $\tau$ exactly when $\norm e$ is at most the close radius. Soundness is thus the statement that \emph{every close target has a dual certificate of rank one, built from its offset}. The Taylor bound gives the certificate $M=2\pi^2ee^{\top}$, $\lambda=1$, which is weaker for small $\tau$ and with $\tau=\frac12$ and $\sigma^2=3$ yields the radius $1/(2\pi\sqrt3)$ behind the constant $100$; the derivative tests add dual multipliers $u\parallel e$ and $U\propto ee^{\top}$ (Remark~\ref{rem:dual-reading}). The dual program admits every $M\succeq0$, not only those of rank one, so it is a convex relaxation of closeness, and completeness at distance $c\sqrt n$ is the statement that \emph{no far target has a dual certificate}: a transference theorem from the phases $\langle w,t'\rangle$ of the dual vectors to the distance of the target. Banaszczyk's Gaussian argument proves it for every $c>1/\sqrt{2\pi}$, because integrating the dual inequality against $D_{\Lat'^*}$ is weak duality with the Gaussian prover as the feasible point. The constant of a verifier is the gap of this relaxation, measured in units of the close radius. At targets of order two the relaxation becomes the spectral capacity of the odd part of the dual lattice, and completeness there is a coset analogue of Banaszczyk's transference theorem in which $\lambda_1(\Lat^*)$ is replaced by this capacity (Section~\ref{sec:order2}).

\subsection{Results}\label{sec:results}

\paragraph{Phases and the phase program.}
On the close side the program reduces to one real variable. For a target $t'$ with closest lattice point $y$ and offset $e=t'-y$, the \emph{phases} of a certificate are the real numbers $\theta_i=\langle w_i,e\rangle$; since $\langle w_i,y\rangle\in\Z$, every test on the derivatives of $f_W$ at $t'$ is a statement about the empirical distribution of the phases: the spectral test implies $\frac1N\sum_i\theta_i^2\le\sigma^2\norm e^2$, the threshold test is $\frac1N\sum_i\cos2\pi\theta_i<\tau$, and the $m$-th derivative of $f_W$ at $t'$ in the direction $e$ is $(2\pi)^m$ times the average of $\theta_i^m\cos(2\pi\theta_i+m\pi/2)$ (Lemma~\ref{lem:phase}). The close side of a verifier that tests the spectral condition, the threshold and the first $k$ derivatives is therefore governed by the \emph{phase program}
\begin{equation}\label{eq:vk-intro}
v_k:=\inf\Big\{\E\theta^2:\ \E\cos2\pi\theta\le0,\ \ \E\big[\theta^m\cos(2\pi\theta+m\pi/2)\big]=0\ \ (1\le m\le k)\Big\},
\end{equation}
the infimum over all probability distributions of a real random variable $\theta$ with $\E|\theta|^{\max(2,k)}<\infty$: a certificate with $f_W(t')\le0$ forces $\sigma^2\norm e^2\ge v_k$, so the close radius is $\sigma^{-1}\sqrt{v_k}$ in the limit of vanishing tolerances. Our first result solves~\eqref{eq:vk-intro} for every $k$; in its statement, and throughout the paper, $\delta_a$ denotes the probability distribution with mass $1$ at the point $a$.

\begin{theorem}[The close side]\label{thm:close}
$v_0=\frac1{16}$, $v_1=\frac18$, and $v_k=\frac3{16}$ for every $k\ge2$. The infimum is attained for $k\le2$, by the distributions $\frac12\delta_{1/4}+\frac12\delta_{-1/4}$, $\frac12\delta_0+\frac12\delta_{1/2}$ and $\frac34\delta_{1/4}+\frac14\delta_{3/4}$ respectively, and is certified by the envelope inequalities, valid for every real $\theta$,
\begin{align}
\theta^2&\ \ge\ \frac1{16}-\frac1{4\pi}\cos2\pi\theta,\label{eq:env0}\\
\theta^2&\ \ge\ \frac18-\frac18\cos2\pi\theta-\frac1\pi\,\theta\sin2\pi\theta,\label{eq:env1}\\
\theta^2&\ \ge\ \frac3{16}-\frac5{8\pi}\cos2\pi\theta-\frac12\,\theta\sin2\pi\theta+\frac2\pi\,\theta^2\cos2\pi\theta.\label{eq:env2}
\end{align}
For $k\ge3$ the infimum is not attained, but for every $k$ and every $\epsilon>0$ there is a probability distribution with finite support in $\frac14\Z$ and rational masses which satisfies every constraint of~\eqref{eq:vk-intro} and has $\E\theta^2\le\frac3{16}+\epsilon$.
\end{theorem}

The inequality~\eqref{eq:env2} is tight at $\theta=\pm\frac14,\pm\frac34$, where it has double zeros (Lemma~\ref{lem:env2}). What makes the derivative tests different from every test of $f_W$ at a point is that the phases enter them as real numbers and not modulo $1$: a test of $f_W$ at a point sees $\theta$ only through $\cos2\pi j\theta$, while the gradient sees $\theta\sin2\pi\theta$ and the Hessian $\theta^2\cos2\pi\theta$. The values are realized by lattice certificates: for every $r>\sigma^{-1}\sqrt{v_k}$ and every dimension $n$ there are a lattice, a target at distance exactly $r$, and a certificate passing the spectral test with $f_W(t)=0$ and $\nabla^mf_W(t)=0$ for $1\le m\le k$ (Theorem~\ref{thm:exact}). So no test on the derivatives of $f_W$ at the target, of any order, raises the close radius above $\sigma^{-1}\sqrt{3/16}$.

The second result is the constant that follows for the \emph{Hessian-augmented Aharonov--Regev verifier}, the verifier of~\cite{AR05} with the spectral bound at its Gaussian expected value and two added tests, $\norm{\nabla f_W(t')}\le\eta_1$ and $\norm{\nabla^2f_W(t')}\le\eta_2$ (Section~\ref{sec:verifier}). The extra tests cost nothing for completeness, because for the Gaussian certificate, $N$ samples from the discrete Gaussian on the dual lattice, $f_W$ is close to a ratio of Gaussian masses whose gradient and Hessian are exponentially small at a far target.

\begin{theorem}[Main theorem]\label{thm:main}
For every $c>2/(\pi\sqrt3)\approx0.3676$,
\[\mathrm{GapCVP}_{c\sqrt n},\ \mathrm{GapSVP}_{c\sqrt n}\ \in\ \NP\cap\coNP .\]
The $\coNP$ verifier is the Hessian-augmented Aharonov--Regev verifier; it is sound in every dimension $n\ge8$ and complete in every dimension $n\ge n_0(c)$, an explicit cutoff (Table~\ref{tab:n0}). Its certificate consists of the integer coefficient vectors of $N=100n^2\lceil\log_2n\rceil$ dual lattice vectors, and it evaluates one trigonometric sum, its gradient and Hessian, and checks three positive-semidefiniteness conditions, all in exact rational arithmetic (Appendix~\ref{app:precision}).
\end{theorem}

The cutoff $n_0(c)$ is only a restriction of the proof: in a fixed dimension the closest vector problem is solvable exactly in polynomial time (Kannan~\cite{Kannan87}), so the verifier decides the dimensions below $n_0(c)$ itself. The number $2/(\pi\sqrt3)$ is the ratio of Banaszczyk's far threshold $1/\sqrt{2\pi}$ to the close radius $\sqrt{2\pi\cdot3/16}=\sqrt{3\pi/8}=1.0854$ of Theorem~\ref{thm:close} at the Gaussian spectral level $\sigma^2=1/2\pi$; the constants for the original tests and for the gradient test alone are in Table~\ref{tab:hierarchy}.

The third group of results is the duality theory, which is what makes the two sides one subject.

\begin{theorem}[Duality, informal]\label{thm:duality-intro}
\begin{enumerate}[label=\textup{(\roman*)}]
\item The phase program with tolerances $\tau>0$ on the threshold and $\eta'_m>0$ on the derivative constraints has the same value as its dual, the best lower bound $v-a\tau-\sum_m|b_m|\eta'_m$ over envelope inequalities $\theta^2\ge v-a\cos2\pi\theta-\sum_mb_m\theta^m\cos(2\pi\theta+m\pi/2)$ with $a\ge0$, and the dual is attained \textup{(Theorem~\ref{thm:phase-duality})}.
\item The prover's program at a pair $(\Lat,t)$ with spectral level $\sigma^2$ and tolerances $\eta_1,\eta_2>0$ on $\norm{\nabla f_W(t)}$ and $\norm{\nabla^2f_W(t)}$ has the same value as its dual
\[\max\Big\{\lambda-\sigma^2\tr M-\frac{\eta_1}{2\pi}\norm u-\frac{\eta_2}{4\pi^2}\norm U_*\Big\},\]
the maximum over $M\succeq0$, $u\in\R^n$, $U\in\Sym$ and $\lambda\in\R$ such that
$\cos(2\pi\langle w,t\rangle)+w^{\top}Mw+\langle u,w\rangle\sin(2\pi\langle w,t\rangle)+\cos(2\pi\langle w,t\rangle)\,w^{\top}Uw\ge\lambda$ for every $w\in\Lat^*$ \textup{(Theorem~\ref{thm:prover-duality})}.
\item The constant achievable by the family of Hessian-augmented verifiers with the best possible prover is $\cAR^{(2)}=\sqrt{8/(3\pi)}\;\cfar^{(2)}=0.9213\,\cfar^{(2)}$, where $\cfar^{(2)}$ is the far threshold of the prover's program with the gradient and Hessian constraints \textup{(Theorem~\ref{thm:constant})}.
\end{enumerate}
\end{theorem}

On the far side Banaszczyk's bound certifies the Gaussian prover beyond $1/\sqrt{2\pi}$, and we show that it is not always optimal: on Construction-A lattices of random binary codes, at a deep hole of the code at distance greater than $0.2735\sqrt n$, the coset Gaussian mass is at least half the lattice mass, so the Gaussian certificate fails, while a \emph{parity certificate} built from short dual vectors with odd pairing against the target passes every test (Section~\ref{sec:codelattices}). The exact threshold of the Gaussian prover remains unknown. Finally we solve the unbounded-support prover's program at targets of order two, $2t\in\Lat$, which include these deep holes. There $\cos(2\pi\langle w,t\rangle)=\chi(w)\in\{\pm1\}$ splits $\Lat^*$ into even and odd vectors, and the program is decided by one geometric quantity, the \emph{spectral capacity} of the odd coset,
\[\mu(\Lat,t):=\sup\Big\{\nu(O):\ \nu\ge0\ \text{a finitely supported measure on $O$ with}\ \sum_{w\in O}\nu(w)\,ww^{\top}\preceq I\Big\},\]
which by Lemma~\ref{lem:capacity} equals $\min\{\tr M:\ M\succeq0,\ w^{\top}Mw\ge1\ \text{for all }w\in O\}$.

\begin{theorem}[Targets of order two]\label{thm:order2-intro}
Let $2t\in\Lat$, $t\notin\Lat$. The value of the prover's program with the spectral constraint alone is $1-2\min(1,\sigma^2\mu(\Lat,t))$, and with the exact Hessian constraint added its infimum is $1-2\min(1,\sigma^2\mu(\Lat,t)/2)$. The latter value is approached by distributions on odd dual vectors, sufficiently large even multiples of them, and the zero vector; if the capacity is attained, so is the program value when $\sigma^2\mu<2$. At and above $\sigma^2\mu=2$, the infimum is $-1$ and is not attained. Combined with a random-lattice construction of targets of order two whose odd dual vectors are all long (Lemma~\ref{lem:randomorder2}), this gives $\cfar^{(2)}\ge1/(e\sqrt{2\pi})=0.1468$ and
\[0.1352\ \le\ \cAR^{(2)}\ \le\ \frac{2}{\pi\sqrt3}=0.3676 .\]
\end{theorem}

The Hessian test is not free on the far side: at a target of order two it forces the prover to balance the second moments of its even and odd vectors, which doubles the capacity it needs; for the gradient-only verifier the same obstruction gives only $\cAR^{(1)}\ge0.1171$. The parity prover survives the Hessian test after this balancing (Proposition~\ref{prop:parity}).

\subsection{Organization}

Section~\ref{sec:prelim} collects the facts about Gaussian sums and concentration that we use, with proofs. Section~\ref{sec:verifier} defines the verifiers with derivative tests, the prover's program and the phase program, and proves the elementary relations between them. Section~\ref{sec:duality} proves the two duality theorems. Section~\ref{sec:close} proves Theorem~\ref{thm:close}, including the proof of the envelope~\eqref{eq:env2}, and the realization theorem. Section~\ref{sec:complete} proves the completeness lemma for the Hessian test and assembles Theorem~\ref{thm:main}. Section~\ref{sec:far} treats the far side: the Gaussian prover, the identity $\cAR^{(2)}=\sqrt{8/(3\pi)}\,\cfar^{(2)}$, targets of order two, the code lattices on which the Gaussian prover fails, and the obstructions from random lattices. Appendix~\ref{app:precision} gives the finite-bit verifier.

\section{Preliminaries}\label{sec:prelim}

\subsection{Lattices, Gaussian sums and the coset mass function}

For a lattice $\Lat=B\Z^n$ its dual is $\Lat^*=\{w:\langle w,v\rangle\in\Z\ \forall v\in\Lat\}=B^{-\top}\Z^n$, so that $(\Lat^*)^*=\Lat$. The Fourier transform of an integrable $f:\R^n\to\C$ is
\[\hat f(w)=\int_{\R^n}f(x)\,e^{-2\pi i\langle x,w\rangle}\,dx .\]
For $s>0$ let
\[\rho_s(x)=e^{-\pi\norm x^2/s^2},\qquad\rho=\rho_1,\qquad\rho_s(A)=\sum_{a\in A}\rho_s(a)\ \text{ for a countable set }A,\]
so that $\hat\rho=\rho$ and $\widehat{\rho_s}=s^n\rho_{1/s}$. As in Section~\ref{sec:intro}, the discrete Gaussian probability distribution $D_\Lat$ on $\Lat$ assigns to $v\in\Lat$ the probability $\rho(v)/\rho(\Lat)$. The Poisson summation formula, for $f=\rho_s$,
\begin{equation}\label{eq:poisson}
\rho_s(\Lat+x)=\frac{s^n}{\det\Lat}\sum_{w\in\Lat^*}\rho_{1/s}(w)\cos(2\pi\langle w,x\rangle),
\end{equation}
converges absolutely together with all its termwise derivatives in $x$, locally uniformly. We use throughout the \emph{coset mass function}
\begin{equation}\label{eq:f}
f_\Lat(x):=\frac{\rho(\Lat+x)}{\rho(\Lat)}=\E_{w\sim D_{\Lat^*}}\big[\cos(2\pi\langle w,x\rangle)\big],
\end{equation}
the second expression being~\eqref{eq:poisson} at $s=1$ divided by the same identity at $x=0$. Differentiating~\eqref{eq:f} termwise,
\begin{equation}\label{eq:fderiv}
\begin{aligned}
\nabla f_\Lat(x)&=-\frac{2\pi}{\rho(\Lat)}\sum_{y\in\Lat+x}y\,e^{-\pi\norm y^2}=-2\pi\,\E\big[\sin(2\pi\langle w,x\rangle)\,w\big],\\
\nabla^2f_\Lat(x)&=\frac1{\rho(\Lat)}\sum_{y\in\Lat+x}\big(4\pi^2yy^{\top}-2\pi I\big)e^{-\pi\norm y^2}=-4\pi^2\,\E\big[\cos(2\pi\langle w,x\rangle)\,ww^{\top}\big],
\end{aligned}
\end{equation}
with $w\sim D_{\Lat^*}$ in the expectations. For a real symmetric matrix $A$ we write $s_{\max}(A)$ for its largest eigenvalue, $\norm A$ for its operator norm (the largest absolute value of an eigenvalue), $\norm A_*$ for its trace norm (the sum of the absolute values of the eigenvalues), and $A\preceq B$ when $B-A$ is positive semidefinite; $\Sym$ is the space of real symmetric $n\times n$ matrices with the inner product $\langle A,B\rangle=\tr(AB)$. Numerical constants are given to four decimal places; a decimal next to an exact expression is that expression rounded. We write $\mathcal B(x,r)$ for the open Euclidean ball.

\begin{lemma}\label{lem:facts}
Let $\Lat\subseteq\R^n$ be a lattice and $x,u\in\R^n$.
\begin{enumerate}[label=\textup{(\roman*)}]
\item $e^{-\pi\dist(x,\Lat)^2}\le f_\Lat(x)\le1$.
\item For $0<t\le1$: $\rho(t\Lat)\le t^{-n}\rho(\Lat)$ and $\rho(t(\Lat+x))\le\rho(t\Lat)$.
\item With $w\sim D_{\Lat^*}$: $\E e^{2\pi\langle w,u\rangle}=e^{\pi\norm u^2}\rho(\Lat^*-u)/\rho(\Lat^*)\le e^{\pi\norm u^2}$.
\item \textup{(Banaszczyk~\cite{Banaszczyk93})} For $c\ge1/\sqrt{2\pi}$,
\[\rho\big((\Lat+u)\setminus\mathcal B(0,c\sqrt n)\big)\le\beta(c)^n\rho(\Lat),\qquad\beta(c):=c\sqrt{2\pi e}\,e^{-\pi c^2},\]
where $\beta(1/\sqrt{2\pi})=1$, $\beta$ is strictly decreasing on $[1/\sqrt{2\pi},\infty)$, and $\beta(1)=0.1786$. In particular $f_\Lat(x)\le\beta(c)^n$ whenever $\dist(x,\Lat)\ge c\sqrt n$.
\end{enumerate}
\end{lemma}

\begin{proof}
(i) The upper bound is~\eqref{eq:f}. For the lower bound let $e=x-v_0$ with $v_0\in\Lat$ closest to $x$; pairing $v$ with $-v$,
\[\rho(\Lat+x)=\sum_{v\in\Lat}e^{-\pi\norm{v+e}^2}=e^{-\pi\norm e^2}\sum_{v\in\Lat}e^{-\pi\norm v^2}\cosh(2\pi\langle v,e\rangle)\ \ge\ e^{-\pi\norm e^2}\rho(\Lat).\]
(ii) Since $\rho(t\Lat)=\rho_{1/t}(\Lat)$, \eqref{eq:poisson} at $x=0$ with $s=1/t$ and with $s=1$ gives $\rho(t\Lat)/\rho(\Lat)=t^{-n}\rho_t(\Lat^*)/\rho(\Lat^*)\le t^{-n}$, as $\rho_t\le\rho$ termwise for $t\le1$; the second inequality is (i) for $t\Lat$. (iii) Completing the square, $\sum_{w\in\Lat^*}e^{-\pi\norm w^2+2\pi\langle w,u\rangle}=e^{\pi\norm u^2}\rho(\Lat^*-u)$, and $\rho(\Lat^*-u)\le\rho(\Lat^*)$ by (i).

(iv) For $t\in(0,1]$ and $\norm y\ge c\sqrt n$, $e^{-\pi\norm y^2}\le e^{-\pi(1-t^2)c^2n}e^{-\pi t^2\norm y^2}$; summing over $y\in(\Lat+u)\setminus\mathcal B(0,c\sqrt n)$ and applying (ii) twice,
\[\rho\big((\Lat+u)\setminus\mathcal B(0,c\sqrt n)\big)\le e^{-\pi(1-t^2)c^2n}\rho\big(t(\Lat+u)\big)\le e^{-\pi(1-t^2)c^2n}t^{-n}\rho(\Lat),\]
which is minimized at $t^2=1/2\pi c^2\le1$ with value $\beta(c)^n$. The last claim holds because all of $\Lat+x$ lies outside $\mathcal B(0,c\sqrt n)$.
\end{proof}

The next identity is the one place in this section where the Gaussians on $\Lat$ and on $\Lat^*$ are compared directly. It is a transference identity: it relates the second moments of the two discrete Gaussians in every direction, and it is what fixes the spectral bound of the verifier at $1/2\pi$.

\begin{lemma}[Second-moment transference]\label{lem:transference}
Let $\Lat\subseteq\R^n$ be a lattice, $w$ from the discrete Gaussian probability distribution $D_{\Lat^*}$ and $v$ from $D_\Lat$. Then for every $e\in\R^n$,
\[\E\langle w,e\rangle^2+\E\langle v,e\rangle^2=\frac{\norm e^2}{2\pi}.\]
In particular $\E\langle w,e\rangle^2<\norm e^2/2\pi$ for every $e\ne0$, and
\[\E\,ww^{\top}\preceq\frac1{2\pi}I,\qquad\E\norm w^2\le\frac n{2\pi}.\]
\end{lemma}

\begin{proof}
Consider $\phi(s):=f_\Lat(se)$ for real $s$. By~\eqref{eq:f}, $\phi(s)=\E\cos(2\pi s\langle w,e\rangle)$. On the other hand, expanding $\norm{v+se}^2$,
\[\phi(s)=\frac{\rho(\Lat+se)}{\rho(\Lat)}=e^{-\pi s^2\norm e^2}\,\E\big[e^{-2\pi s\langle v,e\rangle}\big].\]
Both expressions may be differentiated twice under the expectation at $s=0$: the first because $\E\langle w,e\rangle^2<\infty$ by Lemma~\ref{lem:facts}(iii), the second because $\E e^{\theta|\langle v,e\rangle|}<\infty$ for every $\theta$ by Lemma~\ref{lem:facts}(iii) applied to $\Lat^*$ in place of $\Lat$. Comparing the second derivatives at $s=0$ gives
\[-4\pi^2\E\langle w,e\rangle^2=-2\pi\norm e^2+4\pi^2\E\langle v,e\rangle^2,\]
which is the identity. If $e\ne0$, some $v\in\Lat$ has $\langle v,e\rangle\ne0$, since $\Lat$ has full rank, and $D_\Lat$ gives it positive probability, so $\E\langle v,e\rangle^2>0$ and the inequality is strict. The matrix inequality is the identity read as a statement about the quadratic form $e\mapsto e^{\top}(\E ww^{\top})e$, and taking $e$ in an orthonormal basis and summing gives $\E\norm w^2\le n/2\pi$.
\end{proof}

The inequality $\E\langle w,e\rangle^2\le\norm e^2/2\pi$ is known: it is the second-moment bound used in~\cite{AR05}, due to Banaszczyk~\cite{Banaszczyk93}. The identity shows exactly what is lost, namely the second moment of the primal Gaussian in the direction $e$; the loss is small precisely in the directions in which $\Lat$ is stretched, so the bound $1/2\pi$ cannot be improved in general.

We also record the estimate behind the gradient and Hessian bounds: for $s'=(1-1/n)^{-1/2}$ and every $y\in\R^n$,
\begin{equation}\label{eq:peel}
\norm y^2e^{-\pi\norm y^2}=\norm y^2e^{-\pi\norm y^2/n}\,\rho_{s'}(y)\le\frac{n}{\pi e}\,\rho_{s'}(y),\qquad
\norm y\,e^{-\pi\norm y^2}\le\sqrt{\frac{n}{2\pi e}}\,\rho_{s'}(y),
\end{equation}
since $a\mapsto ae^{-\pi a/n}$ has maximum $n/\pi e$ on $a\ge0$ and $a\mapsto ae^{-\pi a^2/n}$ has maximum $\sqrt{n/2\pi e}$. If $\dist(x,\Lat)\ge c''\sqrt n$ with $c''\sqrt{1-1/n}\ge1/\sqrt{2\pi}$, then every point of $(\Lat+x)/s'$ has norm at least $c''\sqrt{1-1/n}\,\sqrt n$, and Lemma~\ref{lem:facts}(iv) for the lattice $\Lat/s'$ followed by Lemma~\ref{lem:facts}(ii) gives
\begin{equation}\label{eq:rhosprime}
\frac{\rho_{s'}(\Lat+x)}{\rho(\Lat)}\ \le\ \big(s'\beta(c''/s')\big)^n\ =\ e^{\pi c''^2}\beta(c'')^n .
\end{equation}

\subsection{Concentration}

\begin{lemma}[Hoeffding~\cite{Hoeffding63}]\label{lem:hoeffding}
If $X_1,\dots,X_N$ are independent real random variables with values in $[-1,1]$, then for every $\zeta>0$,
$\PP\big(|\frac1N\sum_i(X_i-\E X_i)|\ge\zeta\big)\le2e^{-N\zeta^2/2}$.
\end{lemma}

\begin{lemma}[Matrix Bernstein; Tropp~\cite{Tropp12}]\label{lem:mb}
Let $X_1,\dots,X_N$ be independent random matrices in $\Sym$ with $\E X_i=0$ and $\norm{X_i}\le R$ almost surely, and let $\norm{\sum_i\E X_i^2}\le v$. Then for $u\ge0$,
\[\PP\Big(s_{\max}\Big(\sum_iX_i\Big)\ge u\Big)\le n\exp\Big(-\frac{u^2/2}{v+Ru/3}\Big),\qquad
\PP\Big(\Big\lVert\sum_iX_i\Big\rVert\ge u\Big)\le2n\exp\Big(-\frac{u^2/2}{v+Ru/3}\Big).\]
\end{lemma}

The second bound follows from the first applied to $\pm X_i$. Finally we need a truncation estimate for the discrete Gaussian.

\begin{lemma}[Truncation]\label{lem:trunc}
Let $n\ge2$, $w\sim D_\Lambda$ for a lattice $\Lambda\subseteq\R^n$, and $R\ge1$. Then
\[\E\big[\norm w^2\,\one_{\norm w>R\sqrt n}\big]\ \le\ \frac{n}{\pi e}\,e^{\pi R^2}\beta(R)^n,\qquad\text{in particular}\qquad\E\big[\norm w^2\,\one_{\norm w>\sqrt n}\big]\le\frac{e^{\pi-1}}{\pi}\,n\,\beta(1)^n .\]
Consequently, with $p:=\PP(\norm w>\sqrt n)\le\beta(1)^n$, if $\tilde w$ denotes $w$ conditioned on $\norm w\le\sqrt n$ and $g:\Lambda\to[-1,1]$ is any function, then for $n\ge8$,
\[\big\lVert\E[g(\tilde w)\tilde w\tilde w^{\top}]-\E[g(w)ww^{\top}]\big\rVert\le\kappa_n:=6n\beta(1)^n,\qquad\text{and}\qquad\E[\tilde w\tilde w^{\top}]\preceq\frac{1+\epsilon_n}{2\pi}I,\quad\epsilon_n:=2\beta(1)^n .\]
\end{lemma}

\begin{proof}
By~\eqref{eq:peel}, $\sum_{\norm w>R\sqrt n}\norm w^2\rho(w)\le\frac n{\pi e}\rho_{s'}\big(\Lambda\setminus\mathcal B(0,R\sqrt n)\big)$. The set $(\Lambda\setminus\mathcal B(0,R\sqrt n))/s'$ consists of points of the lattice $\Lambda/s'$ of norm at least $R\sqrt{1-1/n}\,\sqrt n$, and $R\sqrt{1-1/n}\ge1/\sqrt{2\pi}$ for $n\ge2$, so Lemma~\ref{lem:facts}(iv) and (ii) give
\[\rho_{s'}\big(\Lambda\setminus\mathcal B(0,R\sqrt n)\big)\le\beta\big(R\sqrt{1-1/n}\big)^n\rho(\Lambda/s')\le(1-1/n)^{n/2}e^{\pi R^2}\beta(R)^n\cdot(1-1/n)^{-n/2}\rho(\Lambda),\]
which equals $e^{\pi R^2}\beta(R)^n\rho(\Lambda)$,
using $\beta(R\sqrt{1-1/n})^n=(1-1/n)^{n/2}R^n(2\pi e)^{n/2}e^{-\pi R^2(n-1)}$. Dividing by $\rho(\Lambda)$ gives the first claim. For the second,
\[\E[g(\tilde w)\tilde w\tilde w^{\top}]-\E[g(w)ww^{\top}]=\frac{p\,\E[g(w)ww^{\top}]-\E[g(w)ww^{\top}\one_{\norm w>\sqrt n}]}{1-p},\]
of norm at most $(p\,n/2\pi+2.7097\,n\beta(1)^n)/(1-p)\le6n\beta(1)^n$ by Lemma~\ref{lem:transference} and $p\le\beta(1)^n\le1.1\cdot10^{-6}$. The last claim is $\E[\tilde w\tilde w^{\top}]\preceq\E[ww^{\top}]/(1-p)$ and $(1-p)^{-1}\le1+2p$.
\end{proof}

\section{Verifiers with derivative tests, and the prover's program}\label{sec:verifier}

\subsection{The verifiers}

Membership of $\mathrm{GapCVP}_\gamma$ in $\NP$ is trivial for every $\gamma\ge1$; the content of~\eqref{eq:AR} is the $\coNP$ part. We consider the following family of verifiers, indexed by $k\in\{0,1,2\}$ and parameters $\tau\in(0,\frac12]$, $\sigma^2>0$, $\eta_1,\eta_2\ge0$ (the restriction $\tau\le\frac12$ keeps the radii below well defined; we shall use $\tau=1/n$). Define the \emph{close radii}
\begin{equation}\label{eq:rk}
\begin{aligned}
r_0&:=\frac1\sigma\sqrt{\frac1{16}-\frac{\tau}{4\pi}},\\
r_1&:=\text{the positive root of }\ \sigma^2r^2+\frac{\eta_1}{2\pi^2}\,r=\frac{1-\tau}8,\\
r_2&:=\text{the positive root of }\ \Big(\sigma^2+\frac{\eta_2}{2\pi^3}\Big)r^2+\frac{\eta_1}{4\pi}\,r=\frac3{16}-\frac{5\tau}{8\pi},
\end{aligned}
\end{equation}
so that $r_0\to\frac1{4\sigma}$, $r_1\to\frac1{\sigma\sqrt8}$ and $r_2\to\frac1\sigma\sqrt{\frac3{16}}$ as $\tau,\eta_1,\eta_2\to0$. These radii are sound for every $\tau$ (Proposition~\ref{prop:sound}) but sharp only in this limit, which is the regime we use. For instance, without derivative tests the exact radius at threshold $\tau\in[0,1)$ is $\arccos(\tau)/(2\pi\sigma)$: the function $\phi(\tau)=(\arccos\tau/2\pi)^2$ is convex and decreasing on $[-1,1]$, and $\theta^2\ge\phi(\cos2\pi\theta)$, so $\E\theta^2\ge\phi(\E\cos2\pi\theta)$ by Jensen's inequality, with equality for a point mass; $\sigma^2r_0^2=\frac1{16}-\frac\tau{4\pi}$ is the tangent of $\phi$ at $\tau=0$. Given an instance $(B,t,r)$ of $\mathrm{GapCVP}_\gamma$ in dimension $n$, the verifier $\mathcal V^{(k)}=\mathcal V^{(k)}_{\tau,\sigma,\eta}$ rescales it by $r_k/r$, obtaining $\Lat'=(r_k/r)\Lat(B)$ and $t'=(r_k/r)t$, so that a YES instance has $\dist(t',\Lat')\le r_k$ and a NO instance has $\dist(t',\Lat')>\gamma r_k$. A certificate is a list $Z=(z_1,\dots,z_N)$ of integer vectors, from which the verifier forms $w_i:=(r/r_k)B^{-\top}z_i\in\Lat'^*$ and $W=(w_1,\dots,w_N)$. With
\begin{equation}\label{eq:fW}
\begin{gathered}
f_W(y)=\frac1N\sum_{i=1}^N\cos(2\pi\langle w_i,y\rangle),\qquad
\nabla f_W(y)=-\frac{2\pi}N\sum_{i=1}^N\sin(2\pi\langle w_i,y\rangle)\,w_i,\\
\nabla^2f_W(y)=-\frac{4\pi^2}N\sum_{i=1}^N\cos(2\pi\langle w_i,y\rangle)\,w_iw_i^{\top},
\end{gathered}
\end{equation}
the verifier accepts if and only if
\begin{enumerate}[label=\textup{(\alph*)},itemsep=1pt]
\item every $w_i$ lies in $\Lat'^*$ (which holds by construction);
\item $s_{\max}\big(\sum_{i=1}^Nw_iw_i^{\top}\big)\le\sigma^2N$;
\item $f_W(t')<\tau$;
\item (if $k\ge1$) $\norm{\nabla f_W(t')}\le\eta_1$;
\item (if $k=2$) $\norm{\nabla^2f_W(t')}\le\eta_2$.
\end{enumerate}
$\mathcal V^{(0)}$ is the verifier of~\cite{AR05} with its two parameters left free (there $\tau=\frac12$, $\sigma^2=3$, and certificates are rational vectors with (a) checked explicitly); (d) is the \emph{gradient test} and (e) the \emph{Hessian test}, and $\mathcal V^{(2)}$ is the \emph{Hessian-augmented} verifier of Theorem~\ref{thm:main}. Appendix~\ref{app:precision} implements these tests in exact rational arithmetic; there the phase $\langle w_i,t'\rangle=z_i^{\top}B^{-1}t$ is rational and is reduced modulo $1$ exactly. The Gaussian certificate, with which completeness is proved, consists of $N$ independent samples from $D_{\Lat'^*}$; by~\eqref{eq:f} and~\eqref{eq:fderiv}, $\E f_W=f_{\Lat'}$, $\E\nabla f_W=\nabla f_{\Lat'}$ and $\E\nabla^2f_W=\nabla^2f_{\Lat'}$ at every point, and all three are exponentially small at a far target.

\subsection{The prover's program}

Let $\Lat\subseteq\R^n$ be a lattice and $t\in\R^n$. For a probability distribution $P$ on $\Lat^*$ with $\E_P\norm w^2<\infty$ put
\begin{equation}\label{eq:D}
\begin{gathered}
\Sigma(P):=\E_Pww^{\top},\qquad D_1(P):=\E_P\big[\sin(2\pi\langle w,t\rangle)\,w\big]\in\R^n,\\
D_2(P):=\E_P\big[\cos(2\pi\langle w,t\rangle)\,ww^{\top}\big]\in\Sym,
\end{gathered}
\end{equation}
for a target $t$ fixed throughout. For a list $W$ of vectors of $\Lat^*$ let $P_W$ be its empirical distribution, defined in Section~\ref{sec:intro}. Then by~\eqref{eq:fW}
\begin{equation}\label{eq:empirical}
\begin{gathered}
\frac1N\sum_iw_iw_i^{\top}=\Sigma(P_W),\qquad f_W(t)=\E_{P_W}\cos(2\pi\langle w,t\rangle),\\
\nabla f_W(t)=-2\pi D_1(P_W),\qquad\nabla^2f_W(t)=-4\pi^2D_2(P_W)
\end{gathered}
\end{equation}
(the last two at the fixed target $t$ of~\eqref{eq:D}), so the tests (b), (d), (e) are the constraints $\Sigma(P_W)\preceq\sigma^2I$, $\norm{D_1(P_W)}\le\eta_1/2\pi$, $\norm{D_2(P_W)}\le\eta_2/4\pi^2$ on $P_W$, and (c) is a bound on the objective below. This is an instance of the following general notion.

\begin{definition}[Linear verifiers]\label{def:linearverifier}
Let $\varphi_1,\dots,\varphi_J$ be functions on $\R^n\times\R^n$, each $\varphi_j(w,t)$ taking values in $\R$ or in the symmetric $d_j\times d_j$ matrices for some $d_j\le\mathrm{poly}(n)$, polynomially bounded in $\norm w$ and computable in polynomial time, and let $\alpha_1,\dots,\alpha_J$ be thresholds (numbers or symmetric matrices), $S\subseteq\R^n$ a closed set containing a ball around $0$, and $\tau>0$. The \emph{linear verifier} $\mathcal V(\varphi,\alpha,S,\tau)$ accepts a certificate $W$ for the rescaled pair $(\Lat',t')$ if and only if every $w_i$ lies in $\Lat'^*\cap S$, $\E_{P_W}\varphi_j(w,t')\le\alpha_j$ (respectively $\preceq\alpha_j$) for $j=1,\dots,J$, and $\E_{P_W}\cos(2\pi\langle w,t'\rangle)<\tau$. Its \emph{feasible set} at $(\Lat',t')$ is
\[\begin{aligned}
\mathcal F(\Lat',t'):=\big\{P\ :\ &P\text{ is a probability distribution on }\Lat'^*\cap S,\\
&\E_P\norm{\varphi_j(w,t')}<\infty,\quad
\E_P\varphi_j(w,t')\le\alpha_j\quad(j=1,\dots,J)\big\},
\end{aligned}\]
and its \emph{value} is $\Phi(\Lat',t'):=\inf_{P\in\mathcal F(\Lat',t')}\E_P\cos(2\pi\langle w,t'\rangle)$. A \emph{prover} for the verifier is a map assigning to every pair $(\Lat',t')$ with $\dist(t',\Lat')$ beyond a prescribed distance a feasible distribution. A family of provers is \emph{optimal} if at every such pair the objectives of its distributions converge to $\Phi(\Lat',t')$; the infimum need not be attained (Theorem~\ref{thm:order2}(ii)).
\end{definition}

The verifiers $\mathcal V^{(k)}$ are the linear verifiers with $\varphi_1=ww^{\top}$, $\alpha_1=\sigma^2I$, and, for $k\ge1$, the constraints $\norm{D_1}\le\eta_1/2\pi$ and $\norm{D_2}\le\eta_2/4\pi^2$, which are linear matrix inequalities in $P_W$: $\norm{D_1}\le\eta$ is $\begin{psmallmatrix}\eta I&D_1\\ D_1^{\top}&\eta\end{psmallmatrix}\succeq0$ by a Schur complement, and $\norm{D_2}\le\eta$ is $-\eta I\preceq D_2\preceq\eta I$; $S=\R^n$, or $S=\{\norm w\le n^2\}$ when the norm test of Section~\ref{sec:farthresholds} is added. Line tests $f_W(jt')\ge-\epsilon_j$ are linear verifiers as well. A verifier that computes something else from $W$, for instance a reduced basis of the lattice it generates, is not.

\begin{definition}[The prover's program]\label{def:program}
For $k\in\{0,1,2\}$, $\sigma^2>0$ and tolerances $\eta_1,\eta_2\ge0$, the \emph{prover's program} at $(\Lat,t)$ is
\begin{equation}\label{eq:Phi}
\begin{aligned}
\Phi^{(k)}_\sigma(\Lat,t;\eta):=\inf\Big\{\E_P\cos(2\pi\langle w,t\rangle)\ :\ &P\in\mathcal P(\Lat^*),\ \Sigma(P)\preceq\sigma^2I,\\
&\norm{D_1(P)}\le\frac{\eta_1}{2\pi}\ \text{if }k\ge1,\ \ \norm{D_2(P)}\le\frac{\eta_2}{4\pi^2}\ \text{if }k=2\Big\},
\end{aligned}
\end{equation}
where $\mathcal P(\Lat^*)$ is the set of probability distributions on $\Lat^*$ with finite second moments. We write $\Phi^{(k)}_\sigma(\Lat,t)$ when $\eta=0$.
\end{definition}

The program is feasible ($P=\delta_0$) and its value lies in $[-1,1]$. Conversely any feasible $P$ with finite support and rational masses is the empirical distribution of a list, and a general feasible $P$ is approximated by such lists with the concentration of Section~\ref{sec:prelim} (this is made precise in Theorem~\ref{thm:constant}). So, up to the slack in the parameters, \emph{the verifier $\mathcal V^{(k)}$ has an accepting certificate for $(\Lat',t')$ if and only if $\Phi^{(k)}_\sigma(\Lat',t';\eta)<\tau$}.

\subsection{Phases and the phase program}

Fix $t=y+e$ with $y\in\Lat$. For $w\in\Lat^*$, $\langle w,y\rangle\in\Z$, so every trigonometric function of $2\pi\langle w,t\rangle$ equals the same function of the \emph{phase} $\theta=\langle w,e\rangle$, a real number. For $m\ge1$ let
\begin{equation}\label{eq:Jm}
J_m(\theta):=\begin{cases}\theta^m\sin2\pi\theta,&m\text{ odd},\\ \theta^m\cos2\pi\theta,&m\text{ even},\end{cases}
\qquad\text{so that}\qquad\theta^m\cos\big(2\pi\theta+\tfrac{m\pi}2\big)=\pm J_m(\theta).
\end{equation}

\begin{lemma}[Phase reduction]\label{lem:phase}
Let $t=y+e$ with $y\in\Lat$, let $P$ be a probability distribution on $\Lat^*$ with finite second moments, and let $\theta=\langle w,e\rangle$ with $w\sim P$. Then
\[\E_P\cos(2\pi\langle w,t\rangle)=\E\cos2\pi\theta,\qquad e^{\top}\Sigma(P)e=\E\theta^2,\]
\[\langle D_1(P),e\rangle=\E J_1(\theta),\qquad e^{\top}D_2(P)e=\E J_2(\theta).\]
More generally, for a list $W$ and every $m\ge1$, the $m$-th directional derivative of $f_W$ at $t$ along $e$ is $(2\pi)^m\frac1N\sum_i\theta_i^m\cos(2\pi\theta_i+m\pi/2)=\pm(2\pi)^m\E_{P_W}J_m(\theta)$.
\end{lemma}

\begin{proof}
$\langle w,t\rangle=\langle w,y\rangle+\theta$ with $\langle w,y\rangle\in\Z$; the rest is~\eqref{eq:D} and the chain rule.
\end{proof}

\begin{definition}[The phase program]\label{def:phase}
For $k\ge0$, $\tau\ge0$ and tolerances $\eta'=(\eta'_1,\dots,\eta'_k)$ with $\eta'_m\ge0$, let
\begin{equation}\label{eq:vk}
v_k(\tau;\eta'):=\inf\Big\{\E\theta^2:\ \E\cos2\pi\theta\le\tau,\ \ |\E J_m(\theta)|\le\eta'_m\ \ (1\le m\le k)\Big\},
\end{equation}
the infimum over probability distributions of a real random variable $\theta$ with $\E|\theta|^{\max(2,k)}<\infty$, and $v_k:=v_k(0;0)$, which is~\eqref{eq:vk-intro}.
\end{definition}

The function $v_k(\tau;\eta')$ is nonincreasing in $\tau$ and in each $\eta'_m$ and nondecreasing in $k$. Its role is the following general form of soundness.

\begin{proposition}[Soundness in LP form]\label{prop:sound-general}
Let $W$ pass tests \textup{(a)}, \textup{(b)} (which do not depend on $k$) and, for $1\le m\le k$, satisfy $\sup_{\norm u=1}|\nabla^mf_W(t')[u,\dots,u]|\le\eta_m$ (which for $m=1,2$ is tests \textup{(d)}, \textup{(e)}). Write $t'=y+e$ with $y\in\Lat'$ a lattice point closest to $t'$, so that $\norm e=\dist(t',\Lat')$. If $f_W(t')<\tau$ then
\[\sigma^2\norm e^2\ \ge\ v_k\Big(\tau;\ \Big(\frac{\eta_m\norm e^m}{(2\pi)^m}\Big)_{1\le m\le k}\Big).\]
\end{proposition}

\begin{proof}
Let $\theta$ have the empirical distribution of the phases $\theta_i=\langle w_i,e\rangle$. By Lemma~\ref{lem:phase} and (a), (b),
\[\E\cos2\pi\theta=f_W(t')<\tau,\qquad\E\theta^2=e^{\top}\Sigma(P_W)e\le\sigma^2\norm e^2,\]
\[|\E J_m(\theta)|=\frac{|\nabla^mf_W(t')[e,\dots,e]|}{(2\pi)^m}\le\frac{\eta_m\norm e^m}{(2\pi)^m}\qquad(1\le m\le k).\] So $\theta$ is feasible for the phase program with these tolerances, and $\E\theta^2$ is at least its value.
\end{proof}

\subsection{Far thresholds and the constant}\label{sec:farthresholds}

The far side is governed by the same program at far targets. We normalize $\sigma^2=1/2\pi$, which costs nothing (Remark~\ref{rem:normalization}).

\begin{definition}[Far thresholds]\label{def:far}
Let $\sigma^2=1/2\pi$ and $k\in\{0,1,2\}$. A number $c>0$ is \emph{$k$-far-certifiable} if there are sequences $\tau_n,\delta_n,\eta_{1,n},\eta_{2,n}\to0$ of positive numbers such that for all sufficiently large $n$, every lattice $\Lat\subseteq\R^n$ and every $t$ with $\dist(t,\Lat)\ge c\sqrt n$ admit a probability distribution $P$ on $\{w\in\Lat^*:\norm w\le n^2\}$ with
\[\Sigma(P)\preceq(1+\delta_n)\sigma^2I,\qquad\E_P\cos(2\pi\langle w,t\rangle)\le\tau_n,\]
\[\norm{D_1(P)}\le\eta_{1,n}\ \ (k\ge1),\qquad\norm{D_2(P)}\le\eta_{2,n}\ \ (k=2).\]
The \emph{far threshold} $\cfar^{(k)}$ is the infimum of the $k$-far-certifiable $c$. Clearly $\cfar^{(0)}\le\cfar^{(1)}\le\cfar^{(2)}$.
\end{definition}

The support bound $\norm w\le n^2$ is a normalization (the scale of Gaussian samples is $\norm w\approx\sqrt n$) that makes the passage from distributions to certificates of polynomial bit length routine; the results below use it only through Theorem~\ref{thm:constant}, whose proof works verbatim for any polynomial bound $R_n\ge n^2$ with $N_n=n^4R_n^2$ samples, so nothing depends on the particular choice. Consider a family $\mathcal V_n=\mathcal V^{(k)}_{\tau_n,\sigma_n,\eta_n}$ with $\tau_n\to0$, $\sigma_n^2=(1+\delta_n)/2\pi$ with $\delta_n\to0$, $\eta_{1,n},\eta_{2,n}\to0$, reading certificates of $N_n$ vectors of norm at most $n^2$, with $N_n$ polynomially bounded (the verifier checks the norm bound; the verifier of Theorem~\ref{thm:main} does not, but adding the check changes nothing there, since its Gaussian certificates have $\norm{w_i}\le\sqrt n$ on the event $\mathcal E$ of Lemma~\ref{lem:EBcomplete}); by~\eqref{eq:rk} its close radius $r_k(n)$ tends to
\begin{equation}\label{eq:rhok}
\varrho_0:=\sqrt{\pi/8}=0.6267,\qquad\varrho_1:=\sqrt{\pi/4}=0.8862,\qquad\varrho_2:=\sqrt{3\pi/8}=1.0854 .
\end{equation}
The family \emph{achieves} $c$ if for all sufficiently large $n$, every lattice $\Lat\subseteq\R^n$ and every $t$ with $\dist(t,\Lat)>c\,r_k(n)\sqrt n$ have an accepting certificate of polynomial bit length. Let $\cAR^{(k)}$ be the infimum of the $c$ achieved by some such family. A family that achieves $c$ and is sound (Section~\ref{sec:close}) places $\mathrm{GapCVP}_{c\sqrt n}$ in $\coNP$ for all large $n$, exactly as in the proof of Theorem~\ref{thm:main}. Theorem~\ref{thm:constant} below proves
\begin{equation}\label{eq:cAR}
\cAR^{(k)}=\frac{\cfar^{(k)}}{\varrho_k}\qquad(k=0,1,2),\qquad\text{in particular}\qquad\cAR^{(2)}=\sqrt{\frac{8}{3\pi}}\;\cfar^{(2)}=0.9213\,\cfar^{(2)} .
\end{equation}

\section{Duality}\label{sec:duality}

Both programs are linear programs over probability measures with finitely many linear or conic constraints. In both, the dual variables have a concrete meaning, and we prove strong duality by a separation argument in finite dimension; the only tool is the following standard fact (\cite[Theorem~11.3]{Rockafellar70}): if $C_1,C_2\subseteq\R^d$ are nonempty, convex and disjoint, and $C_2$ has nonempty interior, then there is a nonzero linear functional $\ell$ on $\R^d$ and $\alpha\in\R$ with $\ell\le\alpha$ on $C_1$ and $\ell\ge\alpha$ on $C_2$.

\subsection{The phase program and envelope inequalities}

\begin{definition}[Envelopes]\label{def:envelope}
An \emph{envelope of order $k$} is a triple $(v,a,b)$ with $v\in\R$, $a\ge0$, $b=(b_1,\dots,b_k)\in\R^k$ such that
\begin{equation}\label{eq:envelope}
\theta^2\ \ge\ v-a\cos2\pi\theta-\sum_{m=1}^kb_mJ_m(\theta)\qquad\text{for every }\theta\in\R .
\end{equation}
\end{definition}

\begin{theorem}[Duality for the phase program]\label{thm:phase-duality}
Let $k\ge0$, $\tau\ge0$ and $\eta'\in\R^k_{\ge0}$.
\begin{enumerate}[label=\textup{(\roman*)}]
\item \textup{(Weak duality)} For every envelope $(v,a,b)$ of order $k$, $v_k(\tau;\eta')\ge v-a\tau-\sum_m|b_m|\eta'_m$.
\item \textup{(Strong duality)} If the phase program has a finitely supported feasible distribution $Q_0$ satisfying all its constraints strictly, in particular if $\tau>0$ and $\eta'_m>0$ for all $m$ (then the distributions of Lemma~\ref{lem:faratoms}, which have $\E\cos2\pi\theta=0$ and $\E J_m=0$ for all $m\le k$, qualify), then
\[v_k(\tau;\eta')=\max\Big\{v-a\tau-\sum_m|b_m|\eta'_m:\ (v,a,b)\ \text{an envelope of order }k\Big\},\]
and the maximum is attained.
\end{enumerate}
\end{theorem}

\begin{proof}
(i) Take expectations in~\eqref{eq:envelope} with respect to a feasible distribution: $\E\theta^2\ge v-a\E\cos2\pi\theta-\sum_mb_m\E J_m\ge v-a\tau-\sum_m|b_m|\eta'_m$, using $a\ge0$, $\E\cos2\pi\theta\le\tau$ and $|\E J_m|\le\eta'_m$.

(ii) Let $v^*=v_k(\tau;\eta')$, which is finite ($0\le v^*\le\E_{Q_0}\theta^2$). Consider in $\R^{2+2k}$ the set
\[K:=\Big\{\Big(\E_Q\theta^2+r_0,\ \E_Q\cos2\pi\theta+r,\ \E_QJ_1-\eta'_1+r_1^+,\ -\E_QJ_1-\eta'_1+r_1^-,\ \dots\Big):\]
\[\hspace{8em}Q\ \text{finitely supported},\ r_0,r,r_m^\pm\ge0\Big\},\]
where the coordinates after the first two list, for each $m$, the pairs $\E_QJ_m-\eta'_m+r_m^+$ and $-\E_QJ_m-\eta'_m+r_m^-$. The set $K$ is convex, being the image of a convex set under an affine map plus a convex cone. Fix $\epsilon>0$ and let $C_2:=\{(x_0,x,y^\pm_1,\dots):\ x_0<v^*-\epsilon,\ x<\tau,\ y_m^\pm<0\}$, an open convex set. If a point of $K$ lay in $C_2$, the corresponding $Q$ would be feasible with $\E_Q\theta^2<v^*-\epsilon$, contradicting the definition of $v^*$. So $K\cap C_2=\emptyset$, and there is a nonzero $(\alpha_0,\alpha,\beta^\pm_m)$ and $\gamma$ with
\[\alpha_0x_0+\alpha x+\sum_m(\beta_m^+y_m^++\beta_m^-y_m^-)\ \ge\ \gamma\ \text{ on }K,\qquad\le\gamma\ \text{ on }C_2 .\]
Since $C_2$ is unbounded below in every coordinate, $\alpha_0,\alpha,\beta_m^\pm\ge0$, and $\gamma\ge\sup_{C_2}=\alpha_0(v^*-\epsilon)+\alpha\tau$. On $K$ take $Q=\delta_\theta$ and all $r=0$:
\begin{equation}\label{eq:sep}
\alpha_0\theta^2+\alpha\cos2\pi\theta+\sum_m(\beta_m^+-\beta_m^-)J_m(\theta)-\sum_m(\beta_m^++\beta_m^-)\eta'_m\ \ge\ \alpha_0(v^*-\epsilon)+\alpha\tau\qquad\forall\theta\in\R .
\end{equation}
If $\alpha_0=0$, then integrating~\eqref{eq:sep} against the strictly feasible $Q_0$ gives $\alpha(\E_{Q_0}\cos2\pi\theta-\tau)+\sum_m\big(\beta_m^+(\E_{Q_0}J_m-\eta'_m)-\beta_m^-(\E_{Q_0}J_m+\eta'_m)\big)\ge0$, where every bracket is strictly negative; hence $\alpha=\beta_m^\pm=0$, contradicting that the functional is nonzero. So $\alpha_0>0$; dividing by it and writing $a=\alpha/\alpha_0$, $b_m=(\beta_m^+-\beta_m^-)/\alpha_0$, $v:=v^*-\epsilon+a\tau+\sum_m(\beta_m^++\beta_m^-)\eta'_m/\alpha_0$, inequality~\eqref{eq:sep} says that $(v,a,b)$ is an envelope, and its dual value is $v-a\tau-\sum|b_m|\eta'_m\ge v^*-\epsilon$ because $|b_m|\le(\beta_m^++\beta_m^-)/\alpha_0$. With (i), the supremum of the dual values is $v^*$. It is attained: if an envelope has dual value $\ge v^*-1$, then integrating~\eqref{eq:envelope} against $Q_0$, whose constraints have slack $s>0$, gives
\[\E_{Q_0}\theta^2\ \ge\ v-a(\tau-s)-\sum_m|b_m|(\eta'_m-s)\ \ge\ v^*-1+s\Big(a+\sum_m|b_m|\Big),\]
so $a+\sum_m|b_m|\le(\E_{Q_0}\theta^2-v^*+1)/s$, and $v^*-1\le v\le a$ (take $\theta=0$ in~\eqref{eq:envelope}); such envelopes lie in a compact set, which is closed since~\eqref{eq:envelope} is a pointwise condition, and the dual value is continuous on it.
\end{proof}

The dual variables are the multipliers of the three kinds of tests: $a$ for the threshold, $b_1$ for the gradient, $b_2$ for the Hessian, and so on. In Section~\ref{sec:close} we exhibit, for $k=0,1,2$, an envelope and a feasible distribution with the same value; by (i) this proves the value and shows that no slack is lost, without appeal to (ii). Part (ii) says that this is always possible in principle, i.e. that \emph{every} sound close radius of a verifier with derivative tests is certified by an envelope inequality, which is the sense in which the envelope method is complete.

\subsection{The prover's program and quadratic certificates}

\begin{theorem}[Duality for the prover's program]\label{thm:prover-duality}
Let $\Lat\subseteq\R^n$ be a lattice, $t\in\R^n$, $\sigma^2>0$, and $\eta_1,\eta_2>0$. Write $c(w)=\cos(2\pi\langle w,t\rangle)$ and $s(w)=\sin(2\pi\langle w,t\rangle)$. Then
\begin{equation}\label{eq:prover-dual}
\Phi^{(2)}_\sigma(\Lat,t;\eta)=\max\Big\{\lambda-\sigma^2\tr M-\frac{\eta_1}{2\pi}\norm u-\frac{\eta_2}{4\pi^2}\norm U_*\Big\},
\end{equation}
the maximum over all $\lambda\in\R$, $M\in\Sym$ with $M\succeq0$, $u\in\R^n$ and $U\in\Sym$ such that
\begin{equation}\label{eq:dualfeas}
c(w)+w^{\top}Mw+\langle u,w\rangle\,s(w)+c(w)\,w^{\top}Uw\ \ge\ \lambda\qquad\text{for every }w\in\Lat^*,
\end{equation}
and the maximum is attained. The same holds for $\Phi^{(1)}_\sigma$ with $U=0$ and for $\Phi^{(0)}_\sigma$ with $u=0$, $U=0$ (in which case $\eta$ plays no role).
\end{theorem}

\begin{proof}
\emph{Weak duality.} Let $P$ be feasible and $(\lambda,M,u,U)$ satisfy~\eqref{eq:dualfeas}. Integrating~\eqref{eq:dualfeas} against $P$,
\[\E_Pc(w)\ \ge\ \lambda-\tr\big(M\Sigma(P)\big)-\langle u,D_1(P)\rangle-\tr\big(UD_2(P)\big)\ \ge\ \lambda-\sigma^2\tr M-\frac{\eta_1}{2\pi}\norm u-\frac{\eta_2}{4\pi^2}\norm U_*,\]
because $\tr(M\Sigma(P))\le\sigma^2\tr M$ for $M\succeq0$ and $\Sigma(P)\preceq\sigma^2I$, $|\langle u,D_1\rangle|\le\norm u\norm{D_1}$, and $|\tr(UD_2)|\le\norm U_*\norm{D_2}$.

\emph{Strong duality.} Let $\Phi=\Phi^{(2)}_\sigma(\Lat,t;\eta)$ and $\epsilon>0$. In the space $\R\times\Sym\times\R^n\times\Sym$ consider the convex sets
\begin{align*}
K&:=\big\{(\E_Pc+r,\ \Sigma(P)+S,\ D_1(P),\ D_2(P)):\ P\ \text{finitely supported on }\Lat^*,\ r\ge0,\ S\succeq0\big\},\\
C_2&:=\big\{(x,B,y,Y):\ x<\Phi-\epsilon,\ B\prec\sigma^2I,\ \norm y<\eta_1/2\pi,\ \norm Y<\eta_2/4\pi^2\big\}.
\end{align*}
$C_2$ is open and convex, and $K\cap C_2=\emptyset$ (a point of $K$ in $C_2$ would come from a feasible $P$ with $\E_Pc<\Phi-\epsilon$). Separation gives $(\alpha,M,u,U)\ne0$ and $\gamma$ with
\[\alpha x+\tr(MB)+\langle u,y\rangle+\tr(UY)\le\gamma\ \text{ on }C_2,\qquad\ge\gamma\ \text{ on }K.\]
Letting $x\to-\infty$ gives $\alpha\ge0$; letting $B\to-\infty$ in a direction $-vv^{\top}$ gives $v^{\top}Mv\ge0$, i.e. $M\succeq0$; and taking suprema over $C_2$,
\[\gamma\ \ge\ \alpha(\Phi-\epsilon)+\sigma^2\tr M+\frac{\eta_1}{2\pi}\norm u+\frac{\eta_2}{4\pi^2}\norm U_* .\]
If $\alpha=0$, evaluating the functional at the point of $K$ given by $P=\delta_0$, $r=0$, $S=0$ gives $0\ge\gamma\ge\sigma^2\tr M+\frac{\eta_1}{2\pi}\norm u+\frac{\eta_2}{4\pi^2}\norm U_*$, whence $M=0$, $u=0$, $U=0$ (as $\eta_1,\eta_2>0$ and $M\succeq0$), a contradiction. So $\alpha>0$; normalize $\alpha=1$, and evaluate at $P=\delta_w$, $r=0$, $S=0$:
\[c(w)+w^{\top}Mw+\langle u,w\rangle s(w)+c(w)w^{\top}Uw\ \ge\ \gamma\ \ge\ \Phi-\epsilon+\sigma^2\tr M+\frac{\eta_1}{2\pi}\norm u+\frac{\eta_2}{4\pi^2}\norm U_*=:\lambda .\]
Thus $(\lambda,M,u,U)$ satisfies~\eqref{eq:dualfeas} with dual value $\Phi-\epsilon$. Together with weak duality, the supremum of the dual values is $\Phi$. Attainment: dual feasible points with value $\ge\Phi-1$ have $\lambda\le1$ (take $w=0$ in~\eqref{eq:dualfeas}) and hence $\sigma^2\tr M+\frac{\eta_1}{2\pi}\norm u+\frac{\eta_2}{4\pi^2}\norm U_*\le2-\Phi\le3$, a compact set, on which the constraints~\eqref{eq:dualfeas} are closed.
\end{proof}

The separation used only finitely supported distributions. Running it with $\Phi$ replaced by the infimum $\Phi_{\mathrm{fin}}\ge\Phi$ over finitely supported feasible $P$ gives dual values at least $\Phi_{\mathrm{fin}}-\epsilon$, so by weak duality $\Phi_{\mathrm{fin}}=\Phi$; the same holds for the phase program in Theorem~\ref{thm:phase-duality}(ii).

\begin{remark}[The class of linear verifiers]\label{rem:linearverifiers}
Nothing in the proof of Theorem~\ref{thm:prover-duality} uses the particular test functions: weak duality needs only that the $\varphi_j$ be integrable, which is part of Definition~\ref{def:linearverifier}, the separation uses only finitely supported distributions, and the case $\alpha=0$ is excluded by a strictly feasible point. So for a linear verifier $\mathcal V(\varphi,\alpha,S,\tau)$ with a finitely supported feasible distribution $Q_0$ satisfying every constraint strictly (such as $\delta_0$ when every constraint has slack at $w=0$), the same argument gives
\[\Phi(\Lat,t)=\max\Big\{\lambda-\sum_j\langle\alpha_j,Y_j\rangle:\ Y_j\succeq0,\ \ \cos(2\pi\langle w,t\rangle)+\sum_j\langle\varphi_j(w,t),Y_j\rangle\ge\lambda\ \text{for all }w\in\Lat^*\cap S\Big\},\]
with scalar multipliers $Y_j\ge0$ for scalar constraints; the maximum is attained, because integrating the dual inequality against $Q_0$ bounds $\lambda$ and the $\tr Y_j$ on the dual solutions of value at least $\Phi-1$. The exact values of Section~\ref{sec:close} and the identification of the constant in Theorem~\ref{thm:constant} are proved for the verifiers $\mathcal V^{(k)}$; their proofs use the phase reduction (Lemma~\ref{lem:phase}) and the scalings of Lemma~\ref{lem:realize}(iv) and Remark~\ref{rem:normalization}, which a general test function need not share.
\end{remark}

\begin{remark}[Reading the dual]\label{rem:dual-reading}
A dual feasible point of value at least $\tau$ proves that no certificate is accepted: in the spectral part it is a quadratic form $M\succeq0$ of small trace such that $\cos(2\pi\langle w,t\rangle)$ is bounded below by $\lambda-w^{\top}Mw$ on the whole dual lattice, i.e. the dual vectors with $\langle w,t\rangle$ far from an integer lie outside the ellipsoid $\{w^{\top}Mw\le\lambda\}$. On the close side the dual solutions have rank one. If $t=y+e$ with $y\in\Lat$ and $(v,a,b)$ is an envelope of order $k\le2$ with $a>0$ (Definition~\ref{def:envelope}), then dividing~\eqref{eq:envelope} by $a$ and putting $\theta=\langle w,e\rangle$ shows that
\[\lambda=\frac va,\qquad M=\frac1a\,ee^{\top},\qquad u=\frac{b_1}a\,e,\qquad U=\frac{b_2}a\,ee^{\top}\]
satisfies~\eqref{eq:dualfeas}, with value $\big(v-\sigma^2\norm e^2-\frac{\eta_1|b_1|}{2\pi}\norm e-\frac{\eta_2|b_2|}{4\pi^2}\norm e^2\big)/a$. For the envelopes of Lemmas~\ref{lem:env01} and~\ref{lem:env2} this value is at least $\tau$ exactly when $\norm e\le r_k$, with $r_k$ as in~\eqref{eq:rk}; so Proposition~\ref{prop:sound} is weak duality for these rank-one solutions, and Theorems~\ref{thm:phase-duality} and~\ref{thm:exact} show that, for derivative tests, restricting the dual to rank one along the offset loses nothing in the close radius.
\end{remark}

\begin{remark}[Normalization]\label{rem:normalization}
Scaling a lattice by $\lambda>0$ scales $\Lat^*$ by $\lambda^{-1}$, distances by $\lambda$, $\Sigma(P)$ by $\lambda^{-2}$, and leaves $\E_Pc$, the direction of $D_1$ and the sign pattern of $D_2$ unchanged, while $\norm{D_1}$ and $\norm{D_2}$ scale by $\lambda^{-1}$, $\lambda^{-2}$; the close radii~\eqref{eq:rk} scale by $\sigma^{-1}$. The ratio of far threshold to close radius is therefore the same at every spectral level, and $\sigma^2=1/2\pi$ in Definition~\ref{def:far} is no loss of generality, provided the tolerances and the support bound are rescaled along with $\sigma$.
\end{remark}

\section{The close side: exact solution of the phase program}\label{sec:close}

\subsection{Three envelopes}

\begin{lemma}[Envelopes of order $0$ and $1$]\label{lem:env01}
The inequalities~\eqref{eq:env0} and~\eqref{eq:env1} hold for every real $\theta$, with equality exactly at $\theta=\pm\frac14$ in the first and at $\theta\in\{0,\pm\frac12\}$ in the second.
\end{lemma}

\begin{proof}
Let $g(\theta)=\theta^2-\frac1{16}+\frac1{4\pi}\cos2\pi\theta$, an even function with $g(\frac14)=0$ and $g'(\theta)=2\theta-\frac12\sin2\pi\theta$. On $[0,\frac14]$ the concave function $\sin2\pi\theta$ lies above its chord $4\theta$, so $g'\le0$; for $\theta\ge\frac14$, $2\theta\ge\frac12\ge\frac12\sin2\pi\theta$, so $g'\ge0$, with equality only where $\theta=\frac14$ and $\sin2\pi\theta=1$. Hence $g$ decreases on $[0,\frac14]$ to $0$ and increases strictly afterwards, and $g\ge0$ with equality only at $\pm\frac14$.

For the second inequality put $\varphi=2\pi\theta$; multiplied by $4\pi^2$, the claim reads $G(\varphi)\ge0$ for $G(\varphi)=\varphi^2-\frac{\pi^2}2+\frac{\pi^2}2\cos\varphi+2\varphi\sin\varphi$, an even function with $G(0)=G(\pi)=0$ and
\[G'(\varphi)=2\varphi(1+\cos\varphi)-\Big(\frac{\pi^2}2-2\Big)\sin\varphi=4\cos^2\frac\varphi2\Big(\varphi-\kappa_0\tan\frac\varphi2\Big),\qquad\kappa_0=\frac{\pi^2}4-1=1.4674 .\]
On $(0,\pi)$ the sign of $G'$ is that of $h_1(\varphi)=\varphi-\kappa_0\tan\frac\varphi2$, which vanishes at $0$, has derivative $1-\frac{\kappa_0}2\sec^2\frac\varphi2$ positive at $0$ (as $\kappa_0<2$) and decreasing, and tends to $-\infty$ at $\pi$; so $h_1$ increases from $0$ and then decreases through a unique zero, $G$ increases from $G(0)=0$ and then decreases to $G(\pi)=0$, and $G>0$ on $(0,\pi)$. On $(\pi,2\pi)$, $\tan\frac\varphi2<0$, so $G'>0$ and $G>G(\pi)=0$. For $\varphi\ge2\pi$, $G(\varphi)\ge\varphi^2-\pi^2-2\varphi>0$.
\end{proof}

\begin{lemma}[The Hessian envelope]\label{lem:env2}
The inequality~\eqref{eq:env2} holds for every real $\theta$; that is, the difference of its two sides,
\begin{equation}\label{eq:h}
h(\theta):=\theta^2-\frac3{16}+\frac5{8\pi}\cos2\pi\theta+\frac12\,\theta\sin2\pi\theta-\frac2\pi\,\theta^2\cos2\pi\theta,
\end{equation}
is nonnegative on $\R$, with equality if and only if $\theta\in\{\pm\frac14,\pm\frac34\}$, where $h$ has double zeros.
\end{lemma}

The multipliers in~\eqref{eq:h} are forced. Let $(v,a,b_1,b_2)$ be any envelope of order two with $v=\frac3{16}$, and $g(\theta)\ge0$ the difference of the two sides of~\eqref{eq:envelope}. The distribution $\frac34\delta_{1/4}+\frac14\delta_{3/4}$ has $\E\theta^2=\frac3{16}$ and $\E\cos2\pi\theta=\E J_1=\E J_2=0$, so $\E g(\theta)=0$ under it; hence $g$ vanishes at $\frac14$ and $\frac34$, which are then minima of $g$, so $g'$ vanishes there too, and these four conditions determine $(a,b_1,b_2)=(\frac5{8\pi},\frac12,-\frac2\pi)$. The content of the lemma is that this envelope is valid on all of $\R$. The function $h$ is small on $[-\frac34,\frac34]$, where its maximum is $0.0241$, so crude estimates do not suffice there; the proof below is exact, by a chain of three auxiliary functions whose monotonicity is established from their derivatives.

\subsection{Proof of the Hessian envelope}\label{sec:envproof}

Since $h$ is even, it suffices to treat $\theta\ge0$. We write $\varphi=2\pi\theta$ throughout.

\emph{Step 1: $\theta\ge1$.} Write $\theta=j+u$ with $j\ge1$ an integer and $u\in[0,1)$; then $\cos\varphi=\cos2\pi u$ and $\sin\varphi=\sin2\pi u$. Put $a=\frac5{8\pi}=0.1989$. If $u\in[0,\frac12]$, then $\sin\varphi\ge0$ and $\frac2\pi\theta^2-a>0$, so
\[h(\theta)\ \ge\ \theta^2-\frac3{16}-\Big(\frac2\pi\theta^2-a\Big)\cos\varphi\ \ge\ \Big(1-\frac2\pi\Big)\theta^2+a-\frac3{16}\ \ge\ 1-\frac2\pi+\frac5{8\pi}-\frac3{16}\ =\ 0.3748\ >\ 0 .\]
If $u\in(\frac12,\frac34]$, then $\sin\varphi\le0$ and $\cos\varphi\le0$, so $\frac12\theta\sin\varphi\ge-\frac\theta2$, $a\cos\varphi\ge-a$ and $-\frac2\pi\theta^2\cos\varphi\ge0$; hence, as $\theta\ge\frac32$,
\[h(\theta)\ \ge\ \theta^2-\frac\theta2-a-\frac3{16}\ \ge\ \frac94-\frac34-a-\frac3{16}\ >\ 1 .\]
If $u\in(\frac34,1)$, then $\sin\varphi<0$ and $\cos\varphi>0$, so $a\cos\varphi\ge0$ and $-\frac2\pi\theta^2\cos\varphi\ge-\frac2\pi\theta^2$; hence, as $\theta\ge\frac74$ and the right side below increases in $\theta$,
\[h(\theta)\ \ge\ \Big(1-\frac2\pi\Big)\theta^2-\frac\theta2-\frac3{16}\ \ge\ \frac{49}{16}\Big(1-\frac2\pi\Big)-\frac78-\frac3{16}\ =\ \frac{32-98/\pi}{16}\ =\ 0.0504\ >\ 0 .\]

\emph{Step 2: the function $H$ on $[0,2\pi]$.} For $\theta\in[0,1]$ put
\begin{equation}\label{eq:H}
H(\varphi):=16\pi^2\,h\Big(\frac\varphi{2\pi}\Big)=4\varphi^2-\frac8\pi\varphi^2\cos\varphi+4\pi\varphi\sin\varphi+10\pi\cos\varphi-3\pi^2,\qquad\varphi\in[0,2\pi].
\end{equation}
One checks $H(\frac\pi2)=\pi^2+2\pi^2-3\pi^2=0$ and $H(\frac{3\pi}2)=9\pi^2-6\pi^2-3\pi^2=0$. We must show $H\ge0$ on $[0,2\pi]$ with no other zeros. Let
\[\kappa':=4\pi-\frac{16}\pi=7.4734,\qquad b:=\frac8\pi=2.5465,\]
and record the signs we shall use:
\begin{equation}\label{eq:signs}
\begin{gathered}
8-\kappa'=0.5266>0,\qquad 8-4b=-2.1859<0,\\
6b-2\kappa'=\frac{80}\pi-8\pi=0.3320\in(0,8),\qquad 10b-2\kappa'=\frac{112}\pi-8\pi=10.5180>0 .
\end{gathered}
\end{equation}
Differentiating~\eqref{eq:H},
\begin{equation}\label{eq:Hp}
H'(\varphi)=8\varphi+\kappa'\varphi\cos\varphi+(b\varphi^2-6\pi)\sin\varphi,
\end{equation}
and $H'(\frac\pi2)=4\pi+(2\pi-6\pi)=0$, $H'(\frac{3\pi}2)=12\pi-(18\pi-6\pi)=0$, as it must be at double zeros. We shall determine the sign of $H'$ on $(0,2\pi)$ through three auxiliary functions. On $(0,\pi)\cup(\pi,2\pi)$ put
\begin{equation}\label{eq:psi}
\psi(\varphi):=\frac{H'(\varphi)}{\sin\varphi}=\frac{8\varphi}{\sin\varphi}+\kappa'\varphi\cot\varphi+b\varphi^2-6\pi,\qquad\text{so}\qquad\psi\Big(\frac\pi2\Big)=\psi\Big(\frac{3\pi}2\Big)=0,
\end{equation}
and compute, using $(\varphi/\sin\varphi)'=(\sin\varphi-\varphi\cos\varphi)/\sin^2\varphi$ and $(\varphi\cot\varphi)'=(\sin\varphi\cos\varphi-\varphi)/\sin^2\varphi$,
\begin{equation}\label{eq:q}
\psi'(\varphi)\sin^2\varphi=q(\varphi):=8(\sin\varphi-\varphi\cos\varphi)+\kappa'(\sin\varphi\cos\varphi-\varphi)+2b\varphi\sin^2\varphi .
\end{equation}
Then, since $(\sin\varphi\cos\varphi)'=\cos^2\varphi-\sin^2\varphi=1-2\sin^2\varphi$,
\begin{equation}\label{eq:qp}
q'(\varphi)=8\varphi\sin\varphi-2\kappa'\sin^2\varphi+2b\sin^2\varphi+4b\varphi\sin\varphi\cos\varphi=\sin\varphi\cdot B(\varphi),
\end{equation}
where
\begin{equation}\label{eq:B}
B(\varphi):=\varphi(8+4b\cos\varphi)-2(\kappa'-b)\sin\varphi,
\end{equation}
and finally
\begin{equation}\label{eq:Bp}
B'(\varphi)=8+(6b-2\kappa')\cos\varphi-4b\varphi\sin\varphi,\qquad B''(\varphi)=-(10b-2\kappa')\sin\varphi-4b\varphi\cos\varphi .
\end{equation}
The values we need are
\begin{equation}\label{eq:values}
\begin{gathered}
B(0)=0,\qquad B(\pi)=\pi(8-4b)<0,\qquad B(2\pi)=2\pi(8+4b)>0,\\
B'(0)=8+(6b-2\kappa')>0,\qquad B'\Big(\frac\pi2\Big)=8-2\pi b=-8,\qquad B'(\pi)=8-(6b-2\kappa')>0,\\
q(0)=0,\qquad q(\pi)=\pi(8-\kappa')>0,\qquad q(2\pi)=-2\pi(8+\kappa')<0,\\
q\Big(\frac{3\pi}2\Big)=-8-\frac{3\pi}2\kappa'+3\pi b=40-6\pi^2<0 .
\end{gathered}
\end{equation}

\emph{Step 3: the sign of $H'$ on $(0,\pi)$.} We show $q>0$ on $(0,\pi]$. First, $B''<0$ on $(0,\frac\pi2]$, both terms in~\eqref{eq:Bp} being nonpositive and not both zero. On $(\frac\pi2,\pi)$ write $\varphi=\pi-\varphi'$ with $\varphi'\in(0,\frac\pi2)$; then
\[B''(\varphi)=\cos\varphi'\Big[4b(\pi-\varphi')-(10b-2\kappa')\tan\varphi'\Big],\]
and the bracket is strictly decreasing in $\varphi'$, from $4b\pi>0$ to $-\infty$. Hence there is $\varphi^{**}\in(\frac\pi2,\pi)$ with $B''<0$ on $(0,\varphi^{**})$ and $B''>0$ on $(\varphi^{**},\pi)$: $B'$ is strictly concave on $[0,\varphi^{**}]$ and strictly convex on $[\varphi^{**},\pi]$. By~\eqref{eq:values}, $B'(0)>0>B'(\frac\pi2)$, so the concave function $B'$ has exactly one zero $\varphi_a\in(0,\frac\pi2)$ on $[0,\varphi^{**}]$ and is negative on $[\frac\pi2,\varphi^{**}]$ (a concave function that is positive at $0$ and negative at $\frac\pi2$ is negative beyond $\frac\pi2$); and since $B'(\pi)>0$, the convex function $B'$ has exactly one zero $\varphi_c$ on $[\varphi^{**},\pi]$. Thus $B$ increases on $[0,\varphi_a]$, decreases on $[\varphi_a,\varphi_c]$ and increases on $[\varphi_c,\pi]$. As $B(0)=0$ and $B(\pi)<0$, we get $B(\varphi_c)<B(\pi)<0$, and $B$ has exactly one zero $\varphi_b\in(\varphi_a,\varphi_c)$ in $(0,\pi]$, with $B>0$ on $(0,\varphi_b)$ and $B<0$ on $(\varphi_b,\pi]$. By~\eqref{eq:qp}, $q$ increases on $[0,\varphi_b]$ and decreases on $[\varphi_b,\pi]$; since $q(0)=0$ and $q(\pi)>0$, $q>0$ on $(0,\pi]$. By~\eqref{eq:q}, $\psi$ is strictly increasing on $(0,\pi)$, and $\psi(\frac\pi2)=0$; as $\sin\varphi>0$ there, $H'=\psi\sin\varphi$ is negative on $(0,\frac\pi2)$ and positive on $(\frac\pi2,\pi)$.

\emph{Step 4: the sign of $H'$ on $(\pi,2\pi)$.} Here $\sin\varphi\le0$, so by~\eqref{eq:Bp} and~\eqref{eq:signs}, $B'\ge8-(6b-2\kappa')>0$: $B$ is strictly increasing on $[\pi,2\pi]$, from $B(\pi)<0$ to $B(2\pi)>0$, with a unique zero $\varphi_B$. By~\eqref{eq:qp}, $q'=B\sin\varphi$ is positive on $(\pi,\varphi_B)$ and negative on $(\varphi_B,2\pi)$. Since $q(\pi)>0$, $q>0$ on $[\pi,\varphi_B]$; since $q(\frac{3\pi}2)<0$, we have $\varphi_B<\frac{3\pi}2$, and $q$ has exactly one zero $\varphi_q\in(\varphi_B,\frac{3\pi}2)$ in $(\pi,2\pi)$, with $q>0$ before it and $q<0$ after it. By~\eqref{eq:q}, $\psi$ increases on $(\pi,\varphi_q]$ and decreases on $[\varphi_q,2\pi)$. At the ends,
\[\psi(\varphi)=\frac{\varphi(8+\kappa'\cos\varphi)+(b\varphi^2-6\pi)\sin\varphi}{\sin\varphi}\ \longrightarrow\ -\infty\qquad(\varphi\to\pi^+\text{ and }\varphi\to2\pi^-),\]
because the numerator tends to $\pi(8-\kappa')>0$, respectively $2\pi(8+\kappa')>0$, while $\sin\varphi\to0^-$. Since $\psi(\frac{3\pi}2)=0$ and $\varphi_q<\frac{3\pi}2$, $\psi$ is decreasing on $[\frac{3\pi}2,2\pi)$, so $\psi(\varphi_q)>0$ and $\psi<0$ on $(\frac{3\pi}2,2\pi)$; and $\psi$ increases on $(\pi,\varphi_q]$ from $-\infty$ to $\psi(\varphi_q)>0$, so it has exactly one zero $\varphi_1\in(\pi,\varphi_q)$ there, with $\psi<0$ on $(\pi,\varphi_1)$ and $\psi>0$ on $(\varphi_1,\frac{3\pi}2)$. Multiplying by $\sin\varphi<0$: $H'>0$ on $(\pi,\varphi_1)$, $H'<0$ on $(\varphi_1,\frac{3\pi}2)$, $H'>0$ on $(\frac{3\pi}2,2\pi)$.

\emph{Step 5: conclusion.} By Steps 3 and 4, $H$ decreases on $[0,\frac\pi2]$ to $H(\frac\pi2)=0$, increases on $[\frac\pi2,\varphi_1]$, decreases on $[\varphi_1,\frac{3\pi}2]$ to $H(\frac{3\pi}2)=0$, and increases on $[\frac{3\pi}2,2\pi]$. Hence $H\ge0$ on $[0,2\pi]$, with equality exactly at $\frac\pi2$ and $\frac{3\pi}2$, i.e. $h\ge0$ on $[0,1]$ with equality exactly at $\frac14,\frac34$; with Step~1 and evenness, $h\ge0$ on $\R$ with equality exactly at $\pm\frac14,\pm\frac34$. The zeros are double because $H'(\frac\pi2)=H'(\frac{3\pi}2)=0$ while $h''(\frac14)=6-\frac{\pi^2}2>0$ and $h''(\frac34)=\frac{3\pi^2}2-10>0$. \qed

\medskip
Numerically, $\varphi_1=3.3903$ (i.e. $\theta=0.5396$, where $h=0.0241$), and $h(0)=\frac5{8\pi}-\frac3{16}=0.0114$, $h(\frac12)=\frac1{16}-\frac5{8\pi}+\frac1{2\pi}=0.0227$.

\subsection{Proof of Theorem~\ref{thm:close}}

\begin{proof}[Proof of Theorem~\ref{thm:close}, the values]
\emph{Lower bounds.} By Theorem~\ref{thm:phase-duality}(i) with $\tau=0$, $\eta'=0$: Lemma~\ref{lem:env01} gives the envelopes $(\frac1{16},\frac1{4\pi},\emptyset)$ of order $0$ and $(\frac18,\frac18,\frac1\pi)$ of order $1$, and Lemma~\ref{lem:env2} gives the envelope $(\frac3{16},\frac5{8\pi},(\frac12,-\frac2\pi))$ of order $2$, which is also an envelope of every order $k\ge2$ (with $b_m=0$ for $m\ge3$). Hence $v_0\ge\frac1{16}$, $v_1\ge\frac18$, $v_k\ge\frac3{16}$ for $k\ge2$.

\emph{Upper bounds for $k\le2$.} The distribution $\frac12\delta_{1/4}+\frac12\delta_{-1/4}$ has $\E\cos2\pi\theta=0$ and $\E\theta^2=\frac1{16}$. The distribution $\frac12\delta_0+\frac12\delta_{1/2}$ has $\E\cos2\pi\theta=\frac12(1-1)=0$, $\E J_1=\frac12\cdot\frac12\sin\pi=0$, and $\E\theta^2=\frac18$. The distribution $\frac34\delta_{1/4}+\frac14\delta_{3/4}$ has $\cos2\pi\theta=0$ on its support, hence $\E\cos2\pi\theta=0$ and $\E J_2=0$; $\E J_1=\frac34\cdot\frac14\cdot\sin\frac\pi2+\frac14\cdot\frac34\cdot\sin\frac{3\pi}2=\frac3{16}-\frac3{16}=0$; and $\E\theta^2=\frac34\cdot\frac1{16}+\frac14\cdot\frac9{16}=\frac3{16}$. This proves $v_0=\frac1{16}$, $v_1=\frac18$, $v_2=\frac3{16}$ with the stated minimizers.

\emph{$k\ge3$: the infimum is $\frac3{16}$ and is not attained.} The upper bound $v_k\le\frac3{16}$ for every $k$ follows from Lemma~\ref{lem:faratoms} below, which also gives the finitely supported rational distributions of the statement. Suppose $\theta$ were feasible for $k\ge3$ with $\E\theta^2=\frac3{16}$. Taking expectations in~\eqref{eq:h} and using the constraints, $0=\E\theta^2-\frac3{16}=\E h(\theta)-\frac5{8\pi}\E\cos2\pi\theta\ge\E h(\theta)$, so $\E h(\theta)\le0$, whence $h(\theta)=0$ almost surely and $\theta\in\{\pm\frac14,\pm\frac34\}$ by Lemma~\ref{lem:env2}. Let $p=\PP(\theta=\pm\frac14)$ and $q=1-p$. Then $J_1(\pm\frac14)=\frac14$, $J_1(\pm\frac34)=-\frac34$, $J_3(\pm\frac14)=\frac1{64}$, $J_3(\pm\frac34)=-\frac{27}{64}$, so $\E J_1=0$ and $\E J_3=0$ give $p=3q$ and $p=27q$, impossible for a probability distribution. So the infimum is not attained for $k\ge3$.
\end{proof}

\begin{lemma}[Far atoms]\label{lem:faratoms}
For every $m\ge1$ and $\epsilon>0$ there is a probability distribution $P$ with finite support in the set $\frac14\Z\setminus\frac12\Z$ of odd multiples of $\frac14$, with rational masses, such that $\E_PJ_i=0$ for every odd $i\le2m+1$ and $\E_P\theta^2\le\frac3{16}+\epsilon$. On this support $\cos2\pi\theta=0$, so $\E_P\cos2\pi\theta=0$ and $\E_PJ_{2i}=0$ for every $i\ge1$; hence $P$ satisfies all constraints of the phase program~\eqref{eq:vk-intro} for every $k\le2m+1$.
\end{lemma}

\begin{proof}
On the odd multiples of $\frac14$, $\sin2\pi\theta=+1$ at $\theta\in\frac14+\Z$ and $-1$ at $\theta\in\frac34+\Z$. Let $T\ge2$ be an integer, to be chosen large, and define sites and signs
\[s_j:=T^{a_j}+c_j,\qquad a_j:=3^{\,j-1},\qquad c_j:=\begin{cases}\frac14,&j\text{ odd},\\ \frac34,&j\text{ even},\end{cases}\qquad\varepsilon_j:=\sin2\pi s_j=(-1)^{j+1},\]
for $j=1,\dots,m$.
Note $s_j\ge T^{a_j}$ and $s_j\le2T^{a_j}$, so that $s_j/s_i\le2T^{-2a_j}$ for $i>j$ (as $a_i\ge3a_j$). We seek
\[P=p\,\delta_{1/4}+q\,\delta_{3/4}+\sum_{j=1}^md_j\,\delta_{s_j}\]
with $p,q,d_j>0$. The conditions $P(\R)=1$ and $\E_PJ_i=0$ for odd $i\le2m+1$ read
\begin{equation}\label{eq:system}
p+q+\sum_jd_j=1,\qquad\frac p4-\frac{3q}4+\sum_j\varepsilon_jd_js_j=0,\qquad\frac p{4^i}-q\Big(\frac34\Big)^i+\sum_j\varepsilon_jd_js_j^{\,i}=0\quad(i=3,5,\dots,2m+1),
\end{equation}
a linear system of $m+2$ equations in the $m+2$ unknowns $(p,q,d_1,\dots,d_m)$ with rational coefficients, whose solution, if unique, is rational. Put $D:=\sum_jd_j$ and $G:=\sum_j\varepsilon_jd_js_j$. The first two equations give
\[q=\frac{1-D}4+G,\qquad p=\frac{3(1-D)}4-G,\]
and substituting into the remaining ones, with $\beta_i:=(3^i-3)/4^{i+1}>0$ and $\gamma_i:=(3^i+1)/4^i$, they become
\begin{equation}\label{eq:reduced}
\sum_{j=1}^md_j\big(\varepsilon_js_j^{\,i}+\beta_i-\gamma_i\varepsilon_js_j\big)=\beta_i\qquad(i=3,5,\dots,2m+1).
\end{equation}
Index the equations by $l=1,\dots,m$ with $i_l=2l+1$, and set $Y_1:=1$, $Y_j:=\prod_{l<j}s_l^2$ for $j\ge2$. Substitute $d_j=Y_jz_j/s_j^{\,i_j}$ and divide the $l$-th equation by $Y_l$; the system becomes $\tilde A z=\tilde\beta$ with
\[\tilde A_{lj}=\frac{Y_j}{Y_l\,s_j^{\,i_j}}\Big(\varepsilon_js_j^{\,i_l}+\beta_{i_l}-\gamma_{i_l}\varepsilon_js_j\Big),\qquad\tilde\beta_l=\frac{\beta_{i_l}}{Y_l}.\]
We claim that, as $T\to\infty$, $\tilde A$ converges to the bidiagonal matrix $B_\infty$ with $(B_\infty)_{ll}=\varepsilon_l$, $(B_\infty)_{l,l-1}=\varepsilon_{l-1}$ and all other entries $0$, and $\tilde\beta\to(\beta_3,0,\dots,0)$. The main term of $\tilde A_{lj}$ is $\varepsilon_js_j^{\,i_l-i_j}Y_j/Y_l=\varepsilon_js_j^{2(l-j)}Y_j/Y_l$. For $j=l$ it is $\varepsilon_l$; for $j=l-1$ it is $\varepsilon_{l-1}s_{l-1}^2Y_{l-1}/Y_l=\varepsilon_{l-1}$; for $j<l-1$ its absolute value is
\[\frac{s_j^{2(l-j)}}{\prod_{j\le i<l}s_i^2}=\prod_{j<i<l}\Big(\frac{s_j}{s_i}\Big)^2\ \le\ \big(2T^{-2a_j}\big)^{2(l-j-1)}\to0;\]
and for $j>l$ it is $\prod_{l\le i<j}(s_i/s_j)^2\to0$. Since $\beta_{i_l}+\gamma_{i_l}s_j\le2s_j$, the remaining part of $\tilde A_{lj}$ is at most $2Y_j/(s_j^{\,i_j-1}Y_l)$, which tends to $0$. For $j<l$, use $Y_j\le Y_l$ and $s_j^{\,i_j-1}\ge T^2$. For $j\ge l$, use $Y_j/Y_l\le s_j^{2(j-l)}$ and $s_j^{\,i_j-1}=s_j^{2j}\ge s_j^{2(j-l)}T^2$. Finally $\tilde\beta_1=\beta_3$ and $\tilde\beta_l=\beta_{i_l}/Y_l\le1/T^2$ for $l\ge2$. The matrix $B_\infty$ is lower triangular with diagonal entries $\pm1$, hence invertible. Solving $B_\infty z^*=(\beta_3,0,\dots,0)^\top$ gives $z^*_1=\beta_3$ and $z^*_l=-\varepsilon_l\varepsilon_{l-1}z^*_{l-1}=z^*_{l-1}$, since $\varepsilon_l\varepsilon_{l-1}=-1$. Thus $z^*=\beta_3(1,\dots,1)$, with $\beta_3=3/32$. By continuity of matrix inversion, for all large $T$ the system~\eqref{eq:reduced} has a unique solution with $z_j\in[\frac12\beta_3,2\beta_3]$ for every $j$, hence $d_j>0$. Moreover
\[d_js_j^2=\frac{z_jY_j}{s_j^{\,2j-1}}\le\frac{2\beta_3\cdot4^{\,j-1}T^{2(a_1+\dots+a_{j-1})}}{T^{(2j-1)a_j}}=\frac{2\beta_3\,4^{\,j-1}}{T^{(2j-2)a_j+1}}\ \le\ \frac{2\beta_3\,4^{\,j-1}}T,\]
using $2(a_1+\dots+a_{j-1})=3^{j-1}-1=a_j-1$. Hence $D\le\sum_jd_js_j^2\to0$, $|G|\le\sum_jd_js_j^2\to0$, so $p\to\frac34$ and $q\to\frac14$ are positive for large $T$, and
\[\E_P\theta^2=\frac p{16}+\frac{9q}{16}+\sum_jd_js_j^2\ \longrightarrow\ \frac3{64}+\frac9{64}=\frac3{16}.\]
Choosing $T$ large enough gives $\E_P\theta^2\le\frac3{16}+\epsilon$. The solution is rational because the coefficients of~\eqref{eq:system} are.
\end{proof}

\subsection{Soundness of the Hessian-augmented verifier}

\begin{proposition}[Soundness]\label{prop:sound}
Let $k\in\{0,1,2\}$. If $W$ passes tests \textup{(a)}, \textup{(b)} and, for $k\ge1$, \textup{(d)}, and for $k=2$ also \textup{(e)} of $\mathcal V^{(k)}_{\tau,\sigma,\eta}$, and $\dist(t',\Lat')\le r_k$ with $r_k$ as in~\eqref{eq:rk}, then $f_W(t')\ge\tau$, so test \textup{(c)} fails. In particular $\mathcal V^{(k)}$ rejects every YES instance.
\end{proposition}

\begin{proof}
Suppose $f_W(t')<\tau$. By the proof of Proposition~\ref{prop:sound-general}, the empirical phase distribution has $\E\theta^2\le\sigma^2\norm e^2$, $\E\cos2\pi\theta<\tau$ and $|\E J_m|\le\eta_m\norm e^m/(2\pi)^m$. Integrating the envelope $(v,a,b)$ of Lemma~\ref{lem:env01} or~\ref{lem:env2} against it, as in Theorem~\ref{thm:phase-duality}(i), and using $a>0$, gives $\sigma^2\norm e^2>v-a\tau-\sum_m|b_m|\eta_m\norm e^m/(2\pi)^m$, that is,
\[\sigma^2\norm e^2>\frac1{16}-\frac\tau{4\pi},\qquad\sigma^2\norm e^2+\frac{\eta_1}{2\pi^2}\norm e>\frac{1-\tau}8,\qquad\Big(\sigma^2+\frac{\eta_2}{2\pi^3}\Big)\norm e^2+\frac{\eta_1}{4\pi}\norm e>\frac3{16}-\frac{5\tau}{8\pi}\]
for $k=0,1,2$. In each case the left side is increasing in $\norm e$ and equals the right side at $\norm e=r_k$ by~\eqref{eq:rk}, so $\dist(t',\Lat')=\norm e>r_k$, contrary to the hypothesis.
\end{proof}

\subsection{Realization by lattice certificates}

\begin{lemma}[Realization]\label{lem:realize}
Let $m\ge2$, let $P$ be a probability distribution with finite support in $\frac1m\Z$ and rational masses $p_x=a_x/D$ ($a_x,D$ positive integers), and let $n\ge1$, $\sigma>0$ and $r>0$ with $\E_P\theta^2\le\sigma^2r^2$; for $n\ge2$ let $0<s\le\sigma\sqrt{n-1}$. Then there are a lattice $\Lat\subseteq\R^n$, a target $t$ with $\dist(t,\Lat)=r$, and a list $W$ of $N=2(n-1)D$ vectors of $\Lat^*$ (of $N=D$ vectors if $n=1$) such that
\begin{enumerate}[label=\textup{(\roman*)}]
\item $\frac1N\sum_iw_iw_i^{\top}\preceq\sigma^2I$;
\item $f_W(jt)=\E_P\cos2\pi j\theta$ for every integer $j$;
\item for every $M\ge1$, the symmetric tensor $\nabla^Mf_W(t)$ vanishes provided $\E_PJ_a=0$ for all $1\le a\le M$ with $a\equiv M\pmod2$ and, when $M$ is even, $\E_P\cos2\pi\theta=0$;
\item \textup{(scale invariance)} $\Lat$ has the basis $B=\diag(mr,1/s,\dots,1/s)$, $t=re_1$, and the integer coefficient vectors $z=B^{\top}w$ of the vectors $w\in W$ do not depend on $r$ and $s$. Consequently, for $\lambda>0$ the lattice $\lambda\Lat$, the target $\lambda t$ and the list $\lambda^{-1}W$ are those of the construction with parameters $\lambda r$ and $s/\lambda$, with the same integer list.
\end{enumerate}
In particular, if $\E_P\cos2\pi\theta=0$ and $\E_PJ_a=0$ for $1\le a\le k$, then $f_W(t)=0$ and $\nabla^Mf_W(t)=0$ for $1\le M\le k$.
\end{lemma}

\begin{proof}
Let
\[\Lat^*:=\frac1{mr}\Z e_1\oplus s\Z^{n-1},\qquad\Lat=mr\Z e_1\oplus\frac1s\Z^{n-1},\qquad t:=re_1\]
(for $n=1$ the second summands are absent). A point $(mrj,u)$ of $\Lat$ has distance $\big((mj-1)^2r^2+\norm u^2\big)^{1/2}\ge r$ from $t$, with equality at $0$, so $\dist(t,\Lat)=r$; the target has order $m$ in $\R^n/\Lat$. For $n\ge2$ let $W$ consist of the vectors
\[\frac xr\,e_1+\varepsilon s\,e_j,\qquad x\in\operatorname{supp}P,\ j\in\{2,\dots,n\},\ \varepsilon\in\{\pm1\},\]
each taken with multiplicity $a_x$, so $N=2(n-1)D$; for $n=1$ let $W$ consist of the numbers $x/r$ with multiplicity $a_x$, $N=D$. Since $x\in\frac1m\Z$, every vector lies in $\Lat^*$. The $\varepsilon$-symmetry kills the cross terms $e_1e_j^{\top}$, each $e_j$ occurs in $2D$ of the vectors with coefficient $\pm s$, and the $e_1$-coefficients are the $x/r$ with frequencies $p_x$; so
\[\frac1N\sum_iw_iw_i^{\top}=\frac{\E_P\theta^2}{r^2}\,e_1e_1^{\top}+\frac{s^2}{n-1}\sum_{j\ge2}e_je_j^{\top}\ \preceq\ \sigma^2I,\]
which is (i). Every vector has $\langle w,t\rangle=x$, so $\langle w,jt\rangle=jx$ and (ii) follows. For (iii), $\nabla^Mf_W(t)=(2\pi)^M\frac1N\sum_i\cos(2\pi x_i+M\pi/2)\,w_i^{\otimes M}$, where $x_i$ is the $x$ of $w_i$. A vector $\frac xre_1+\varepsilon se_j$ has nonzero coordinates only along $e_1$ and $e_j$, so the only nonzero components of $w_i^{\otimes M}$ are those indexed by $a$ copies of $e_1$ and $b=M-a$ copies of a single $e_j$, with value $(x/r)^a(\varepsilon s)^b$. Summing over $\varepsilon$ kills the components with $b$ odd, and for $b$ even the component of $\nabla^Mf_W(t)$ equals
\[(2\pi)^M\,\frac{2s^b}{2(n-1)\,r^a}\sum_xp_x\,x^a\cos\big(2\pi x+\tfrac{M\pi}2\big)=\pm\frac{(2\pi)^Ms^b}{(n-1)r^a}\,\E_P\big[\theta^a\cos(2\pi\theta+\tfrac{a\pi}2)\big],\]
because $a\equiv M\pmod2$ makes $\cos(2\pi\theta+M\pi/2)=\pm\cos(2\pi\theta+a\pi/2)$; for $b=0$ the sum over $j$ gives the same expression with $(n-1)r^a$ replaced by $r^a$. By~\eqref{eq:Jm} the expectation is $\pm\E_PJ_a$ for $a\ge1$ and $\pm\E_P\cos2\pi\theta$ for $a=0$ (which occurs only for $M$ even). This proves (iii); the case $n=1$ is the same with $b=0$ only. For (iv), $B^{\top}(\frac xre_1+\varepsilon se_j)=mx\,e_1+\varepsilon e_j$, an integer vector since $x\in\frac1m\Z$ (and $z=mx$ if $n=1$); multiplying $\Lat$ by $\lambda$ replaces $mr$ by $m\lambda r$ and $1/s$ by $1/(s/\lambda)$.
\end{proof}

\begin{theorem}[Exactness of the close radii]\label{thm:exact}
Let $\sigma>0$, $k\ge0$, $n\ge1$ and $r>\sigma^{-1}\sqrt{v_k}$ (with $v_0=\frac1{16}$, $v_1=\frac18$, $v_k=\frac3{16}$ for $k\ge2$; for $k\le2$ also $r=\sigma^{-1}\sqrt{v_k}$ is allowed). There are a lattice $\Lat\subseteq\R^n$, a target $t$ with $\dist(t,\Lat)=r$, and a list $W$ of vectors of $\Lat^*$ with
\[\frac1N\sum_iw_iw_i^{\top}\preceq\sigma^2I,\qquad f_W(t)=0,\qquad\nabla^Mf_W(t)=0\ \ (1\le M\le k).\]
Consequently, let $\mathcal V$ be a verifier that, on a rational input $(B,t,\rho)$ in dimension $n$, rescales the instance by $r_{\mathcal V}/\rho$ for some $r_{\mathcal V}>0$, reads a certificate of integer vectors of any length $N$, and applies tests \textup{(a)}, \textup{(b)} with bound $\sigma^2N$, \textup{(c)} with a threshold $\tau>0$, and any tests on the derivatives of $f_W$ at $t'$ of orders $1,\dots,k$ that are passed when those derivatives vanish. If $\mathcal V$ is sound, then $r_{\mathcal V}\le\sigma^{-1}\sqrt{v_k}$. In particular, with Proposition~\ref{prop:sound}, the close radius of $\mathcal V^{(k)}_{\tau,\sigma,\eta}$ tends to exactly $\sigma^{-1}\sqrt{v_k}$ as $\tau,\eta\to0$, for $k=0,1,2$, and no test on the derivatives of $f_W$ at the target, of any order, can raise it above $\sigma^{-1}\sqrt{3/16}$.
\end{theorem}

\begin{proof}
Apply Lemma~\ref{lem:realize} to the distributions of Theorem~\ref{thm:close}: for $k=0$ to $\frac12\delta_{1/4}+\frac12\delta_{-1/4}$ with $m=4$; for $k=1$ to $\frac12\delta_0+\frac12\delta_{1/2}$ with $m=2$; for $k=2$ to $\frac34\delta_{1/4}+\frac14\delta_{3/4}$ with $m=4$; and for $k\ge3$ to the distribution of Lemma~\ref{lem:faratoms} with $2m'+1\ge k$ and $\epsilon=\sigma^2r^2-\frac3{16}>0$, again with $m=4$. In each case $\E_P\theta^2\le\sigma^2r^2$, $\E_P\cos2\pi\theta=0$, and $\E_PJ_a=0$ for $1\le a\le k$, so the lemma gives (i), $f_W(t)=0$ and $\nabla^Mf_W(t)=0$ for $M\le k$.

For the consequence, suppose $r_{\mathcal V}>\sigma^{-1}\sqrt{v_k}$, and let $P$ be the distribution above with $r=r_{\mathcal V}$, so that $\E_P\theta^2\le\sigma^2r_{\mathcal V}^2$. For $n\ge2$ choose a rational $s$ with $0<s\le r_{\mathcal V}\sigma\sqrt{n-1}$. Let $B=\diag(m,1/s,\dots,1/s)$. The rational instance $(B,e_1,1)$ has $\dist(e_1,\Lat(B))=1$ (as in the proof of Lemma~\ref{lem:realize}), so it is a YES instance. The verifier rescales it by $r_{\mathcal V}$, obtaining $r_{\mathcal V}\Lat(B)=\Lat(\diag(mr_{\mathcal V},r_{\mathcal V}/s,\dots))$ and $r_{\mathcal V}e_1$: by Lemma~\ref{lem:realize}(iv) this is the construction of Lemma~\ref{lem:realize} at distance $r_{\mathcal V}$ with parameter $s/r_{\mathcal V}\le\sigma\sqrt{n-1}$, whose hypotheses hold, and its list $W$ is represented by the integer list $Z$ of the lemma, which the prover submits. By the first part, $W$ passes (a), (b) and the derivative tests and has $f_W(t')=0<\tau$, so $\mathcal V$ accepts a YES instance, contradicting soundness.
\end{proof}

The certificate in this proof has $N=2(n-1)D$ vectors. A verifier may prescribe $N$; for $k\le2$ and $n\ge2$ the conclusion persists for every verifier with the tests and the certificate length $N$ of $\mathcal V^{(k)}_{\tau,\sigma,\eta}$, whatever its rescaling $r_{\mathcal V}$, whenever $N\tau>6n$ and, for $k\ge1$, $N\eta_1r_k\ge4\pi n$ (note $r_{\mathcal V}>\sigma^{-1}\sqrt{v_k}>r_k$ in the proof); this holds for all the verifiers of this paper. Indeed, take $N'=\lfloor N/2(n-1)\rfloor$, give the vectors of Lemma~\ref{lem:realize} the multiplicities $\lfloor p_xN'\rfloor$ in place of $a_x$, and fill the list with fewer than $2(n-1)(|\operatorname{supp}P|+1)\le6(n-1)$ zero vectors. The spectral bound and, for $k=2$, the vanishing of the Hessian persist, since $\cos2\pi x=0$ on the support; the zero vectors give $0\le f_W(t)<6(n-1)/N$; and the rounding errors in $\E_PJ_1=0$ give $\norm{\nabla f_W(t)}\le4\pi(n-1)\sum_x|x|/(Nr)\le4\pi(n-1)/(Nr)$, with $\nabla f_W(t)=0$ for $k=1$ because $J_1$ vanishes on $\{0,\frac12\}$. The integer list is still independent of $r$ and $s$. For $k\ge3$ the vectors realizing the far atoms of Lemma~\ref{lem:faratoms} have norms of order $T^{3^{m'-1}}/r$, independent of $n$, so a verifier that also checks the norm bound $\norm w\le n^2$ of Section~\ref{sec:farthresholds} accepts them only when $n$ is large in terms of $k$ and $r$.

\begin{remark}[Certificates of other instances, and line tests]\label{rem:otherinstances}
The certificates of Theorem~\ref{thm:exact} live at targets of order $4$, $2$ and $4$, and their phases are the phases that a Gaussian certificate for such a target would have (at a target $t$ with $mt\in\Lat$, every $w\in\Lat^*$ has $\langle w,t\rangle\in\frac1m\Z$), rearranged so that $f_W(t)=0$. All the order-$4$ certificates, including those realizing the far atoms for $k\ge3$, have $f_W(2t)=\E\cos4\pi\theta=-1$, since their phases are odd multiples of $\frac14$, so they are caught by the \emph{line test} $f_W(2t')\ge-\epsilon$. Line tests are linear constraints $\E\cos2\pi j\theta\ge-\epsilon_j$ on the phase distribution as well, and Proposition~\ref{prop:linetests} settles their effect on $\mathcal V^{(0)}$ and $\mathcal V^{(1)}$: they raise $v_0=\frac1{16}$ to exactly $\frac2{27}$ and leave $v_1=\frac18$ unchanged. Adding line tests $f_W(jt')\ge-\epsilon_j$, $2\le j\le J$, to the Hessian-augmented verifier raises the value of the phase program, by our numerical solutions on grids, from $0.1875$ to $0.1958$ for $J=2$ and $J=3$ (attained near the phases $0,\pm0.3576,\pm0.8046$), to $0.1962$ for $4\le J\le6$ (near the phases $0,\pm0.3195,\pm0.3949,\pm0.8079$), and further, in the fifth decimal, as $J$ grows ($0.196152$ for $4\le J\le6$, $0.196162$ at $J=15$, $0.196165$ at $J=30$ and $J=50$, with the extremal phases spreading over $[0.30,0.41]$). This would lower the constant of Theorem~\ref{thm:main} from $0.3676$ to $0.3597$ for $J=2$ and to $0.3593$ for large $J$. We have no closed form for these values and no proof, and we do not pursue them. Table~\ref{tab:hierarchy} summarizes the close side.
\end{remark}

\begin{proposition}[Line tests]\label{prop:linetests}
For $\theta\in\R$ let $\tilde\theta\in(-\frac12,\frac12]$ be its representative modulo $1$, so that $\tilde\theta^2\le\theta^2$ and $\cos2\pi j\theta=\cos2\pi j\tilde\theta$.
\begin{enumerate}[label=\textup{(\roman*)}]
\item For every real $\theta$,
\begin{equation}\label{eq:env227}
\begin{aligned}
\tilde\theta^2&\ \ge\ \frac2{27}-\lambda_1\cos2\pi\theta+\lambda_2\cos4\pi\theta,\\
\lambda_1&=\frac4{81}+\frac{2\sqrt3}{27\pi}=0.0902,\qquad
\lambda_2=\frac{2\sqrt3}{27\pi}-\frac2{81}=0.0161 .
\end{aligned}
\end{equation}
Consequently $\inf\{\E\theta^2:\E\cos2\pi\theta\le0,\ \E\cos4\pi\theta\ge0\}=\frac2{27}$, attained at $\frac13\delta_0+\frac13\delta_{1/3}+\frac13\delta_{-1/3}$, and this distribution has $\E\cos2\pi j\theta\in\{0,1\}$ for every integer $j$; so the value of the phase program with $k=0$ and \emph{all} line tests $\E\cos2\pi j\theta\ge0$, $j\ge2$, is exactly $\frac2{27}$.
\item The distribution $\frac12\delta_0+\frac12\delta_{1/2}$ has $\E\cos2\pi j\theta\in\{0,1\}$ for every integer $j$; so the value of the phase program with $k=1$ and all line tests is still $\frac18$.
\item \textup{(Tolerances)} If $W$ passes tests \textup{(a)}, \textup{(b)} and $f_W(2t')\ge-\epsilon_2$, and $f_W(t')<\tau$, then
\[\sigma^2\dist(t',\Lat')^2\ >\ \frac2{27}-\lambda_1\tau-\lambda_2\epsilon_2,\]
so this verifier has the close radius $\sigma^{-1}(\frac2{27}-\lambda_1\tau-\lambda_2\epsilon_2)^{1/2}$, which tends to $\sigma^{-1}\sqrt{2/27}$. Gaussian certificates pass each line test $f_W(jt')\ge-\epsilon_j$ except with probability at most $2e^{-N\epsilon_j^2/2}$.
\end{enumerate}
In cases \textup{(i)} and \textup{(ii)} Lemma~\ref{lem:realize} realizes the extremal distribution by a lattice certificate with $f_W(jt)=\E\cos2\pi j\theta\ge0$ for all $j$, so, as in Theorem~\ref{thm:exact}, the corresponding close radii $\sigma^{-1}\sqrt{2/27}$ and $\sigma^{-1}/\sqrt8$ are exact for the verifiers with line tests.
\end{proposition}

\begin{proof}
(i) Both sides of~\eqref{eq:env227} depend only on $\tilde\theta$, so it suffices to show $h_3(\theta):=\theta^2-\frac2{27}+\lambda_1\cos2\pi\theta-\lambda_2\cos4\pi\theta\ge0$ on $[-\frac12,\frac12]$. The constants are forced by $h_3(0)=0$ (i.e. $\lambda_1-\lambda_2=\frac2{27}$) and $h_3(\frac13)=h_3'(\frac13)=0$. The function $h_3$ is even with $h_3(0)=h_3(\frac13)=h_3'(\frac13)=0$,
\[\begin{aligned}
h_3''(0)&=2-4\pi^2\lambda_1+16\pi^2\lambda_2=0.9881>0,\\
h_3''(\tfrac13)&=2+2\pi^2\lambda_1-8\pi^2\lambda_2=2.5059>0,
\qquad h_3'(\tfrac12)=1 .
\end{aligned}\]
Since $h_3''(\theta)=2-4\pi^2\lambda_1\cos2\pi\theta+16\pi^2\lambda_2\cos4\pi\theta$ is a quadratic polynomial in $\cos2\pi\theta$, it has at most two zeros on $(0,\frac12)$, so $h_3'$ has at most three zeros on $[0,\frac12]$; Rolle's theorem gives one in $(0,\frac13)$, and $0$ and $\frac13$ are two more. Hence $h_3'$ has no other zero on $[0,\frac12]$, is positive just right of $0$ and of $\frac13$ (as $h_3''>0$ there) and negative just left of $\frac13$, so $h_3$ increases from $0$, decreases to $h_3(\frac13)=0$, and increases on $[\frac13,\frac12]$: $h_3\ge0$. Taking expectations in~\eqref{eq:env227} with $\E\cos2\pi\theta\le0$ and $\E\cos4\pi\theta\ge0$ gives $\E\theta^2\ge\E\tilde\theta^2\ge\frac2{27}$, since $\lambda_1,\lambda_2>0$. The three-point distribution has $\E\theta^2=\frac23\cdot\frac19=\frac2{27}$ and $\E\cos2\pi j\theta=\frac13(1+2\cos\frac{2\pi j}3)$, which is $1$ if $3\mid j$ and $0$ otherwise. (ii) $\E\cos2\pi j\theta=\frac12(1+(-1)^j)$. (iii) For the empirical phase distribution, $f_W(2t')=\E\cos4\pi\theta$ by Lemma~\ref{lem:phase}, so integrating~\eqref{eq:env227} gives $\sigma^2\norm e^2\ge\E\theta^2\ge\E\tilde\theta^2>\frac2{27}-\lambda_1\tau-\lambda_2\epsilon_2$, as in Proposition~\ref{prop:sound}. For Gaussian certificates $\E f_W(jt')=f_{\Lat'}(jt')\ge0$ by~\eqref{eq:f} and Lemma~\ref{lem:facts}(i), and Lemma~\ref{lem:hoeffding} gives the bound. The last statement is Lemma~\ref{lem:realize}(ii) with $m=3$, respectively $m=2$.
\end{proof}

\begin{table}[ht]
\centering\scriptsize
\begin{tabular}{@{}lccccc@{}}
\toprule
tests beyond (a),(b),(c) & value $v$ & extremal phases & radius $\sigma^{-1}\sqrt v$ & at $\sigma^2=\frac1{2\pi}$ & constant\\
\midrule
none & $\frac1{16}$ & $\pm\frac14$ & $\sigma^{-1}/4$ & $0.6267$ & $\frac2\pi=0.6366$\\
line tests & $\frac2{27}$ & $0,\pm\frac13$ & $\sigma^{-1}\sqrt{2/27}$ & $0.6822$ & $\frac{3\sqrt6}{4\pi}=0.5848$\\
gradient ($\pm$ line tests) & $\frac18$ & $0,\frac12$ & $\sigma^{-1}/\sqrt8$ & $0.8862$ & $\frac{\sqrt2}\pi=0.4502$\\
gradient, Hessian & $\frac3{16}$ & $\frac14,\frac34$ (masses $\frac34,\frac14$) & $\sigma^{-1}\sqrt{3/16}$ & $1.0854$ & $\frac{2}{\pi\sqrt3}=0.3676$\\
all derivatives & $\frac3{16}$ & not attained & $\sigma^{-1}\sqrt{3/16}$ & $1.0854$ & $\frac{2}{\pi\sqrt3}=0.3676$\\
gradient, Hessian, line tests & $0.1962^{*}$ & $0,\pm[0.30,0.41],\pm0.81$ & $\sigma^{-1}\sqrt{0.1962}^{\,*}$ & $1.1102^{*}$ & $0.3593^{*}$\\
\bottomrule
\end{tabular}
\caption{The close side. Each row is the value of the phase program~\eqref{eq:vk} with the indicated constraints, the extremal phase distribution, the close radius, its value at the Gaussian spectral level $\sigma^2=1/2\pi$, and the constant $(1/\sqrt{2\pi})/\text{radius}$ it yields with the Gaussian prover (far threshold $1/\sqrt{2\pi}$). The line tests are $f_W(jt')\ge-\epsilon_j$ for $2\le j\le J$; in the last row $J$ is large (Remark~\ref{rem:otherinstances}). Starred entries are numerical; all others are proved (Theorem~\ref{thm:close}, Theorem~\ref{thm:exact} and Proposition~\ref{prop:linetests}).}\label{tab:hierarchy}
\end{table}

\section{Completeness and the proof of the main theorem}\label{sec:complete}

Throughout this section $w_1,\dots,w_N$ are independent samples from $D_\Lambda$, $\Lambda:=\Lat'^*$, and $\dist(t',\Lat')\ge c''\sqrt n$. We first treat the threshold, gradient and spectral tests.

\begin{lemma}[Threshold, gradient and spectral tests for Gaussian certificates]\label{lem:EBcomplete}
Let $n\ge8$, $\epsilon_n=2\beta(1)^n$.
\begin{enumerate}[label=\textup{(\roman*)}]
\item If $c''\ge1/\sqrt{2\pi}$, then for every $\zeta>0$, $\PP\big(f_W(t')\ge\beta(c'')^n+\zeta\big)\le2e^{-N\zeta^2/2}$.
\item If $c''\sqrt{1-1/n}\ge1/\sqrt{2\pi}$, then $\norm{\nabla f_{\Lat'}(t')}\le\gamma_n:=\sqrt{2\pi n/e}\;e^{\pi c''^2}\beta(c'')^n$, and for every $a>0$, $\PP\big(\norm{\nabla f_W(t')}>2\pi a+\gamma_n\big)\le n/2\pi Na^2$.
\item For $\delta\in(\epsilon_n,1]$ and $\delta':=\delta-\epsilon_n$, the event $\mathcal E:=\{\norm{w_i}\le\sqrt n\ \forall i\}$ has $\PP(\mathcal E^c)\le N\beta(1)^n$, and
\[\PP\Big(\mathcal E\ \text{and}\ s_{\max}\Big(\sum_iw_iw_i^{\top}\Big)>\frac{1+\delta}{2\pi}N\Big)\le n\exp\Big(-\frac{N\delta'^2}{4\pi n(1+\epsilon_n+\delta'/3)}\Big).\]
\end{enumerate}
\end{lemma}

\begin{proof}
(i) By~\eqref{eq:f} and Lemma~\ref{lem:facts}(iv), $\E\cos(2\pi\langle w_i,t'\rangle)=f_{\Lat'}(t')\le\beta(c'')^n$; apply Lemma~\ref{lem:hoeffding}.

(ii) By~\eqref{eq:fderiv} and~\eqref{eq:peel}, $\norm{\nabla f_{\Lat'}(t')}\le2\pi\rho(\Lat')^{-1}\sum_{y\in\Lat'+t'}\norm y e^{-\pi\norm y^2}\le2\pi\sqrt{n/2\pi e}\;\rho_{s'}(\Lat'+t')/\rho(\Lat')$, and~\eqref{eq:rhosprime} gives $\gamma_n$. For the fluctuation put $X_i=\sin(2\pi\langle w_i,t'\rangle)w_i$, so that $\nabla f_W(t')=-\frac{2\pi}N\sum_iX_i$ and the $X_i$ are independent with $\E\norm{X_i}^2\le\E\norm{w_i}^2\le n/2\pi$ by Lemma~\ref{lem:transference}. Hence $\E\norm{\frac1N\sum_i(X_i-\E X_i)}^2\le n/2\pi N$, Markov's inequality gives $\PP(\norm{\frac1N\sum_i(X_i-\E X_i)}>a)\le n/2\pi Na^2$, and on the complementary event $\norm{\nabla f_W(t')}\le2\pi a+\norm{\nabla f_{\Lat'}(t')}\le2\pi a+\gamma_n$.

(iii) By Lemma~\ref{lem:facts}(iv) with $u=0$, $c=1$, $\PP(\norm w>\sqrt n)\le\beta(1)^n$ for one sample, so $\PP(\mathcal E^c)\le N\beta(1)^n$. Let $\tilde w$ be a sample conditioned on $\norm{\tilde w}\le\sqrt n$ and $\Sigma':=\E\tilde w\tilde w^{\top}\preceq\frac{1+\epsilon_n}{2\pi}I$ (Lemma~\ref{lem:trunc}). Put $X_i=\frac1N(\tilde w_i\tilde w_i^{\top}-\Sigma')$: these are independent, centered, $\norm{X_i}\le\max(\norm{\tilde w_i}^2,\norm{\Sigma'})/N\le n/N=:R$, and
\[\sum_i\E X_i^2\preceq\frac1N\E\big[(\tilde w\tilde w^{\top})^2\big]=\frac1N\E\big[\norm{\tilde w}^2\tilde w\tilde w^{\top}\big]\preceq\frac nN\Sigma'\preceq\frac{(1+\epsilon_n)n}{2\pi N}I,\]
so $v=(1+\epsilon_n)n/2\pi N$ is admissible in Lemma~\ref{lem:mb}. With $u=\delta'/2\pi$,
\[\PP\Big(s_{\max}\Big(\sum_iX_i\Big)\ge\frac{\delta'}{2\pi}\Big)\le n\exp\Big(\frac{-\delta'^2/8\pi^2}{\frac{(1+\epsilon_n)n}{2\pi N}+\frac{n\delta'}{6\pi N}}\Big)=n\exp\Big(-\frac{N\delta'^2}{4\pi n(1+\epsilon_n+\delta'/3)}\Big),\]
and on the complementary event $s_{\max}(\frac1N\sum_i\tilde w_i\tilde w_i^{\top})\le\norm{\Sigma'}+\delta'/2\pi\le(1+\epsilon_n+\delta')/2\pi=(1+\delta)/2\pi$; on $\mathcal E$ the conditioned samples are the actual samples.
\end{proof}

Aharonov and Regev prove the spectral statement with the bound $3N$ in place of $(1+\delta)N/2\pi$, by truncation, Hoeffding's inequality in each direction of an $\epsilon$-net of the sphere, and a union bound; the factor $3$ is a Hoeffding slack of $2$ plus a net slack of $1$. The matrix Bernstein inequality removes both slacks, and nothing below $1/2\pi$ is possible (Lemma~\ref{lem:transference}).

The new ingredient is the Hessian test.

\begin{lemma}[The Hessian test for Gaussian certificates]\label{lem:hessian}
Let $n\ge8$ and $c''\sqrt{1-1/n}\ge1/\sqrt{2\pi}$. Then
\begin{equation}\label{eq:gamma2}
\norm{\nabla^2f_{\Lat'}(t')}\ \le\ \gamma^{(2)}_n:=\Big(\frac{4\pi n}e+2\pi\Big)e^{\pi c''^2}\beta(c'')^n,
\end{equation}
and for every $u>0$, with $\kappa_n=6n\beta(1)^n$, $R=2n/N$ and $v=(1+\epsilon_n)n/2\pi N$,
\[\PP\Big(\norm{\nabla^2f_W(t')}>4\pi^2(u+\kappa_n)+\gamma^{(2)}_n\Big)\ \le\ 2n\exp\Big(-\frac{u^2/2}{v+Ru/3}\Big)+N\beta(1)^n .\]
\end{lemma}

\begin{proof}
By~\eqref{eq:fderiv}, $\norm{\nabla^2f_{\Lat'}(t')}\le\rho(\Lat')^{-1}\sum_{y\in\Lat'+t'}(4\pi^2\norm y^2+2\pi)e^{-\pi\norm y^2}$, and by~\eqref{eq:peel} each term is at most $(4\pi^2\cdot\frac n{\pi e}+2\pi)\rho_{s'}(y)$; summing and using~\eqref{eq:rhosprime} gives~\eqref{eq:gamma2}. For the fluctuation, let $\mathcal E$ be the event that $\norm{w_i}\le\sqrt n$ for all $i$, so $\PP(\mathcal E^c)\le N\beta(1)^n$ by Lemma~\ref{lem:facts}(iv), and let $\tilde w_i$ be the samples conditioned on $\norm{\tilde w_i}\le\sqrt n$; on $\mathcal E$ they coincide with the $w_i$. Write $c(w)=\cos(2\pi\langle w,t'\rangle)$ and put
\[X_i:=\frac1N\Big(c(\tilde w_i)\tilde w_i\tilde w_i^{\top}-\E\big[c(\tilde w)\tilde w\tilde w^{\top}\big]\Big).\]
These are independent, centered, and $\norm{X_i}\le\frac1N(\norm{\tilde w_i}^2+\E\norm{\tilde w}^2)\le2n/N=R$. Since the variance of a random matrix is dominated by its second moment,
\[\sum_i\E X_i^2\preceq\frac1N\E\big[c(\tilde w)^2\norm{\tilde w}^2\tilde w\tilde w^{\top}\big]\preceq\frac nN\E[\tilde w\tilde w^{\top}]\preceq\frac{(1+\epsilon_n)n}{2\pi N}I\]
by Lemma~\ref{lem:trunc}, so $v$ is admissible in Lemma~\ref{lem:mb}, whose two-sided form gives $\PP(\norm{\sum_iX_i}\ge u)\le2n\exp(-\frac{u^2/2}{v+Ru/3})$. On $\mathcal E$ and the complementary event,
\[\Big\lVert\frac1N\sum_ic(w_i)w_iw_i^{\top}\Big\rVert\le u+\big\lVert\E[c(\tilde w)\tilde w\tilde w^{\top}]\big\rVert\le u+\kappa_n+\big\lVert\E[c(w)ww^{\top}]\big\rVert=u+\kappa_n+\frac{\norm{\nabla^2f_{\Lat'}(t')}}{4\pi^2},\]
by Lemma~\ref{lem:trunc} and~\eqref{eq:fderiv}, and $\nabla^2f_W(t')=-4\pi^2\cdot\frac1N\sum_ic(w_i)w_iw_i^{\top}$ by~\eqref{eq:fW}.
\end{proof}

\begin{proof}[Proof of Theorem~\ref{thm:main}]
Fix $c>2/(\pi\sqrt3)$. Let
\[\tau=\frac1n,\qquad\delta=\frac{0.4}{\sqrt n},\qquad\sigma^2=\frac{1+\delta}{2\pi},\qquad\eta_1=\frac1{\sqrt n},\qquad\eta_2=\frac{4.5}{\sqrt n},\qquad N=100n^2\lceil\log_2n\rceil,\]
and run $\mathcal V^{(2)}_{\tau,\sigma,\eta}$. Its close radius $r_2(n)$ of~\eqref{eq:rk} tends to $\sqrt{3\pi/8}=1.0854$ as $n\to\infty$ ($r_2(8)=0.8204$, $r_2(9)=0.8367$, $r_2(100)=1.0134$, $r_2(10^4)=1.0783$). \emph{Soundness} is Proposition~\ref{prop:sound}. \emph{Completeness.} We prove slightly more, for use in Appendix~\ref{app:precision}: let the verifier rescale by $q/r$ for any $q\in(r_2(n)-n^{-4},r_2(n)]$ in place of $r_2(n)/r$; we show that some certificate passes tests (b)--(e) with margin at least $n^{-4}$ each. On a NO instance,
\[\dist(t',\Lat')\ >\ cq\sqrt n\ \ge\ c\,\big(r_2(n)-n^{-4}\big)\sqrt n\ =:\ c''\sqrt n,\]
and $c''\to c\sqrt{3\pi/8}>\frac2{\pi\sqrt3}\sqrt{3\pi/8}=1/\sqrt{2\pi}$.
Let $n_0(c)$ be the least $n\ge8$ such that, for all $n'\ge n$,
\begin{equation}\label{eq:n0conds}
c''\sqrt{1-1/n}>\frac1{\sqrt{2\pi}},\qquad\beta(c'')^n\le\frac\tau2,\qquad\gamma_n\le\frac{0.49}{\sqrt n},\qquad\gamma^{(2)}_n+4\pi^2\kappa_n\le\frac{0.55}{\sqrt n},
\end{equation}
(with $n'$ in place of $n$), with $\gamma_n$, $\gamma^{(2)}_n$, $\kappa_n$ as in Lemmas~\ref{lem:EBcomplete}, \ref{lem:hessian} and~\ref{lem:trunc}; such $n_0$ exists because $\beta(c'')<1$ for large $n$, after which $\beta(c'')^n$, $\gamma_n$, $\gamma^{(2)}_n$ and $\kappa_n$ decay exponentially while $\tau/2=1/2n$ and $n^{-1/2}$ do not. For $n\ge n_0$ let $W$ be $N$ independent samples from $D_{\Lat'^*}$. Test (a) holds. All the bounds below are nonincreasing in $n$, and we use $N\ge100n^2\log_2n$.

\emph{Test \textup{(c)}.} By Lemma~\ref{lem:EBcomplete}(i) with $\zeta=\tau/2-n^{-4}$, $f_W(t')<\beta(c'')^n+\tau/2-n^{-4}\le\tau-n^{-4}$ except with probability at most
\[2e^{-N(\tau/2-n^{-4})^2/2}\le2e^{-12.5(1-2n^{-3})^2\log_2n}\le2n^{-17}.\]

\emph{Test \textup{(b)}.} Apply Lemma~\ref{lem:EBcomplete}(iii) with $\delta-2\pi n^{-4}$ in place of $\delta$, so that $\delta'=\delta-\epsilon_n-2\pi n^{-4}\ge0.9891\,\delta$ (as $(\epsilon_n+2\pi n^{-4})/\delta\le0.0109$ for $n\ge8$) and $\delta\le0.4\cdot8^{-1/2}=0.1414$. Except on $\mathcal E^c$, of probability at most $N\beta(1)^n$, and on an event of probability at most
\[n\exp\Big(-\frac{100\cdot0.16\,(0.9891)^2\log_2n}{4\pi(1+\epsilon_n+\delta/3)}\Big)\ \le\ n^{-0.716},\]
we get $s_{\max}(\sum_iw_iw_i^{\top})\le(\sigma^2-n^{-4})N$.

\emph{Test \textup{(d)}.} By Lemma~\ref{lem:EBcomplete}(ii) with $a=0.08\,n^{-1/2}$, except with probability at most
\[\frac n{2\pi Na^2}=\frac{n^2}{0.0128\pi N}\le\frac1{4.0212\log_2n},\]
we get $\norm{\nabla f_W(t')}\le2\pi a+\gamma_n\le(0.5027+0.49)n^{-1/2}\le\eta_1-n^{-4}$.

\emph{Test \textup{(e)}.} Apply Lemma~\ref{lem:hessian} with $u=0.1\,n^{-1/2}$, for which
\[\frac{u^2/2}{v+Ru/3}=\frac{N/200n^2}{(1+\epsilon_n)/2\pi+\frac{0.2}{3}n^{-1/2}}\ \ge\ \frac{\log_2n/2}{0.1592(1+\epsilon_n)+0.0236}\ \ge\ 2.735\log_2n\ \ge\ 3.945\ln n .\]
Except with probability at most $2n^{-2.945}+N\beta(1)^n$, we get $\norm{\nabla^2f_W(t')}\le4\pi^2(u+\kappa_n)+\gamma^{(2)}_n\le(3.9478+0.55)n^{-1/2}\le\eta_2-n^{-4}$, by the last condition in~\eqref{eq:n0conds}.

At $n=8$, where $N=19200$ and $N\beta(1)^8=0.0199$, the failure probabilities are at most $2\cdot8^{-17}$, $0.2256+0.0199$, $0.0829$ and $0.0044+0.0199$, which sum to less than $0.3527$; so for every $n\ge n_0$ some certificate passes tests (b)--(e) with margin at least $n^{-4}$ each, and the verifier accepts. The bit length of the certificate is bounded in Appendix~\ref{app:precision}.

Membership of $\mathrm{GapCVP}$ in $\NP$ is witnessed by a nearby lattice point, and dimensions below $n_0(c)$ are decided by Kannan's algorithm. For $\mathrm{GapSVP}$ we use the index-two sublattices of Goldreich, Micciancio, Safra and Seifert~\cite{GMSS99}: if $B=[b_1,\ldots,b_n]$ and $B^{(j)}$ replaces $b_j$ by $2b_j$, then $\lambda_1(\Lat(B))=\min_j\dist(b_j,\Lat(B^{(j)}))$, because a shortest nonzero vector has an odd coefficient in some basis coordinate $j$ (otherwise it could be halved), and then lies in $b_j-\Lat(B^{(j)})$, while every vector $b_j-v$, $v\in\Lat(B^{(j)})$, is a nonzero vector of $\Lat(B)$. The $\coNP$ verifier for $\mathrm{GapSVP}_{c\sqrt n}$ on $(B,r)$ runs the verifier above on the $n$ instances $(B^{(j)},b_j,r)$, with a certificate of $nN$ integer vectors, and accepts if all of them accept. On a NO instance, $\lambda_1>c\sqrt n\,r$, all $n$ instances are NO instances and their certificates exist; on a YES instance, $\lambda_1\le r$, some instance $(B^{(j)},b_j,r)$ is a YES instance and its verifier rejects. A short nonzero lattice vector is the $\NP$ certificate.
\end{proof}

\begin{table}[ht]
\centering
\begin{tabular}{@{}l*{12}{c}@{}}
\toprule
$c$ & $0.37$ & $0.375$ & $0.38$ & $0.39$ & $0.40$ & $0.42$ & $0.45$ & $0.50$ & $0.60$ & $0.70$ & $1.00$ & $2.00$\\
\midrule
$n_0(c)$ & $666738$ & $62553$ & $20881$ & $5913$ & $2689$ & $967$ & $372$ & $140$ & $46$ & $24$ & $9$ & $8$\\
\bottomrule
\end{tabular}
\caption{The cutoff $n_0(c)$ of Theorem~\ref{thm:main}: the least $n\ge8$ from which on the conditions~\eqref{eq:n0conds} hold, with $c''=c\,(r_2(n)-n^{-4})$ and the parameters $\tau=1/n$, $\delta=0.4n^{-1/2}$, $\eta_1=n^{-1/2}$, $\eta_2=4.5n^{-1/2}$, $N=100n^2\lceil\log_2n\rceil$, evaluated numerically (the conditions were checked for all $n$ up to $3\cdot10^6$, beyond which they hold with a wide margin). It applies to the finite-bit verifier of Appendix~\ref{app:precision}. For every entry except $c=2$ the binding condition is the last one, the size of the Hessian for Gaussian certificates; for $c=2$ it is the requirement $n\ge8$. For $c\le0.5$ the entries are within a factor of $4$ of the least $n$ with $\beta(c\sqrt{3\pi/8})^n\le1/2n$, the requirement of the threshold test alone; the limiting constant $2/(\pi\sqrt3)=0.3676$ does not depend on the parameters.}\label{tab:n0}
\end{table}

\begin{remark}[Where the constants come from]\label{rem:constants}
The constant is the ratio of a far threshold to a close radius. With the Gaussian prover the far threshold is Banaszczyk's $1/\sqrt{2\pi}$, and the close radius at the Gaussian spectral level $\sigma^2=1/2\pi$ is $\sqrt{2\pi v_k}$. The published proof~\cite{AR05} certifies $c=100$: threshold $\tau=\frac12$, spectral bound $3N$, the Taylor inequality $\cos\varphi\ge1-\varphi^2/2$ for soundness, which gives $f_W\ge1-2\pi^2\cdot3\norm e^2\ge\frac12$ for $\norm e\le1/(2\pi\sqrt3)=0.0919$, and a generous rescaling to $1/100$ on the close side and $\sqrt n$ on the far side. Using the close radius $0.0919$ as is and the far side at any $c''>1/\sqrt{2\pi}$ gives $c=\sqrt{6\pi}=4.3416$; placing the spectral bound at its Gaussian expected value $(1+\delta)N/2\pi$ (Lemma~\ref{lem:EBcomplete}(iii)) and letting $\tau\to0$, still with the Taylor inequality, gives the close radius $1/\sqrt\pi$ and $c=1/\sqrt2=0.7071$; the exact second-moment analysis $v_0=\frac1{16}$ gives $c=2/\pi=0.6366$, which by Theorem~\ref{thm:exact} is the end of the road for the original tests; the gradient test, $v_1=\frac18$, gives $c=\sqrt2/\pi=0.4502$; and the Hessian test, $v_2=\frac3{16}$, gives $c=2/(\pi\sqrt3)=0.3676$ (Theorem~\ref{thm:main}). Each of the last three is proved by the argument of Theorem~\ref{thm:main} with the corresponding verifier $\mathcal V^{(k)}$, Lemma~\ref{lem:hessian} being needed only for $k=2$. What remains is the far side.
\end{remark}

\section{The far side}\label{sec:far}

\subsection{The Gaussian prover}

The following threshold governs the Gaussian prover. Let $c_0$ be the infimum of the $c>0$ such that, for every lattice $\Lat\subseteq\R^n$ and every $x\in\R^n$, $\dist(x,\Lat)\ge c\sqrt n$ implies $f_\Lat(x)\le2^{-\Omega(n)}$ (with the implied constant depending on $c$ only). Lemma~\ref{lem:facts}(iv) gives $c_0\le1/\sqrt{2\pi}=0.3989$.

\begin{lemma}[The Gaussian prover passes all tests]\label{lem:gaussianprover}
Let $c>c_0$. There is $\delta=\delta(c)>0$ such that for all sufficiently large $n$, every lattice $\Lat\subseteq\R^n$ and every $t$ with $\dist(t,\Lat)\ge c\sqrt n$, the distribution $P$ of $w\sim D_{\Lat^*}$ conditioned on $\norm w\le n^2$ (so that $P$ is admissible in Definition~\ref{def:far}) satisfies
\[\Sigma(P)\preceq\frac{1+2^{-\delta n}}{2\pi}I,\qquad\E_P\cos(2\pi\langle w,t\rangle)\le2^{-\delta n},\qquad\norm{D_1(P)}\le2^{-\delta n},\qquad\norm{D_2(P)}\le2^{-\delta n}.\]
Consequently $\cfar^{(0)}\le\cfar^{(1)}\le\cfar^{(2)}\le c_0\le1/\sqrt{2\pi}$.
\end{lemma}

\begin{proof}
Pick $c_1\in(c_0,c)$ and $\delta_1>0$ such that, for large $n$, every coset of every lattice in $\R^n$ at distance at least $c_1\sqrt n$ has Gaussian mass at most $2^{-\delta_1n}$ times that of the lattice. Let $Q=D_{\Lat^*}$. Then $\Sigma(Q)\preceq I/2\pi$ by Lemma~\ref{lem:transference}, $\E_Q\cos(2\pi\langle w,t\rangle)=f_\Lat(t)\le2^{-\delta_1n}$ by~\eqref{eq:f}, and by~\eqref{eq:fderiv}, \eqref{eq:peel} and the argument of~\eqref{eq:rhosprime},
\[\norm{D_1(Q)}=\frac{\norm{\nabla f_\Lat(t)}}{2\pi}\le\sqrt{\frac n{2\pi e}}\,\frac{\rho_{s'}(\Lat+t)}{\rho(\Lat)},\qquad
\norm{D_2(Q)}=\frac{\norm{\nabla^2f_\Lat(t)}}{4\pi^2}\le\Big(\frac n{\pi e}+\frac1{2\pi}\Big)\frac{\rho_{s'}(\Lat+t)}{\rho(\Lat)},\]
where $(\Lat+t)/s'$ is a coset of $\Lat/s'$ at distance $\dist(t,\Lat)/s'\ge(c/s')\sqrt n\ge c_1\sqrt n$ for large $n$, so that $\rho_{s'}(\Lat+t)=\rho((\Lat+t)/s')\le2^{-\delta_1n}\rho(\Lat/s')\le2^{-\delta_1n}s'^n\rho(\Lat)\le2\cdot2^{-\delta_1n}\rho(\Lat)$ by Lemma~\ref{lem:facts}(ii). So all four quantities for $Q$ are at most $\mathrm{poly}(n)2^{-\delta_1n}$. Conditioning on $\norm w\le n^2$ changes them by at most $p=Q(\norm w>n^2)$, $\E_Q[\norm w\one_{\norm w>n^2}]$ or $\E_Q[\norm w^2\one_{\norm w>n^2}]$, up to the factor $(1-p)^{-1}$, and by Lemma~\ref{lem:facts}(iv) and Lemma~\ref{lem:trunc} with $R=n^{3/2}$ these are at most $e^{\pi R^2}\beta(R)^n\cdot n=n^{3n/2+1}(2\pi e)^{n/2}e^{-\pi n^3(n-1)}$, which is far smaller than $2^{-\delta_1n}$. Any $\delta<\delta_1$ works for large $n$. The last claim follows with $\tau_n=\eta_{1,n}=\eta_{2,n}=\delta_n=2^{-\delta n}$.
\end{proof}

\subsection{The constant of the verifier family}

\begin{theorem}[The constant is the far threshold over the close radius]\label{thm:constant}
For $k\in\{0,1,2\}$, $\cAR^{(k)}=\cfar^{(k)}/\varrho_k$ with $\varrho_k$ as in~\eqref{eq:rhok}. In particular $\cAR^{(2)}=\sqrt{8/(3\pi)}\;\cfar^{(2)}$, and the proof of Theorem~\ref{thm:main} shows explicitly that $\cAR^{(2)}\le(1/\sqrt{2\pi})/\varrho_2=2/(\pi\sqrt3)$.
\end{theorem}

\begin{proof}
We treat $k=2$; the cases $k=0,1$ are the same with the corresponding tests omitted. Throughout $\sigma^2=1/2\pi$.

\emph{$\cfar^{(2)}\le\varrho_2\,\cAR^{(2)}$.} Let a family $\mathcal V_n$ achieve $c$ and let $c''>c\varrho_2$. Since $c\,r_2(n)\to c\varrho_2<c''$, for large $n$ every $(\Lat,t)$ with $\dist(t,\Lat)\ge c''\sqrt n$ has $\dist(t,\Lat)>c\,r_2(n)\sqrt n$, hence an accepting certificate $W$ of vectors of norm at most $n^2$. Its empirical distribution $P_W$ lives on $\{w\in\Lat^*:\norm w\le n^2\}$ by test (a) and the norm test, and by~\eqref{eq:empirical} tests (b)--(e) say exactly that
\[\Sigma(P_W)\preceq(1+\delta_n)\sigma^2I,\qquad\E_{P_W}\cos(2\pi\langle w,t\rangle)=f_W(t)<\tau_n,\]
\[\norm{D_1(P_W)}\le\frac{\eta_{1,n}}{2\pi},\qquad\norm{D_2(P_W)}\le\frac{\eta_{2,n}}{4\pi^2}.\]
So $c''$ is $2$-far-certifiable, $\cfar^{(2)}\le c\varrho_2$, and the infimum over achieved $c$ gives the claim.

\emph{$\cAR^{(2)}\le\cfar^{(2)}/\varrho_2$.} Let $c''$ be $2$-far-certifiable with slacks $\tau_n,\delta_n,\eta_{1,n},\eta_{2,n}$, and let $c>c''/\varrho_2$. Define the family by
\[\tau'_n=\tau_n+\frac2n,\qquad\delta'_n=\delta_n+\frac{4\pi}n,\qquad\eta'_{1,n}=2\pi\Big(\eta_{1,n}+\frac2n\Big),\qquad\eta'_{2,n}=4\pi^2\Big(\eta_{2,n}+\frac2n\Big),\qquad N_n=n^8 .\]
These slacks tend to $0$, so $r_2(n)\to\varrho_2$ and $c\,r_2(n)\ge c''$ for large $n$. Fix such an $n$ and $(\Lat,t)$ with $\dist(t,\Lat)>c\,r_2(n)\sqrt n\ge c''\sqrt n$, and let $P$ be the distribution provided by certifiability, supported on vectors of norm at most $n^2$. Let $W$ consist of $N=n^8$ independent samples from $P$; test (a) and the norm test hold. For test (c), Lemma~\ref{lem:hoeffding} with $\zeta=1/n$ gives $f_W(t)<\tau_n+1/n<\tau'_n$ except with probability $2e^{-N/2n^2}$. For test (b), apply Lemma~\ref{lem:mb} to $X_i=\frac1N(w_iw_i^{\top}-\Sigma(P))$: here $\norm{X_i}\le2n^4/N$ and
\[\sum_i\E X_i^2\preceq\frac1N\E_P\big[\norm w^2ww^{\top}\big]\preceq\frac{n^4}N\Sigma(P)\preceq\frac{n^4}NI,\]
so with $u=1/n$ the failure probability is at most $n\exp\big(-\frac{1/2n^2}{n^4/N+2n^4/3nN}\big)\le ne^{-N/4n^6}=ne^{-n^2/4}$, and on the complementary event $s_{\max}(\Sigma(P_W))\le(1+\delta_n)\sigma^2+\frac1n\le(1+\delta'_n)\sigma^2$. For test (d), put $X_i=\sin(2\pi\langle w_i,t\rangle)w_i$, so that $\E\norm{X_i}^2\le\E_P\norm w^2\le n$; Markov's inequality gives $\norm{\frac1N\sum_i(X_i-\E X_i)}\le1/n$ except with probability $n^3/N=n^{-5}$, and then $\norm{\nabla f_W(t)}\le2\pi(\eta_{1,n}+1/n)<\eta'_{1,n}$. For test (e), apply Lemma~\ref{lem:mb} to $X_i=\frac1N(\cos(2\pi\langle w_i,t\rangle)w_iw_i^{\top}-D_2(P))$, for which $\norm{X_i}\le2n^4/N$ and $\sum_i\E X_i^2\preceq\frac1N\E_P[\norm w^2ww^{\top}]\preceq\frac{n^4}NI$ as before; with $u=1/n$ the two-sided bound gives failure probability at most $2ne^{-n^2/4}$, and otherwise $\norm{\nabla^2f_W(t)}\le4\pi^2(\eta_{2,n}+1/n)<\eta'_{2,n}$. The four failure probabilities sum to less than $1$ for $n\ge3$, so an accepting certificate exists, consisting of vectors of norm at most $n^2$ and hence of polynomial bit length. Thus the family achieves every $c>c''/\varrho_2$, so $\cAR^{(2)}\le c''/\varrho_2$, and the infimum over $c''$ gives the claim.
\end{proof}

\subsection{Targets of order two}\label{sec:order2}

Let $2t\in\Lat$ and $t\notin\Lat$. Then $\langle w,t\rangle\in\frac12\Z$ for $w\in\Lat^*$, so
\[\cos(2\pi\langle w,t\rangle)=\chi(w):=(-1)^{2\langle w,t\rangle}\in\{\pm1\},\qquad\sin(2\pi\langle w,t\rangle)=0,\]
and $\Lat^*$ splits into the index-two sublattice $E:=\{\chi=1\}=(\Lat+\Z t)^*$ of \emph{even} vectors and its complement $O:=\{\chi=-1\}$ of \emph{odd} vectors. For a feasible $P$ write
\[\Sigma_O(P):=\sum_{w\in O}P(w)\,ww^{\top},\qquad\Sigma_E(P):=\sum_{w\in E}P(w)\,ww^{\top},\]
so that
\[\Sigma(P)=\Sigma_E+\Sigma_O,\qquad D_1(P)=0,\qquad D_2(P)=\Sigma_E-\Sigma_O,\]
and $\E_P\chi=1-2P(O)$. The prover wants mass on odd vectors; the spectral constraint limits how much, and the Hessian constraint forces the even vectors to carry the same second moment as the odd ones.

\begin{definition}[Spectral capacity of the odd coset]\label{def:capacity}
\[\mu(\Lat,t):=\sup\Big\{\nu(O):\ \nu\ \text{a finitely supported measure on $O$ with}\ \nu\ge0,\ \sum_{w\in O}\nu(w)\,ww^{\top}\preceq I\Big\}.\]
\end{definition}

The capacity is finite: if $\ell:=\min\{\norm w:w\in O\}$ then $\nu(O)\ell^2\le\tr\sum\nu(w)ww^{\top}\le n$, so $\mu\le n/\ell^2$. It has a dual description as a trace-minimal ellipsoid.

\begin{lemma}[Capacity duality]\label{lem:capacity}
$\mu(\Lat,t)=\min\{\tr M:\ M\in\Sym,\ M\succeq0,\ w^{\top}Mw\ge1\ \text{for all }w\in O\}$, and the minimum is attained. If $\nu$ attains the supremum and $M$ the minimum, then $\nu$ is supported on $\{w\in O:w^{\top}Mw=1\}$ and $\tr\big(M(I-\sum\nu(w)ww^{\top})\big)=0$.
\end{lemma}

\begin{proof}
Weak duality: for feasible $\nu$ and $M$, $\nu(O)\le\sum_w\nu(w)w^{\top}Mw=\tr(M\sum\nu(w)ww^{\top})\le\tr M$, since $M\succeq0$ and $I-\sum\nu ww^{\top}\succeq0$; the complementary slackness statement is the case of equality in this chain. For strong duality consider in $\R\times\Sym$ the convex sets
\[K:=\Big\{\Big(\nu(O)-r,\ \sum_w\nu(w)ww^{\top}+S\Big):\ \nu\ge0\ \text{finitely supported on }O,\ r\ge0,\ S\succeq0\Big\},\]
\[C_2:=\{(x,B):\ x>\mu,\ B\prec I\}.\]
They are disjoint: a point of $K$ in $C_2$ would come from a $\nu$ with $\sum\nu ww^{\top}\preceq B\prec I$, hence feasible, and $\nu(O)\ge x>\mu$. Since $C_2$ is open, there are $(\alpha,M)\ne0$ and $\gamma$ with $\ell(x,B):=\alpha x+\tr(MB)\le\gamma$ on $K$ and $\ell\ge\gamma$ on $C_2$. Letting $r\to\infty$ in $K$ gives $\alpha\ge0$, and letting $S=\lambda vv^{\top}$, $\lambda\to\infty$, gives $v^{\top}Mv\le0$ for all $v$, i.e. $M\preceq0$. On $C_2$, $\tr(MB)\ge\tr M$ because $M\preceq0$ and $B\prec I$, so $\inf_{C_2}\ell=\alpha\mu+\tr M\ge\gamma$. In $K$ take $\nu=\lambda\delta_w$ with $w\in O$, $\lambda\ge0$, $r=0$, $S=0$: $\lambda(\alpha+w^{\top}Mw)\le\gamma$ for all $\lambda\ge0$, so $\gamma\ge0$ and $\alpha+w^{\top}Mw\le0$ for every $w\in O$. If $\alpha=0$ then $0\le\gamma\le\tr M\le0$, so $\tr M=0$ and $M=0$, a contradiction; hence $\alpha>0$, and we normalize $\alpha=1$. Put $M':=-M\succeq0$. Then $w^{\top}M'w\ge1$ for all $w\in O$, so $M'$ is dual feasible, and $0\le\gamma\le\mu+\tr M=\mu-\tr M'$ gives $\tr M'\le\mu$. With weak duality, $\tr M'=\mu$, which proves the formula and the attainment of the minimum.
\end{proof}

\begin{theorem}[The prover's program at targets of order two]\label{thm:order2}
Let $2t\in\Lat$, $t\notin\Lat$, $\sigma>0$, and $\mu=\mu(\Lat,t)$.
\begin{enumerate}[label=\textup{(\roman*)}]
\item For every $\eta_1\ge0$, $\Phi^{(0)}_\sigma(\Lat,t)=\Phi^{(1)}_\sigma(\Lat,t;\eta)=1-2\min(1,\sigma^2\mu)$.
\item With the exact Hessian constraint ($\eta_2=0$), for every $\sigma>0$,
\[\Phi^{(2)}_\sigma(\Lat,t)=1-2\min\bigl(1,\sigma^2\mu/2\bigr)=\max(-1,1-\sigma^2\mu).\]
If $\sigma^2\mu\ge2$ this infimum is not attained.
\item With $\norm{D_2(P)}\le\eta_2$ in place of $D_2(P)=0$ and spectral level $(1+\delta)\sigma^2$ in place of $\sigma^2$, every feasible $P$ satisfies
\[\E_P\cos(2\pi\langle w,t\rangle)\ \ge\ 1-\frac{n\big((1+\delta)\sigma^2+\eta_2\big)}{\ell^2},\qquad\ell:=\min\{\norm w:w\in O\},\]
whereas without the Hessian constraint the bound is $1-2n(1+\delta)\sigma^2/\ell^2$.
\item \textup{(Distributions approaching the value)} Let $\nu$ be feasible for Definition~\ref{def:capacity} with $\nu(O)\ge\mu-\epsilon$. If $\sigma^2\mu\le1$, the distribution $P=\sigma^2\nu+(1-\sigma^2\nu(O))\delta_0$ is feasible for the program in \textup{(i)} with value $1-2\sigma^2\nu(O)\le1-2\sigma^2\mu+2\sigma^2\epsilon$. For the Hessian program, choose a scalar $a>0$ with $a\sum\nu(w)ww^{\top}\preceq(\sigma^2/2)I$, put $p=a\nu(O)<1$, and choose an integer $K\ge1$ so that $p(1+1/(4K^2))\le1$. Then
\[P=a\nu+\frac{a}{4K^2}\nu_{2K}+\Bigl(1-p-\frac{p}{4K^2}\Bigr)\delta_0,\qquad\nu_{2K}:=\text{the image of $\nu$ under }w\mapsto2Kw,\]
is feasible for \textup{(ii)} and has value $1-2p$. By taking $a$ and $\nu(O)$ appropriately, this approaches the value in \textup{(ii)} for every $\sigma^2\mu$. If $\nu$ attains $\mu$ and $\sigma^2\mu<2$, one may take $a=\sigma^2/2$ and a finite $K$, obtaining an optimum. When $\nu$ and a dual minimizer $M$ both attain their values, complementary slackness places the sampled odd vectors on $\{w\in O:w^{\top}Mw=1\}$.
\end{enumerate}
\end{theorem}

\begin{proof}
(i) Since $\sin(2\pi\langle w,t\rangle)=0$ on $\Lat^*$, $D_1(P)=0$ for every $P$, so the gradient constraint is void and $\Phi^{(1)}=\Phi^{(0)}$. Let $P$ be feasible, i.e. $\Sigma_E+\Sigma_O\preceq\sigma^2I$. Then $\nu:=P|_O/\sigma^2$ satisfies $\sum\nu ww^{\top}=\Sigma_O/\sigma^2\preceq I$, and every finite restriction of $\nu$ is feasible for Definition~\ref{def:capacity}, so $P(O)=\sigma^2\nu(O)\le\sigma^2\mu$; also $P(O)\le1$. Hence $\E_P\chi=1-2P(O)\ge1-2\min(1,\sigma^2\mu)$. The distributions of (iv) show that this bound is approached: for $\sigma^2\mu\le1$, $P=\sigma^2\nu+(1-\sigma^2\nu(O))\delta_0$ has $\Sigma(P)=\sigma^2\sum\nu ww^{\top}\preceq\sigma^2I$ and $\E_P\chi=1-2\sigma^2\nu(O)$; for $\sigma^2\mu>1$ scale $\nu$ down to $\nu'$ with $\sigma^2\nu'(O)=1$ (which preserves $\sum\nu'ww^{\top}\preceq I$) and take $P=\sigma^2\nu'$, with $\E_P\chi=-1$.

(ii) Let $P$ be feasible with $D_2(P)=\Sigma_E-\Sigma_O=0$. Then $\Sigma(P)=2\Sigma_O\preceq\sigma^2I$, so every finite restriction of $\nu:=P|_O/(\sigma^2/2)$ is feasible for Definition~\ref{def:capacity}, and hence $P(O)\le\sigma^2\mu/2$. Also $P(O)\le1$. Therefore
\[\E_P\chi=1-2P(O)\ge1-2\min(1,\sigma^2\mu/2).\]
For the reverse inequality choose a finitely supported feasible $\nu$ with mass arbitrarily near $\mu$. Scale $\nu$ so that $p:=a\nu(O)$ is arbitrarily near $\min(1,\sigma^2\mu/2)$ from below and $a\sum\nu(w)ww^{\top}\preceq(\sigma^2/2)I$. For an integer $K$ large enough that $p(1+1/(4K^2))\le1$, define $P$ as in (iv). Since $w\in O$ implies $2Kw\in E$, its odd and even second moments are both $a\sum\nu(w)ww^{\top}$, while its total second moment is at most $\sigma^2I$. Its objective is $1-2p$, proving the equality on letting $p$ tend to its upper limit. If $\sigma^2\mu\ge2$, equality in the lower bound would require $P(O)=1$. This is impossible: $\Sigma_O\ne0$ for a distribution supported on $O$, whereas $D_2=0$ would require a nonzero even second moment and hence positive mass on $E$. Thus the value $-1$ is not attained in this range.

(iii) If $\Sigma_E+\Sigma_O\preceq(1+\delta)\sigma^2I$ and $\norm{\Sigma_E-\Sigma_O}\le\eta_2$, then $2\Sigma_O=\Sigma(P)+(\Sigma_O-\Sigma_E)\preceq((1+\delta)\sigma^2+\eta_2)I$, so $P(O)\ell^2\le\tr\Sigma_O\le\frac n2((1+\delta)\sigma^2+\eta_2)$ and $\E_P\chi=1-2P(O)\ge1-n((1+\delta)\sigma^2+\eta_2)/\ell^2$. Without the Hessian constraint, $P(O)\ell^2\le\tr\Sigma_O\le\tr\Sigma(P)\le n(1+\delta)\sigma^2$.

(iv) Feasibility and the values were verified in (i) and (ii); the support statement, when the capacity is attained, is the complementary slackness of Lemma~\ref{lem:capacity}.
\end{proof}

At a target of order two, the infimum of the spectral-only program is nonpositive when $\mu\ge\frac12\sigma^{-2}=\pi$ at $\sigma^2=1/(2\pi)$, while the corresponding condition for the exact-Hessian program is $\mu\ge\sigma^{-2}=2\pi$: the Hessian test \emph{doubles the capacity the prover needs}. Acceptance by a concrete verifier also depends on its positive threshold, tolerances, and finite-support sampling slack. The Gaussian prover corresponds to $\nu=(2\pi/\rho(\Lat^*))\,\rho|_O$, whose finite restrictions are feasible by Lemma~\ref{lem:transference}, so $\mu\ge\nu(O)=2\pi\rho(O)/\rho(\Lat^*)=\pi(1-f_\Lat(t))$; when the Hessian $\nabla^2f_\Lat(t)$ vanishes, $\Sigma_O=\Sigma_E$ for the Gaussian, so the same holds for $\nu=(4\pi/\rho(\Lat^*))\rho|_O$ and $\mu\ge2\pi(1-f_\Lat(t))$, consistent with (ii).

Completeness at targets of order two is therefore a transference theorem for cosets: a statement that $\dist(t,\Lat)\ge c\sqrt n$ forces the capacity past the relevant threshold with sufficient slack (respectively near $\pi$ and $2\pi$ for the two programs at Gaussian spectral level). It is an analogue of Banaszczyk's transference theorem, in which a large covering radius forces a short dual vector, with the length of the shortest odd dual vector replaced by the spectral capacity of the odd dual vectors; a single odd vector of norm $\ell$ already proves $\dist(t,\Lat)\ge\dist(\langle w,t\rangle,\Z)/\norm w=1/(2\ell)$, and the capacity is its robust, averaged form. By the computation above and Lemma~\ref{lem:facts}(iv), the Gaussian prover proves this coset transference for every $c>1/\sqrt{2\pi}$, while Proposition~\ref{prop:obstruction} below shows that it fails for $c<1/(2e\sqrt\pi)$, and for $c<1/(e\sqrt{2\pi})$ with the Hessian test.

\subsection{Code lattices: the Gaussian prover is not optimal}\label{sec:codelattices}

We now exhibit lattices on which the Gaussian prover fails at distance greater than $0.2735\sqrt n$ while another prover succeeds; in the language of Definition~\ref{def:linearverifier}, $f_L(x)>\Phi(L,x)$ there, so the Gaussian prover is not the optimal prover for the Aharonov--Regev verifier. For a binary linear code $C\le\mathbb F_2^n$ let
\[\Lambda(C):=\{v\in\Z^n:\ v\bmod2\in C\}\qquad(\text{Construction~A, \cite{CS99}}),\]
a lattice with $2\Z^n\subseteq\Lambda(C)\subseteq\Z^n$, $\det\Lambda(C)=2^{n-\dim C}$ and dual $\Lambda(C)^*=\frac12\Lambda(C^\perp)$. For $s>0$ and a word $z\in\mathbb F_2^n$ put
\[L:=s\Lambda(C),\qquad x:=s\,z\quad(\text{$z$ read as a $0/1$ vector in $\R^n$}).\]
Write $\theta_0:=\rho(2s\Z)$, $\theta_1:=\rho(2s\Z+s)$, and
\begin{equation}\label{eq:ab}
a=a(s):=\frac{\theta_1}{\theta_0},\qquad b=b(s):=\frac{1-a}{1+a}=\frac{\rho(\Z/s+1/2s)}{\rho(\Z/s)},
\end{equation}
where the second expression for $b$ is Poisson summation ($\theta_0+\theta_1=\rho(s\Z)=s^{-1}\rho(\Z/s)$ and $\theta_0-\theta_1=\sum_k(-1)^k\rho(sk)=s^{-1}\rho(\Z/s+1/2s)$). Both lie in $(0,1)$. Let $W_S(z):=\sum_{c\in S}z^{\mathrm{wt}(c)}$ be the weight enumerator of a set $S\subseteq\mathbb F_2^n$, $\mathrm{wt}(S):=\min_{c\in S}\mathrm{wt}(c)$, and $H$ the binary entropy function $H(p)=-p\log_2p-(1-p)\log_2(1-p)$ with inverse $H^{-1}:[0,1]\to[0,\frac12]$ (in this subsection $H$ and $R$ denote the entropy function and a code rate, not the dual variable and the Bernstein range of earlier sections).

\begin{lemma}\label{lem:cosetlattice}
With the notation above:
\begin{enumerate}[label=\textup{(\roman*)}]
\item $\dist(x,L)=s\sqrt{\mathrm{wt}(C+z)}$;
\item $\rho(L)=\theta_0^n\,W_C(a)$ and $\rho(L+x)=\theta_0^n\,W_{C+z}(a)$, so $f_L(x)=W_{C+z}(a)/W_C(a)$;
\item \textup{(MacWilliams)} $\displaystyle f_L(x)=\frac{\sum_{d\in C^\perp}(-1)^{\langle z,d\rangle}\,b^{\mathrm{wt}(d)}}{\sum_{d\in C^\perp}b^{\mathrm{wt}(d)}}\ \ge\ 1-2\big(W_{C^\perp}(b)-1\big)$.
\end{enumerate}
\end{lemma}

\begin{proof}
(i) $L+x=s\{v\in\Z^n:v\bmod2\in C+z\}$, and the shortest $v\in\Z^n$ with $v\bmod2=c$ has entries in $\{0,\pm1\}$ and norm $\sqrt{\mathrm{wt}(c)}$. (ii) Grouping the vectors of $L$ by their residue $c\in C$,
\[\rho(L)=\sum_{c\in C}\prod_{i=1}^n\rho(2s\Z+sc_i)=\sum_{c\in C}\theta_0^{\,n-\mathrm{wt}(c)}\theta_1^{\,\mathrm{wt}(c)}=\theta_0^n\,W_C(a),\]
and likewise for $L+x$ with $c$ replaced by $c+z$. (iii) For $g(v)=a^{\mathrm{wt}(v)}$ on $\mathbb F_2^n$, expanding the indicator of $C$ in characters gives
\[\sum_{c\in C}g(c+u)=\frac1{|C^\perp|}\sum_{d\in C^\perp}(-1)^{\langle u,d\rangle}\hat g(d),\qquad\hat g(d):=\sum_{v\in\mathbb F_2^n}g(v)(-1)^{\langle v,d\rangle}=(1+a)^{n-\mathrm{wt}(d)}(1-a)^{\mathrm{wt}(d)}\]
(the MacWilliams identity, \cite[Chapter~5]{MS77}). Dividing the identity for $u=z$ by the one for $u=0$ gives the formula. The $d=0$ term of the denominator is $1$, every other term is positive, and the numerator differs from the denominator by at most $2\sum_{d\ne0}b^{\mathrm{wt}(d)}=2(W_{C^\perp}(b)-1)$.
\end{proof}

Two facts about codes complete the picture. First, the sphere-covering bound: for every code $C$ of dimension $k$ there is a word $z$ with $\mathrm{wt}(C+z)\ge nH^{-1}(1-k/n)$, because the Hamming balls of radius $r$ about the $2^k$ codewords cover $\mathbb F_2^n$ only if $2^k\cdot2^{nH(r/n)}\ge2^n$. Second, if $C$ is a uniformly random subspace of dimension $k$, then $C^\perp$ is a uniformly random subspace of dimension $n-k$, so
\begin{equation}\label{eq:randomcode}
\E\big[W_{C^\perp}(b)-1\big]=\sum_{d\ne0}b^{\mathrm{wt}(d)}\,\PP(d\in C^\perp)\ \le\ 2^{-k}(1+b)^n .
\end{equation}

\begin{theorem}[Coset lattices]\label{thm:cosets}
Let $s>0$ and $R\in(0,1)$ satisfy $\log_2(1+b(s))<R$. Then for all sufficiently large $n$ there are a lattice $L\subseteq\R^n$ and a point $x$ with
\[\dist(x,L)\ \ge\ s\sqrt{H^{-1}(1-R)}\,\sqrt n\qquad\text{and}\qquad\rho(L+x)\ \ge\ \tfrac12\rho(L).\]
Consequently
\[c_0\ \ge\ c^*:=\sup\big\{\,s\sqrt{H^{-1}(1-R)}\ :\ s>0,\ \log_2(1+b(s))<R<1\,\big\},\]
where $c_0$ is the Gaussian threshold defined above. A one-variable numerical maximization gives $c^*\approx0.2736$ (near $s=0.5038$); this decimal is an approximation, not a certified value of the supremum. The explicit rational choice $s_0=2519/5000$, $R_0=313/2500$ gives the rigorous lower bound $c^*\ge s_0\sqrt{H^{-1}(1-R_0)}>0.2735$.
\end{theorem}

\begin{proof}
Let $k=\lfloor Rn\rfloor$ and let $C$ be a uniformly random $k$-dimensional subspace of $\mathbb F_2^n$. By~\eqref{eq:randomcode}, $\E[W_{C^\perp}(b)-1]\le2\cdot2^{-n(R-\log_2(1+b))}\le\frac18$ for large $n$, so with probability at least $\frac12$, $W_{C^\perp}(b)-1\le\frac14$. Fix such a $C$, let $z$ be a word with $\mathrm{wt}(C+z)\ge nH^{-1}(1-k/n)\ge nH^{-1}(1-R)$ (as $k\le Rn$ and $H^{-1}$ is increasing), and set $L=s\Lambda(C)$, $x=sz$. Lemma~\ref{lem:cosetlattice}(i) gives the distance and (iii) gives $f_L(x)\ge1-2\cdot\frac14=\frac12$. Since this works for every strict admissible pair, $c_0\ge c^*$. For the displayed rational pair, truncate both sums $\theta_0=\sum_je^{-\pi(2s_0j)^2}$ and $\theta_1=\sum_je^{-\pi(2s_0j+s_0)^2}$ at $|j|\le5$; every omitted term has exponent at least $\pi(11s_0)^2$ and consecutive omitted terms on each side decrease by a factor at least $e^{-8\pi s_0^2}$, so all omitted terms together are below $4e^{-\pi(11s_0)^2}/(1-e^{-8\pi s_0^2})<10^{-41}$. Interval arithmetic for the truncated sums, with these tail bounds, then gives $0.09060547<b(s_0)<0.09060549$, $0.1251292<\log_2(1+b(s_0))<0.1251295<R_0$, and $0.2947671<H^{-1}(1-R_0)<0.2947673$ (bisection on the increasing function $H$), whence $s_0\sqrt{H^{-1}(1-R_0)}>0.2735254$. A short script reproducing these enclosures in the interval arithmetic of \texttt{mpmath} accompanies the source files. The approximate value of $c^*$ quoted above comes from maximizing $s^2H^{-1}(1-\log_2(1+b(s)))$ numerically; no claim of a certified global maximum is needed here. (For comparison, Siegel's mean value theorem gives heavy cosets of random unimodular lattices at distance $(1-o(1))\sqrt{n/(2\pi e)}=0.2420\sqrt n$, by a simpler version of the argument of Lemma~\ref{lem:randomorder2} below; the code lattices beat this.)
\end{proof}

For the explicit parameters above, $L=0.5038\,\Lambda(C)$ with $C$ of rate approximately $0.1252$; the lattice contains the cubic lattice $2s\Z^n$ with $2s\approx1$, the coset $L+x$ consists of the $2^{Rn}$ nontrivial residues of $C+z$, all of weight at least $0.2947n$, and these carry at least half of the Gaussian mass of $L$ although the nearest point of the coset is at distance more than $0.2735\sqrt n$. In the dual picture, $L^*=\frac1{2s}\Lambda(C^\perp)$, and $f_L(x)$ is the average of the character $(-1)^{\langle z,d\rangle}$ over $L^*$ with Gaussian weights; the cancellation that would make $f_L(x)$ small never happens because the odd dual vectors (those with $\langle z,d\rangle=1$) carry too little Gaussian mass. A prover that samples odd vectors directly does better, and Proposition~\ref{prop:parity} balances their second moments with doubles to pass the Hessian test.

\begin{proposition}[The parity prover]\label{prop:parity}
Put $c_{\mathrm p}:=s_0\sqrt{H^{-1}(1-R_0)}>0.2735$ for the rational pair $s_0,R_0$ in Theorem~\ref{thm:cosets}. For every $c<c_{\mathrm p}$ and all large $n$ there are $L\subseteq\R^n$ and $x$ with $2x\in L$, $\dist(x,L)\ge c\sqrt n$ and $f_L(x)\ge\frac12$, together with a probability distribution $P$ on $\{w\in L^*:\norm w\le0.77\sqrt n\}$ with
\[\Sigma(P)\preceq0.1582\,I\prec\frac1{2\pi}I,\qquad\E_P\cos(2\pi\langle w,x\rangle)=-0.06,\qquad D_1(P)=0,\qquad D_2(P)=0 .\]
Hence $(L,x)$ is certified by the Hessian-augmented family (by the sampling argument in the proof of Theorem~\ref{thm:constant}), while the Gaussian prover fails on it. In particular nothing we know excludes $\cfar^{(2)}<c_0$.
\end{proposition}

\begin{proof}
Take $s=s_0=2519/5000$ and $R=R_0=313/2500$, let $k=\lfloor Rn\rfloor$, let $C$ be a uniformly random $k$-dimensional code, and let $z$ be a uniformly random word, independent of $C$. Fix $\epsilon\in(0,0.01]$ so small that $s_0\sqrt{(1-\epsilon)H^{-1}(1-R_0)}>c$. Three events hold simultaneously with positive probability for large $n$. (1) $W_{C^\perp}(b)-1\le\frac14$, with probability at least $\frac12$ as before. (2) $\mathrm{wt}(C+z)\ge(1-\epsilon)nH^{-1}(1-k/n)$: the coset $C+z$ is a uniformly random coset, and the expected number of its words of weight below this bound is $2^{-(n-k)}\sum_{u<(1-\epsilon)nH^{-1}(1-k/n)}\binom nu=2^{-\Omega(n)}$. (3) Let $\omega:=\lfloor0.15n\rfloor$, let $O'$ be the set of $d\in C^\perp$ with $\mathrm{wt}(d)=\omega$ and $\langle d,z\rangle=1$, and for each coordinate $i$ let $O'_i=\{d\in O':d_i=1\}$. For a fixed nonzero $d$, $\PP(d\in C^\perp,\ \langle d,z\rangle=1)=\frac12\cdot\frac{2^{n-k}-1}{2^n-1}$, and for distinct nonzero $d,d'$ the two events are negatively correlated, because a uniformly random subspace contains two given independent vectors with probability $\frac{(2^{n-k}-1)(2^{n-k}-2)}{(2^n-1)(2^n-2)}\le\big(\frac{2^{n-k}-1}{2^n-1}\big)^2$ and $z$ is independent of $C$. Hence, since a uniformly random $d$ of weight $\omega$ has $d_i=1$ with probability $\omega/n$,
\[\E|O'_i|=\frac\omega n\,\E|O'|,\qquad\E|O'|=\binom n\omega\frac{2^{n-k}-1}{2(2^n-1)}=2^{n(H(0.15)-R)-O(\log n)}=2^{\Omega(n)},\]
and $\operatorname{Var}|O'_i|\le\E|O'_i|$; by Chebyshev's inequality and a union bound over $i$, with probability $1-2^{-\Omega(n)}$ every $|O'_i|$ lies within a factor $1\pm\epsilon/3$ of $\frac\omega n\E|O'|$, and then $|O'|=\frac1\omega\sum_i|O'_i|$ does too, so every $|O'_i|$ lies within a factor $1\pm\epsilon$ of $\frac\omega n|O'|$.

Fix $(C,z)$ satisfying (1)--(3) and let $L=s\Lambda(C)$, $x=sz$; then $2x=2sz\in2s\Z^n\subseteq L$, and by Lemma~\ref{lem:cosetlattice} and (1), (2), $\dist(x,L)\ge c\sqrt n$ and $f_L(x)\ge\frac12$. For $d\in O'$ and a sign pattern $\varepsilon\in\{\pm1\}^n$ let $w_{d,\varepsilon}:=\frac1{2s}(\varepsilon\circ d)$; since $\varepsilon\circ d\equiv d\pmod2$ this lies in $\frac1{2s}\Lambda(C^\perp)=L^*$, and $\langle w_{d,\varepsilon},x\rangle=\frac12\langle\varepsilon\circ d,z\rangle\equiv\frac12\langle d,z\rangle=\frac12\pmod1$, so these are odd vectors, of norm $\sqrt\omega/2s\le0.3844\sqrt n$. Let $Q$ be uniform on $\{w_{d,\varepsilon}\}$. Averaging over $\varepsilon$ kills the off-diagonal terms, so by (3)
\[\Sigma(Q)=\frac1{4s^2}\diag\Big(\frac{|O'_i|}{|O'|}\Big)_i\preceq\frac{(1+\epsilon)\omega}{4s^2n}I\preceq\frac{0.15(1+\epsilon)}{4\cdot0.5038^2}I\preceq0.14775(1+\epsilon)I .\]
Now let $p=0.53$ and $P:=pQ+\frac p4Q_2+(1-\frac{5p}4)\delta_0$, where $Q_2$ is the image of $Q$ under $w\mapsto2w$; the vectors $2w_{d,\varepsilon}=\frac1s(\varepsilon\circ d)$ lie in $\frac1s\Z^n\subseteq L^*$, are even, and have norm at most $0.7688\sqrt n$. As in Theorem~\ref{thm:order2}(ii), $\Sigma_E(P)=\frac p4\cdot4\Sigma(Q)=p\Sigma(Q)=\Sigma_O(P)$, so $D_2(P)=0$; $D_1(P)=0$ because $\sin(2\pi\langle w,x\rangle)=0$ on $L^*$; $\Sigma(P)=2p\Sigma(Q)\preceq1.06\cdot0.14775\cdot1.01\,I\preceq0.1582\,I$; and $\E_P\chi=1-2P(O)=1-2p=-0.06$.
\end{proof}

\subsection{Obstructions from random lattices}

Lemma~\ref{lem:capacity} also gives obstructions: a target of order two that no certificate can refute is one whose odd dual vectors are all long, since $M=I/\ell^2$ is then dual feasible with $\tr M=n/\ell^2$. Random lattices provide such targets. We use Siegel's mean value theorem~\cite{Siegel45}: if $\Lambda$ is a random unimodular lattice, distributed according to the Haar probability measure on $SL_n(\R)/SL_n(\Z)$, then $\E\sum_{v\in\Lambda\setminus\{0\}}g(v)=\int_{\R^n}g$ for every integrable $g\ge0$; the map $\Lambda\mapsto\Lambda^*$ corresponds to $A\mapsto A^{-\top}$ on $SL_n(\R)$ and preserves Haar measure, so $\Lambda^*$ is Haar-distributed as well. We write $V_n=\pi^{n/2}/\Gamma(\frac n2+1)$ for the volume of the unit ball; Stirling's lower bound $\Gamma(\frac n2+1)\ge\sqrt{2\pi}(\frac n2)^{(n+1)/2}e^{-n/2}$ gives $V_n\le(2\pi e/n)^{n/2}/\sqrt{\pi n}$.

\begin{lemma}[Random lattices with a deep hole of order two]\label{lem:randomorder2}
For every $\epsilon\in(0,\frac12)$ and all sufficiently large $n$ there are a lattice $L\subseteq\R^n$ and a point $t$ with $2t\in L$, $t\notin L$, such that, with $r:=(1-\epsilon)\sqrt{n/2\pi e}$,
\[\dist(t,L)\ge r\qquad\text{and}\qquad\norm w\ge r\ \text{ for every odd vector }w\in L^* .\]
\end{lemma}

\begin{proof}
Let $\Lambda$ be a Haar-random unimodular lattice and, given $\Lambda$, let $u_0$ be a uniformly random nonzero class of $\Lambda^*/2\Lambda^*$. Put $L:=\{v\in\Lambda:\langle v,u_0\rangle\ \text{even}\}$, an index-two sublattice (the parity of $\langle v,u_0\rangle$ is a nonzero homomorphism $\Lambda\to\Z/2$), and $t:=v_0$ for any $v_0\in\Lambda\setminus L$. Then $2t\in2\Lambda\subseteq L$, the coset $t+L$ is $\Lambda\setminus L$, so
\[\dist(t,L)=\min\{\norm v:\ v\in\Lambda,\ \langle v,u_0\rangle\ \text{odd}\},\]
and $L^*=\Lambda^*\cup(\Lambda^*+u_0/2)$, where $u_0/2\in L^*\setminus\Lambda^*$; the odd vectors of $L^*$ are $\Lambda^*+u_0/2$, since $2\langle u+u_0/2,v_0\rangle=2\langle u,v_0\rangle+\langle u_0,v_0\rangle$ is odd. Note $V_nr^n\le(1-\epsilon)^n/\sqrt{\pi n}\le(1-\epsilon)^n$.

\emph{The hole is deep.} Vectors of $2\Lambda$ pair evenly with $u_0$. For a fixed $v\in\Lambda\setminus2\Lambda$ the parity of $\langle v,u_0\rangle$ is a nonzero linear functional of the class $u_0$, so it is odd with probability $2^{n-1}/(2^n-1)\le1$. By Siegel's theorem the expected number of $v\in\Lambda$ with $0<\norm v\le r$ and $\langle v,u_0\rangle$ odd is at most $V_nr^n$.

\emph{The odd dual vectors are long.} We have $\min\norm{\Lambda^*+u_0/2}=\frac12\min\norm{u_0+2\Lambda^*}$. For a uniformly random nonzero class, the expected number of vectors of $u_0+2\Lambda^*$ of norm at most $2r$ is
\[\frac1{2^n-1}\,\E\big|\{y\in\Lambda^*\setminus2\Lambda^*:\norm y\le2r\}\big|=\frac{V_n(2r)^n-V_nr^n}{2^n-1}=V_nr^n,\]
because $\Lambda^*$ is Haar-distributed, so the expected number of nonzero vectors of $\Lambda^*$ of norm at most $2r$ is $V_n(2r)^n$, of which those in $2\Lambda^*$ are the doubles of the nonzero vectors of $\Lambda^*$ of norm at most $r$, expected number $V_nr^n$.

By Markov's inequality both counts are zero with probability at least $1-2(1-\epsilon)^n>0$, and then $\dist(t,L)\ge r$ and every odd vector of $L^*$ has norm at least $r$.
\end{proof}

\begin{proposition}[Obstructions]\label{prop:obstruction}
$\cfar^{(2)}\ge\dfrac1{e\sqrt{2\pi}}=0.1468$ and $\cfar^{(1)}\ge\cfar^{(0)}\ge\dfrac1{2e\sqrt\pi}=0.1038$. Consequently
\[\cAR^{(2)}\ \ge\ \sqrt{\frac8{3\pi}}\cdot\frac1{e\sqrt{2\pi}}\ =\ \frac{2}{e\pi\sqrt3}\ =\ 0.1352,\qquad\cAR^{(1)}\ge\frac1{e\pi}=0.1171,\qquad\cAR^{(0)}\ge\frac{\sqrt2}{e\pi}=0.1656 .\]
\end{proposition}

\begin{proof}
Fix $\epsilon$ and take $L,t,r$ as in Lemma~\ref{lem:randomorder2}. For the first statement rescale: $L'=\lambda L$ with $\lambda=r\sqrt{2\pi}/((1+\epsilon)\sqrt n)$, so that the odd vectors of $L'^*=\lambda^{-1}L^*$ have norm at least $\ell:=r/\lambda=(1+\epsilon)\sqrt{n/2\pi}$, while
\[\dist(\lambda t,L')\ \ge\ \lambda r=\frac{r^2\sqrt{2\pi}}{(1+\epsilon)\sqrt n}=\frac{(1-\epsilon)^2}{1+\epsilon}\cdot\frac{\sqrt n}{e\sqrt{2\pi}} .\]
By Theorem~\ref{thm:order2}(iii) with $\sigma^2=1/2\pi$, every $P$ on $L'^*$ with $\Sigma(P)\preceq(1+\delta_n)\sigma^2I$ and $\norm{D_2(P)}\le\eta_{2,n}$ has
\[\E_P\cos(2\pi\langle w,\lambda t\rangle)\ \ge\ 1-\frac{n\big((1+\delta_n)/2\pi+\eta_{2,n}\big)}{(1+\epsilon)^2n/2\pi}=1-\frac{1+\delta_n+2\pi\eta_{2,n}}{(1+\epsilon)^2}\ >\ \tau_n\]
for all large $n$, whatever the slacks $\tau_n,\delta_n,\eta_{2,n}\to0$. So no $c<\frac{(1-\epsilon)^2}{1+\epsilon}\cdot\frac1{e\sqrt{2\pi}}$ is $2$-far-certifiable; let $\epsilon\to0$. For the second statement rescale instead by $\lambda=r\sqrt\pi/((1+\epsilon)\sqrt n)$, so that $\ell\ge(1+\epsilon)\sqrt{n/\pi}$ and $\dist(\lambda t,L')\ge\frac{(1-\epsilon)^2}{1+\epsilon}\cdot\frac{\sqrt n}{2e\sqrt\pi}$; the bound without the Hessian constraint in Theorem~\ref{thm:order2}(iii) gives $\E_P\cos(2\pi\langle w,\lambda t\rangle)\ge1-\frac{1+\delta_n}{(1+\epsilon)^2}>\tau_n$. The last claims are~\eqref{eq:cAR}.
\end{proof}

\subsection{Where the far side stands}

Collecting Lemma~\ref{lem:gaussianprover}, Theorem~\ref{thm:constant}, Theorem~\ref{thm:cosets} and Proposition~\ref{prop:obstruction}:
\[\begin{gathered}
0.1468=\frac1{e\sqrt{2\pi}}\ \le\ \cfar^{(2)}\ \le\ c_0\ \le\ \frac1{\sqrt{2\pi}}=0.3989,\\
c_0\ge c^*>0.2735\qquad(c^*\approx0.2736\text{ numerically}),
\end{gathered}\]
\[\cAR^{(2)}=\sqrt{\frac8{3\pi}}\,\cfar^{(2)}\ \in\ [\,0.1352,\ 0.3676\,].\]
The upper end is Theorem~\ref{thm:main}. The Gaussian prover alone cannot give a constant below $c^*/\varrho_2>0.2520$ (using $c^*>0.2735254$ from the proof of Theorem~\ref{thm:cosets}), since it fails on the coset lattices of Theorem~\ref{thm:cosets}; going beyond this bound requires a different prover, and the balanced parity prover of Proposition~\ref{prop:parity} is one, on one family of lattices. A better prover must integrate the dual inequality~\eqref{eq:dualfeas} against something other than $D_{\Lat^*}$. In the other direction, raising $\cfar^{(2)}$ means constructing dual solutions at targets farther than $0.1468\sqrt n$; at targets of order two the best one is the trace-minimal ellipsoid excluding the odd dual vectors (Lemma~\ref{lem:capacity} and Theorem~\ref{thm:order2}), and random lattices give $0.1468$.

\section{Concluding remarks}\label{sec:conclusion}

The linear program over distributions on the dual lattice is the natural object behind the Aharonov--Regev verifier: its value at close targets gives soundness, its value at far targets gives completeness, and its dual is an envelope inequality on the close side and a quadratic certificate on the far side. On the close side the theory is complete for tests on the derivatives of the cosine sum at the target, and it yields the constant $2/(\pi\sqrt3)$ of Theorem~\ref{thm:main}; line tests at multiples of the target would lower the constant by about two percent more (Remark~\ref{rem:otherinstances}), numerically and without a closed form.

The far side is where the theory points; Section~\ref{sec:far} collects what is known. We regard the following as the main open questions. Is $\cfar^{(2)}<1/\sqrt{2\pi}$, i.e. is there a prover that beats Banaszczyk's threshold on \emph{every} lattice? The dual formulation reduces this to excluding quadratic certificates of small trace, and the optimal distributions at order-two targets suggest provers that weight short dual vectors by how far $\langle w,t\rangle$ is from an integer rather than by Gaussian mass. What is the value of the Gaussian threshold $c_0\in[c^*,1/\sqrt{2\pi}]$? And is there a polynomial-time checkable certificate for the near-lattice behaviour of $f_W$ beyond second moments and the derivatives at the target, which is what approaching $\sqrt{n/\log n}$ would require?

\appendix
\section{A finite-bit verifier}\label{app:precision}

This appendix implements the Hessian-augmented verifier of Theorem~\ref{thm:main} in the bit model, in dimensions $n\ge8$; smaller dimensions, and more generally those below $n_0(c)$, are decided by Kannan's algorithm. The input $(B,t,r)$ has rational entries encoded in binary, with $B$ nonsingular and $r>0$; let $\ell_{\mathrm{in}}$ be its bit length. All decisions below are exact comparisons of rational numbers or rational symmetric matrices; transcendental quantities are enclosed by rational intervals before any decision is made.

\subsection{Rational scale and computable constants}

Keep the parameters of the proof of Theorem~\ref{thm:main},
\begin{equation}\label{eq:finite-parameters}
\begin{aligned}
\tau_n&=\frac1n,\qquad \delta_n=\frac{0.4}{\sqrt n},\qquad
S_n=\frac{1+\delta_n}{2\pi},\\
E_{1,n}&=\frac1{\sqrt n},\qquad E_{2,n}=\frac{4.5}{\sqrt n},\qquad
N_n=100n^2\lceil\log_2 n\rceil ,
\end{aligned}
\end{equation}
so that $S_n=\sigma^2$, $E_{1,n}=\eta_1$ and $E_{2,n}=\eta_2$ there. The close radius $r_{2,n}$ of~\eqref{eq:rk} for these parameters is
\[r_{2,n}=\frac{-b_n+\sqrt{b_n^2+4A_nC_n}}{2A_n},\qquad A_n=S_n+\frac{E_{2,n}}{2\pi^3},\quad b_n=\frac{E_{1,n}}{4\pi},\quad C_n=\frac3{16}-\frac{5\tau_n}{8\pi},\]
and for $n\ge8$ the numbers $A_n$, $C_n$ and $r_{2,n}$ lie in fixed compact subintervals of $(0,\infty)$. Put $m_n=n^{-4}$. Enclose $\pi$ by Machin's identity
\[\frac\pi4=4\arctan\frac15-\arctan\frac1{239}\]
and the alternating Taylor series of $\arctan$, whose first omitted term bounds the error; enclose $\sqrt n$ and the square root in $r_{2,n}$ by rational bisection; and evaluate the formula for $r_{2,n}$ in rational interval arithmetic. With $O(\log n)$ bits of precision the resulting interval has width less than $m_n$, and we let $q_n$ be its left endpoint, so that
\begin{equation}\label{eq:rationalradius}
q_n\in\mathbb Q,\qquad 0<r_{2,n}-q_n<m_n .
\end{equation}
In the same way choose rationals $S_n^-$, $E_{1,n}^-$, and $E_{2,n}^-$ with
\begin{equation}\label{eq:lowerthresholds}
0\le S_n-S_n^-<m_n/8,\quad
0\le E_{1,n}-E_{1,n}^-<m_n/8,\quad
0\le E_{2,n}-E_{2,n}^-<m_n/8.
\end{equation}
The rational scale is $q_n/r$, so $\Lat'=(q_n/r)B\Z^n$, $t'=(q_n/r)t$, and the certificate vector represented by $z_i$ is
\[w_i=\frac r{q_n}B^{-\top}z_i .\]
In particular the phase is the rational number
\begin{equation}\label{eq:rationalphase}
\langle w_i,t'\rangle=z_i^{\top}B^{-1}t,
\end{equation}
which is reduced modulo $1$ to $\varphi_i\in[0,1)$ by exact rational arithmetic. The matrix $\Sigma_W=N_n^{-1}\sum_iw_iw_i^{\top}$ is rational and is formed exactly. A certificate is a list of exactly $N_n$ integer vectors $z_i\in\Z^n$ with entries of absolute value at most $n^24^{\ell_{\mathrm{in}}}$; the verifier rejects every other string.

\subsection{Guarded acceptance tests}

Write $F_W=f_W(t')$, $G_W=\nabla f_W(t')$, and $H_W=\nabla^2f_W(t')$. Let $M_W$ be the rational number $1+\max_i\norm{w_i}_1\ge1+\max_i\norm{w_i}$ and $\xi=m_n/(2400M_W^2)$. Compute a rational $\tilde\pi$ with $|\tilde\pi-\pi|\le\xi$, and rationals $\tilde c_i$, $\tilde s_i$ within $\xi$ of $\cos2\tilde\pi\varphi_i$ and $\sin2\tilde\pi\varphi_i$, by Taylor polynomials with the remainder bound, and put
\[\widetilde F_W=\frac1{N_n}\sum_i\tilde c_i,\qquad\widetilde G_W=-\frac{2\tilde\pi}{N_n}\sum_i\tilde s_iw_i,\qquad\widetilde H_W=-\frac{4\tilde\pi^2}{N_n}\sum_i\tilde c_i\,w_iw_i^{\top},\]
a rational number, vector and symmetric matrix. Since $|\cos2\tilde\pi\varphi-\cos2\pi\varphi|\le2|\tilde\pi-\pi|$ for $\varphi\in[0,1)$, each $\tilde c_i$ and $\tilde s_i$ is within $3\xi$ of $\cos2\pi\varphi_i$, respectively $\sin2\pi\varphi_i$, and therefore
\[|\widetilde F_W-F_W|\le3\xi,\qquad\norm{\widetilde G_W-G_W}\le21\xi M_W,\qquad\norm{\widetilde H_W-H_W}\le145\xi M_W^2,\]
using $|\tilde\pi^2-\pi^2|\le(2\pi+\xi)\xi$. Hence
\begin{equation}\label{eq:approxerrors}
|\widetilde F_W-F_W|\le m_n/16,\qquad
\norm{\widetilde G_W-G_W}\le m_n/16,\qquad
\norm{\widetilde H_W-H_W}\le m_n/16 .
\end{equation}
The bit length of $\xi^{-1}$ is $O(\log n+\log M_W)$, polynomial in the input and certificate lengths, and so are the costs of the Taylor polynomials and of the rational vector and matrix operations.

The finite-bit verifier accepts precisely when the certificate has the format above and the following four rational tests succeed:
\begin{align}
(S_n^--m_n/2)I-\Sigma_W&\succeq0,\label{eq:guard1}\\
\widetilde F_W&<\tau_n-m_n/2,\label{eq:guard2}\\
\norm{\widetilde G_W}&\le E_{1,n}^--m_n/2,\label{eq:guard3}\\
(E_{2,n}^--m_n/2)I\pm\widetilde H_W&\succeq0.\label{eq:guard4}
\end{align}
The vector test compares squared norms, so it is an exact rational comparison. A rational symmetric matrix $A$ is positive semidefinite if and only if all coefficients of $\det(xI+A)$ are nonnegative: the forward implication follows by expanding $\prod_j(x+\lambda_j)$ with $\lambda_j\ge0$, while a negative eigenvalue would give a positive root, impossible for a polynomial with nonnegative coefficients. The characteristic polynomial is computable in polynomial time over $\mathbb Q$, so~\eqref{eq:guard1} and the two conditions~\eqref{eq:guard4} are three deterministic polynomial-time tests.

\begin{proposition}[Soundness of the finite-bit verifier]\label{prop:bit-sound}
Let $n\ge8$. For every rational input and every certificate, the verifier~\eqref{eq:guard1}--\eqref{eq:guard4} rejects when $\dist(t,B\Z^n)\le r$.
\end{proposition}
\begin{proof}
By~\eqref{eq:lowerthresholds} and~\eqref{eq:approxerrors}, acceptance implies the analytical tests $\Sigma_W\preceq S_nI$, $F_W<\tau_n$, $\norm{G_W}\le E_{1,n}$, and $\norm{H_W}\le E_{2,n}$ of $\mathcal V^{(2)}$ at the rescaled pair $(\Lat',t')$. If $\dist(t,B\Z^n)\le r$, then $\dist(t',\Lat')\le q_n<r_{2,n}$, and Proposition~\ref{prop:sound}, applied with the close radius $r_{2,n}$, gives $F_W\ge\tau_n$, a contradiction.
\end{proof}

\begin{proposition}[Completeness and certificate size]\label{prop:bit-complete}
Let $c>2/(\pi\sqrt3)$ and $n\ge n_0(c)$, the cutoff of Table~\ref{tab:n0}. Every rational input with $\dist(t,B\Z^n)>c\sqrt n\,r$ has a certificate accepted by~\eqref{eq:guard1}--\eqref{eq:guard4}, of bit length polynomial in $n$ and $\ell_{\mathrm{in}}$.
\end{proposition}
\begin{proof}
The verifier rescales by $q_n/r$ with $q_n\in(r_{2,n}-n^{-4},r_{2,n})$, and $\Sigma_W$, $F_W$, $G_W$, $H_W$ are the quantities tested by $\mathcal V^{(2)}$ for this rescaling. By the proof of Theorem~\ref{thm:main}, there is a list $W$ of $N_n$ vectors of $\Lat'^*$ with $\norm{w_i}\le\sqrt n$ (the event $\mathcal E$) such that
\[\Sigma_W\preceq(S_n-m_n)I,\qquad F_W<\tau_n-m_n,\qquad\norm{G_W}\le E_{1,n}-m_n,\qquad\norm{H_W}\le E_{2,n}-m_n .\]
By~\eqref{eq:lowerthresholds} and~\eqref{eq:approxerrors}, $S_n^--m_n/2>S_n-m_n$, $\widetilde F_W<\tau_n-\frac{15}{16}m_n<\tau_n-m_n/2$, $\norm{\widetilde G_W}\le E_{1,n}-\frac{15}{16}m_n<E_{1,n}^--m_n/2$ and $\norm{\widetilde H_W}\le E_{2,n}-\frac{15}{16}m_n<E_{2,n}^--m_n/2$, so all four tests succeed. The integer vectors $z_i=(q_n/r)B^{\top}w_i$ satisfy $\norm{z_i}\le(q_n/r)\norm B\sqrt n$, where $q_n<2$, $1/r\le2^{\ell_{\mathrm{in}}}$ and $\norm B\le n2^{\ell_{\mathrm{in}}}$ because every input entry and the inverse of every nonzero input entry have absolute value at most $2^{\ell_{\mathrm{in}}}$; so their entries are at most $2n^{3/2}4^{\ell_{\mathrm{in}}}\le n^24^{\ell_{\mathrm{in}}}$ in absolute value, the certificate has the required format, and its length is $O(N_nn(\ell_{\mathrm{in}}+\log n))$ bits.
\end{proof}

Theorem~\ref{thm:main} follows from Propositions~\ref{prop:bit-sound} and~\ref{prop:bit-complete}, the $\NP$ witnesses for YES instances, Kannan's algorithm in dimensions below $n_0(c)$, and the $\mathrm{GapSVP}$ transfer in the proof of Theorem~\ref{thm:main}.

\end{document}